\documentclass[11pt]{article}

\usepackage{authblk}
\usepackage[utf8]{inputenc}
\usepackage[hidelinks]{hyperref}
\usepackage[italicdiff]{physics}
\usepackage{graphicx}
\usepackage{subcaption}
\usepackage{float}
\usepackage{multirow}
\usepackage{lscape}
\usepackage{amsmath}
\usepackage{amsthm}
\usepackage{amssymb}
\usepackage{algorithm}
\usepackage{algpseudocode}
\usepackage{fancyhdr}
\usepackage{booktabs}
\usepackage{geometry}
\usepackage{amsthm}  
\usepackage{subcaption}
\usepackage{ragged2e}
\usepackage[american]{babel}
\usepackage{setspace}
\usepackage{bm}
\usepackage[dvipsnames,table]{xcolor}
\usepackage{tcolorbox}
\usepackage{natbib}  
\usepackage{mathtools}
\usepackage{listings}
\usepackage{appendix}
\usepackage{lipsum}
\usepackage{alltt}
\usepackage{accents}
\usepackage[font=small,labelfont=bf]{caption}
\usepackage{bigstrut}
\usepackage[export]{adjustbox}
\usepackage[table]{xcolor}
\definecolor{lightgray}{gray}{0.92}
\definecolor{lightviolet}{RGB}{235,230,250}
\usepackage{makecell}

\AtBeginDocument{\RenewCommandCopy\qty\SI}

\hypersetup{colorlinks=true, allcolors=[rgb]{0,0,1}}
\usepackage{tikz}
\usetikzlibrary{shapes, arrows, fit}
\usetikzlibrary{positioning}
\usetikzlibrary{bayesnet}
\usetikzlibrary{positioning, fit, backgrounds, shapes.misc, calc}

\tikzset{
  obs/.style={circle, draw=black, fill=white, inner sep=1pt, minimum size=8mm},
  param/.style={circle, draw=black, fill=gray!20, inner sep=1pt, minimum size=8mm},
  hyper/.style={rectangle, draw=black, fill=blue!20, inner sep=2pt, minimum size=6mm},
}

\usepackage{wallpaper}
\usepackage{xcolor}
\usepackage{setspace}
\usepackage{microtype}

\usepackage{nomencl}
\makenomenclature

\renewcommand{\nompreamble}{The next list describes several symbols that will be later used within the body of the document}

\newcommand\mc{\mathcal}
\newcommand\mbb{\mathbb}

\newtheorem{theorem}{Theorem}
\newtheorem{lemma}[theorem]{Lemma}
\newtheorem{definition}[theorem]{Definition}

\title{\vspace{-2em}Scalable Bayesian inference for high-dimensional mixed-type multivariate spatial data}
\author{Arghya Mukherjee}
\author{Arnab Hazra}
\author{Dootika Vats}
\affil{Indian Institute of Technology Kanpur}
\date{}
\begin{document}
\maketitle
\begin{abstract}
\noindent Spatial generalized linear mixed-effects models are popularly used to analyze spatially indexed univariate responses. However, with modern technology, it is common to observe vector-valued mixed-type responses, e.g., a combination of binary, count, or continuous types, at each location. Methods for jointly modeling such mixed-type multivariate spatial responses are rare. Using multivariate Gaussian processes (GPs) in the latent layer, we present a class of Bayesian spatial methods applicable to any combination of exponential family responses. Since multivariate GP-based methods can suffer from computational bottlenecks when the number of spatial locations is high, we further employ a computationally efficient Vecchia approximation for fast posterior inference and prediction. Key theoretical properties of the proposed model, such as identifiability and the structure of the induced covariance, are established. Our approach employs a Markov chain Monte Carlo-based inference method that uses elliptical slice sampling within a blocked Metropolis-within-Gibbs sampling framework. We illustrate the efficacy of the proposed method through simulation studies and a real-data application on joint modeling of wildfire counts and burnt areas across the United States.
\end{abstract}

\setstretch{1.15}
\begin{keywords}
\noindent Large spatial data analysis, Latent Gaussian models, Mixed-type spatial responses, Multivariate spatial Bayesian modeling, Markov chain Monte Carlo, Vecchia approximation.
\end{keywords}



\section{Introduction}

In recent years, the scale and ubiquity of vast, spatially indexed datasets have grown significantly. This expansion is largely fueled by technological advancements, such as the Global Positioning System and Remote Sensing, as well as the relentless growth in digital storage capacity. Spatially indexed responses can be multivariate and, more specifically, ``mixed-type'', i.e., a combination of continuous, skewed continuous, binary, count, and other types. In this work, we introduce a simple and interpretable Bayesian model for spatial point-referenced multivariate data that accommodates mixed-type responses from the exponential family.

Flexible statistical modeling and inference for mixed-type spatial data is important due to two types of dependence: (a) spatial dependence across locations, and (b) dependence across different response variables. Understanding the relationship between different types of spatial responses is crucial across numerous scientific domains, including environmental science, epidemiology, and risk assessment. For instance, to improve predictive accuracy in wildfire modeling, \cite{cisneros2023combined} jointly analyze count-type responses, such as wildfire frequencies, and continuous-type responses, such as total burnt area. Similarly, \cite{yadav2023joint} develop Bayesian hierarchical marked point-process models for the joint prediction of landslide counts and the size of affected areas to quantify regional landslide hazard better. More broadly, climate scientists may be interested in studying whether incorporating both types of information enhances inference in assessing environmental hazards. Similarly, in public health research, the National Academy of Sciences of the United States has emphasized the importance of monitoring daily fine particulate matter ($\textrm{PM}_{2.5}$) due to its strong association with adverse health outcomes \citep{Burnett18}. Elevated $\textrm{PM}_{2.5}$ levels have been linked to increased mortality risks, including lung cancer in individuals with no prior history of smoking \citep{Turner2011} and heightened cardiovascular mortality \citep{Brook2010}. The interdependence among different responses suggests that jointly modeling $\textrm{PM}_{2.5}$ concentrations (continuous) and disease incidence or mortality counts (discrete) could yield significant insights into their relationships. Our methodological framework aims to address such challenges by providing a flexible and scalable approach for high-dimensional mixed-type spatial data.

Most existing methods for mixed-type response models have been developed for specific applications and do not apply to generic spatial models. For example, \cite{Gueorguieva01StatM} propose joint models for dependent discrete and continuous outcomes in biomedical studies, and \cite{Goldstein09} extend multilevel models to deal with incomplete and structured data. Generalized linear latent variable models offer a more flexible framework for modeling mixed-type responses \citep{Sammel02}, including extensions to longitudinal settings \citep{Yang2014}. More recent approaches also account for zero inflation and over-dispersion in continuous and count responses \citep{Kassahun15, Molenberghs10}. Their methods are designed for specific paired combinations of response types, such as continuous and count, or continuous and binary. They are not flexible enough to be extended to spatially correlated data with various possible types of multi-response data.

We develop a Bayesian hierarchical framework for modeling mixed-type spatial responses. The model flexibly accommodates data comprising different response types by employing a unified latent process structure, enabling joint inference across varied measurement scales. To address the computational challenges inherent in analyzing large spatial datasets, we use the Vecchia approximation \citep{Vecchia88} in our model that significantly reduces the cost of matrix operations typically associated with GPs, making posterior inference via Markov chain Monte Carlo (MCMC) amenable even in high-dimensional spatial settings. The proposed approach is particularly well-suited for applications involving a large number of spatial locations and a moderate number of response variables, which is a common scenario.

The structure of the paper is as follows. Section~\ref{sec:methodology} provides a concise review of spatial multivariate models. Building on this foundation, we highlight key challenges and introduce the core components of our proposed mixed-type multivariate spatial model. We then extend this model to high-dimensional settings by incorporating the Vecchia approximation for improved scalability. In Section~\ref{sec:bayesian_computation}, we present an MCMC algorithm for scalable Bayesian inference. Section~\ref{sec:predictive_modeling} focuses on fast predictive inference, while Section~\ref{sec:data_analysis} explores the performance of our approach through both simulated and real data analysis. Finally, in Section~\ref{sec:discussions}, we summarize our contributions and outline promising directions for future research. Additional discussions and numerical experiments are provided in the Supplementary Materials, with sections therein denoted by the prefix ``S''. Equations referred from the Supplementary Materials are denoted by ``E\S''.

\section{Methodology}\label{sec:methodology}

We present the ingredients for a multivariate spatial model that accommodates mixed-type outcomes, briefly review the existing approaches to model multivariate point-referenced spatial data, and subsequently develop our proposed model. 

\subsection{Background}

Suppose for each spatial site $\bm{s}$ within a compact domain $\mc{D} \subset \mbb{R}^d$, has multivariate response $\mc{Y}(\bm{s}) = \big[Y_1(\bm{s}), Y_2(\bm{s}),\ldots, Y_q(\bm{s})\big]^{\top}$ on $q$ variables of interest, along with information on spatially varying vector of covariates denoted by $\mc{X}(\bm{s}) \in \mbb{R}^p$. Classical multivariate spatial models \citep{Zhang21mult_geo} specify all components of $\mathcal{Y}(\bm{s})$ to be continuous and equipped with a latent multivariate zero-mean spatially colored process  $\mc{W}(\bm{s}) = \big[{W}_1(\bm{s}), {W}_2(\bm{s}), \ldots, {W}_q(\bm{s})\big]^{\top}$ such that for different $j \in \{1, \cdots, q\}$, $W_j(\bm{s})$ stitches the cross-dependence among the $q$ outcomes. Each component of $\mc{Y}(\bm{s})$ is thus modeled using a spatial regression model given by
\begin{equation}\label{splm}
    Y_j(\bm{s}) = \mc{X}(\bm{s})^{\top}\bm{\beta}_j + W_j(\bm{s}) + \varepsilon_j(\bm{s}),\; j = 1,2, \cdots ,q,
\end{equation}
where  the $p$-dimensional regression coefficient corresponding to the covariate $\mc{X}(\bm{s})$ for the $j$-th response is denoted by $\bm{\beta}_j$ and $\varepsilon_j(\bm{s}) \overset{\text{iid}}{\sim} \mc{N}(0, \tau^2)$ captures measurement error or nugget effect independently of $\bm{s}$. The dependence structure of the process $\mc{W}(\cdot)$ is specified by a $q$-dimensional matrix-valued covariance kernel, denoted as $\bm{C}(\cdot, \cdot)$. At locations $\bm{s}$ and $\bm{s}'$ in $ \mathcal{D}$, spatial covariance between the $i$-th and $j$-th components of the process can be specified by the $(i,j)$-th element of the matrix-valued covariance function as
\begin{equation}\label{cov_kernel}
    C_{ij}(\bm{s}, \bm{s}') = \textrm{Cov}\left[W_i(\bm{s}), W_j(\bm{s}') \right] \;  \textrm{for all}\; i,j = 1, \ldots, q.
\end{equation}
\cite{Genton15} review diverse choices of covariance functions of the form compatible with \eqref{cov_kernel}, which can yield a flexible class of models discussed in \eqref{splm}. 
Among these, multivariate Mat\'{e}rn constructions provide a particularly interpretable and widely used parametric class with controllable smoothness and cross-dependence parameters \citep{gneiting2010matern, porcu2016spatio}.
However, such models in \eqref{splm} assume Gaussian marginals for each $Y_j(\bm{s})$ and do not naturally extend to mixed-type outcomes.

A major challenge of Bayesian spatial generalized linear mixed-effects models (spGLMMs) in high dimensions lies in posterior sampling, due to the additional complexity of intractable likelihoods \citep{Christensen06}. One of the earliest works in this direction is by \cite{Zhu05}, who proposed a frequentist generalized latent variable model for replicated multivariate spatiotemporal data. More recently, conjugate Bayesian models have been introduced for multivariate responses from exponential family distributions using Diaconis–Ylvisaker conjugate priors, which enable MCMC-free inference and thereby offer computational advantages \citep{Bradley2024, Nandy2022, Zhou2024}. However, these methods often impose simplified or fixed spatial dependence structures and may not scale well with increasing numbers of response types or spatial locations.

Some studies use copula-based models to capture dependence in mixed-type responses \citep{deLeon11, Fitzmaurice95, Song09}. While useful, copulas do not uniquely specify the marginal dependence with discrete outcomes, and computation is challenging in high-dimensional spatial models \citep{hazra2021estimating, hazra2025efficient}. In contrast, hierarchical models offer a more tractable way to model marginal distributions with an explainable dependence structure jointly. In the non-spatial setting, mixed-type response models have been studied using latent factor models \citep{Jiryaie16}, or multivariate generalized linear models \citep{Ekvall22} that link each response to its own linear predictor. While these methods are effective for small-scale and non-spatial problems, they often struggle with computational scalability in high dimensions. The sparse scientific literature in the spatial statistics community has led us to develop a simple yet useful Bayesian hierarchical model that jointly handles mixed-type responses and spatial dependence within a coherent latent variable framework.

\subsection{Proposed joint model}

For any arbitrary spatial location $\bm{s} \in \mc{D}$, let the $q$-variate mixed-type response vector be $\mathcal{Y}(\bm{s})~=~\big[Y_1(\bm{s}), \ldots, Y_q(\bm{s})\big]^{\top}$ with a latent multivariate process $\mathcal{W}(\bm{s})~=~\big[W_1(\bm{s}), \ldots, W_q(\bm{s})\big]^{\top}$. For each $j = 1, \ldots, q$ and $\bm{s} \in \mc{D}$, we model $Y_j(\bm{s})$ as conditionally independent given $W_j(\bm{s})$. We specify a known canonical link function $g_j$ for each $j$-th response type, according to \cite{diggle1998model} and define the conditional mean of $Y_j(\bm{s})$ through
\begin{equation}\label{link_functions}
    g_j\left(\textrm{E}[Y_j(\bm{s}) \mid W_j(\bm{s})]\right) = W_j(\bm{s}).
\end{equation}
Let the cumulant function be denoted by $b_j(\cdot)$ and the known dispersion parameter associated with the $j$-th response be $\psi_j$. We assume $Y_j(\bm{s})$ follows an exponential family distribution, denoted by $\textrm{EF} \left(Y_j(\bm{s}) \mid W_j(\bm{s}) = w, \psi_j \right)$, with a Lebesgue or count measure dominated density, evaluated at $y$ as 
\begin{equation}\label{spatial_exp_fam}
    f_j\left(y \mid w, \psi_j\right) 
    = h\left(y, \psi_j\right) 
    \exp\left\{ \dfrac{y w - b_j(w)}{\psi_j}  \right\}.
\end{equation}
We adopt a centered parameterization proposed by \cite{Christensen06} for the latent multivariate process $\mathcal{W}(\cdot) = \left[W_1(\cdot), \ldots, W_q(\cdot)\right]^{\top}$ and assume
\begin{equation}\label{centered_gp}
    \mathcal{W}(\cdot) \sim \mathcal{GP}_q\left(\bm{B}^\top \mathcal{X}(\cdot),\; \bm{C}(\cdot, \cdot)\right),
\end{equation}
where $\bm{B} = (\bm{\beta}_1, \ldots, \bm{\beta}_q) \in \mathbb{R}^{p \times q}$ denotes the regression coefficient matrix with each element $(\bm{\beta}_{ij})_{1 \leq i \leq p, 1 \leq j \leq q}$ being the fixed effect corresponding to the $i$-th covariate and $j$-th response, $\mathcal{GP}_q$ denotes a $q$-variate Gaussian process (GP), and $\bm{C}(\cdot, \cdot)$ is the multivariate covariance kernel. To enable efficient computation and parsimonious inference on the covariance function, we impose a separable structure
\begin{equation}\label{sep_kernel}
    \bm{C}(\bm{s}, \bm{s}') = \mathcal{K}(\bm{s}, \bm{s}') \bm{\Sigma}, \quad \bm{s}, \bm{s}' \in \mathcal{D},
\end{equation}
where $\bm{\Sigma} = (\bm{\Sigma}_{ij})_{1 \leq i, j \leq q}$ is a $q \times q$-dimensional spatially-invariant matrix that captures cross-response covariance, and $\mathcal{K}(\bm{s}, \bm{s}')$ is a valid univariate spatial correlation function. We use the commonly used Mat\'ern kernel \citep{Matern60} as a flexible choice for $\mathcal{K}$ given by
\begin{equation}\label{eq:matern_kernel}
    \mathcal{K}(\bm{s}, \bm{s}') = \dfrac{1}{2^{\nu-1} \Gamma(\nu)} 
    \left( \dfrac{\norm{\bm{s} - \bm{s}'}_2 }{\phi} \right)^{\nu} 
    K_{\nu}\left( \dfrac{\norm{\bm{s} - \bm{s}'}_2}{\phi} \right), 
    \quad \phi > 0,\; \nu > 0,
\end{equation}
where $K_{\nu}$ is the modified Bessel function of the second kind. Here, the parameter $\phi$ controls the range of spatial dependence, and the smoothness parameter $\nu$, which we assume to be fixed, determines the smoothness of the process. Hereafter, we denote the Mat\'ern kernel $\mc{K}$ in \eqref{eq:matern_kernel} as $\mc{K}_{\phi}$ with parameter $\phi$. 
In a mixed-type response model, if the observed data do not admit a direct transformation to a Gaussian scale, variogram-based methods are not applicable for learning the smoothness parameter, $\nu$. Moreover, $\nu$ is weakly identifiable, particularly under non-Gaussian likelihoods, and cannot be estimated reliably from the data \citep{Stein1999}. Consequently, a common effective strategy is to fix $\nu$ at a prespecified value, selecting from a range of plausible smoothness levels based on the application. Typical choices include $\nu \in \{0.5,\ldots, 2.5\}$, where $\nu = 0.5$ corresponds to the exponential covariance function yielding rough spatial surfaces. Larger values of $\nu$ imply a smoother process. As $\nu \to \infty$, $\mc{W}(\cdot)$ approaches to a $q$-variate GP with a squared-exponential covariance kernel, suitable for modeling an oversmoothed latent process. In practice, we recommend selecting $\nu$ based on prior knowledge of process smoothness, or by conducting sensitivity analyses over a small set of candidate values, as inference is often robust to such choices relative to other covariance parameters \citep{geoga2023fitting}.

Before introducing our joint model, we define the parameters and describe prior specifications used in our proposed Bayesian framework. Conditioning on the $q \times q$ response cross-covariance matrix $\bm{\Sigma}$, we specify a Matrix-Normal prior on the $p \times q$-dimensional regression coefficient matrix $\bm{B}$, denoted by $\mc{MN}_{p,q}$, given by $ \bm{B} \mid \bm{\Sigma} \sim \mathcal{MN}_{p, q}( \bm{M}, \bm{V},  \bm{\Sigma})$. Here $\bm{M}$ is a $p \times q$ mean matrix and $\bm{V}$ is a $p \times p$ dimensional positive definite matrix expressing the row-wise covariance matrix of $\bm{B}$. The prior distribution of $\bm{B}$ depends on the data-level covariance matrix $\bm{\Sigma}$. While this prior covariance may differ from $\bm{\Sigma}$, it would significantly increase the computational burden \citep{hazra2020multivariate}. For the cross-response covariance matrix $\bm{\Sigma}$, we adopt a commonly used Inverse-Wishart prior $ \bm{\Sigma} \sim \mathcal{IW}(\bm{S}, v)$, with a positive definite scale matrix $\bm{S}$ and degrees of freedom $v$. The spatial correlation is governed by a Mat\'ern kernel in \eqref{eq:matern_kernel} with $\phi$ and $\nu$. Since consistent estimation of the Mat\'ern parameters is challenging \citep{Zhang2004inconsistent}, we fix $\nu$ at a reasonable value and assign a uniform prior for $\phi$ as $ \phi \sim \mathcal{U}(0, b_{\phi})$, where $b_\phi$ is chosen so that the effective range corresponds to a minimal correlation (say, $0.01$ or $0.05$) at the domain diameter $\Delta: = \underset{{k,l}}{\max}\lVert \bm{s}_k - \bm{s}_l \rVert_2$, ensuring sufficient posterior learning through data. We will discuss the reasoning for the prior of $\phi$ in Section~\ref{sec:bayesian_computation}. We further assume that the conditional cumulant functions $b_j(\cdot)$ for one-parameter regular exponential families are strictly convex for the $j$-th response.
Consequently, our proposed hierarchical joint model is
\begin{equation}\label{our_model}
\begin{aligned}
\text{Data level:} \quad
& Y_j(\bm{s}) \mid W_j(\bm{s}) \;\overset{\mathrm{ind}}{\sim}\; \mathrm{EF}\big(Y_j(\bm{s}) \mid W_j(\bm{s}), \psi_j \big), 
\quad j = 1, \ldots, q,\quad \bm{s} \in \mc{D},\\[0.3em]
\text{Process level:} \quad
& \mc{W}(\cdot) \mid \bm{B}, \bm{\Sigma}, \phi \;\sim\; \mc{GP}_q\big(\bm{B}^\top \mc{X}(\cdot),\; \mc{K}_\phi(\cdot, \cdot)\, \bm{\Sigma} \big), \\[0.3em]
\text{Parameter level:} \quad
& \bm{B} \mid \bm{\Sigma} \;\sim\; \mathcal{MN}_{p, q}\left(\bm{M}, \bm{V}, \bm{\Sigma}\right), \\
& \bm{\Sigma} \;\sim\; \mathcal{IW}(\bm{S}, v), \\
& \phi \;\sim\; \mathcal{U}(0, b_{\phi}).
\end{aligned}
\end{equation}

We observe the data at $n$ spatial locations $\mathcal{S} = \{\bm{s}_1, \ldots, \bm{s}_n\} \subset \mathcal{D}$, where each site leads to a $q$-variate response $\mathcal{Y}(\bm{s}_i) \in \mbb{R}^q$ and covariates $\mathcal{X}(\bm{s}_i) \in \mathbb{R}^p$. The stacked observed response matrix, the latent spatial random effects matrix, and the covariate matrix, respectively, are
$$
\bm{Y} = 
\begin{bmatrix}
\mc{Y}(\bm{s}_1)^\top \\
\mc{Y}(\bm{s}_2)^\top \\
\vdots \\
\mc{Y}(\bm{s}_n)^\top
\end{bmatrix}_{n \times q},
\;\;
\bm{W} =
\begin{bmatrix}
\mc{W}(\bm{s}_1)^\top \\
\mc{W}(\bm{s}_2)^\top \\
\vdots \\
\mc{W}(\bm{s}_n)^\top
\end{bmatrix}_{n \times q},
\;\;
\bm{X} =
\begin{bmatrix}
\mc{X}(\bm{s}_1)^\top \\
\mc{X}(\bm{s}_2)^\top \\
\vdots \\
\mc{X}(\bm{s}_n)^\top
\end{bmatrix}_{n \times p}
.$$
 Under the separable covariance $\bm{C}(\bm{s}, \bm{s}') = \mathcal{K}_{\phi}(\bm{s}, \bm{s}') \bm{\Sigma}$, the induced covariance of the vectorized latent process, denoted as $\operatorname{vec}(\bm{W})$ across all spatial sites is 
$\bm{\Omega}_{\phi}: = \bm{\Sigma} \otimes \bm{K}_{\phi}$, 
where $\bm{K}_{\phi} = \big(\mathcal{K}_\phi(\bm{s}_i, \bm{s}_j)\big)_{1 \leq i,j \leq n}$. The $\operatorname{vec}(\cdot)$ operator transforms a matrix into a column vector by vertically stacking the columns of the matrix. The latent matrix $\bm{W}$ thus follows a Matrix-Normal distribution specified as
\begin{equation} \label{W_likelihood}
    \bm{W} | \bm{B}, \bm{\Sigma}, \phi \sim \mathcal{MN}_{n,q}(\bm{X} \bm{B}, \bm{K}_\phi, \bm{\Sigma})
    \Leftrightarrow 
    \operatorname{vec}(\bm{W}) | \bm{B}, \bm{\Sigma}, \phi  \sim \mathcal{N}_{nq}( (\bm{I}_q \otimes \bm{X}) \operatorname{vec}(\bm{B}), \bm{\Omega}_{\phi}).
\end{equation}
We assume that the underlying data over $\mc{S}$ is generated from our proposed model in \eqref{our_model}. To be specific, each column of $\bm{Y}$ is conditionally independent of the corresponding columns of $\bm{W}$. We demonstrate our hierarchical model using a directed acyclic graph in Figure~\ref{fig:dag}.
\begin{figure}[]
\centering
\begin{tikzpicture}[node distance=8mm and 8mm, >=stealth]

\node[param] (Wj) {$W_j(\bm{s}_i)$};
\node[obs, right=of Wj] (Yj) {$Y_j(\bm{s}_i)$};

\node[param, above left=5mm and 5mm of Wj] (B) {$\bm{B}$};
\node[param, below left=5mm and 5mm of Wj] (Sigma) {$\bm{\Sigma}$};

\node[param, below=of Wj] (phi) {$\phi$};
\node[hyper, below=of phi] (bphi) {$b_{\phi}$};


\node[hyper, left=of B] (M) {$\bm{M}$};
\node[hyper, below=of M] (V) {$\bm{V}$};

\node[hyper, left=of Sigma] (S) {$\bm{S}$};
\node[hyper, below=of S] (v) {$v$};

\draw[->] (M) -- (B);
\draw[->] (V) -- (B);
\draw[->] (Sigma) -- (B);

\draw[->] (S) -- (Sigma);
\draw[->] (v) -- (Sigma);

\draw[->] (bphi) -- (phi);

\draw[->] (B) -- (Wj);
\draw[->] (Sigma) -- (Wj);
\draw[->] (phi) -- (Wj);

\draw[->] (Wj) -- (Yj);

\node[
  draw,
  fit=(Wj)(Yj),
  inner sep=5pt,
  label=below right:{\normalsize
    \shortstack{$i = 1,\ldots,n$ \\ $j = 1,\ldots,q$}}
] (plate1) {};

\end{tikzpicture}
\caption{Directed acyclic graph representation of our model. Gray shaded nodes denote the model parameters and violet nodes denote the fixed hyperparameters.}
\label{fig:dag}
\end{figure}

\subsection{Model properties}\label{sec:model_properties}

We discuss properties of our model that are essential for reliable posterior inference. Although the simple structure of our proposed model facilitates interpretability, it also raises concerns about potential misspecification and its impact on inference. In the context of spatial mixed-type response models, this typically involves selecting a parameterized covariance structure for the random effects. However, even with a well-specified, non-overparameterized, and separable covariance matrix for the responses, models for data with non-replicated observations may still suffer from non-identifiability \citep{Zhang2004inconsistent}. We discuss some model properties using the marginal mean and covariance of $\textrm{vec}(\bm{Y})$. Denoting $\bm{g} = (g_1, g_2, \cdots, g_q)^{\top}$ to be the vectorized link function corresponding to $q$-types of responses, we write  
\begin{equation*}\label{cond_exp_var}
    \textrm{E}[\textrm{vec}({\bm{Y}})] = \bm{g}^{-1}(\textrm{vec}{(\bm{W}})), 
\quad 
\mathrm{cov}[\textrm{vec}(\bm{Y})] = \mathrm{cov}[\bm{g}^{-1}(\textrm{vec}(\bm{W}))] + \textrm{E}[\mathrm{cov}[\textrm{vec}(\bm{Y}) \mid \textrm{vec}(\bm{W})]].
\end{equation*}
From the above equation, several observations follow. For any $k \in \{1,2\ldots,q\}$, we define $b_{k}^{''}(\cdot)$ to be the second derivative of the cumulant function. We can write the conditional covariance as
$$ \mathrm{cov}[\textrm{vec}(\bm{Y}) \mid \textrm{vec}(\bm{W})] = \textrm{block-diag}(\bm{A}_1, \cdots, \bm{A}_q), \; \bm{A}_k = \psi_{k} \;\textrm{diag}\left(\textrm{E}\big[b_{k}^{''}(W_k(\bm{s}_i))\big]\right)_{1 \leq i \leq n}, $$ which is block-diagonal and the cross-dependence across the $q$ responses is determined by $\mathrm{cov}[\bm{g}^{-1}(\textrm{vec}(\bm{W}))]$. Due to having different dispersion parameters $\psi_k$ and cumulant functions $b_k$, the block-diagonal matrices $\bm{A}_k$ are thus distinct for $q$-different responses. Hence, $\mathrm{cov}[\textrm{vec}(\bm{Y}) \mid \textrm{vec}(\bm{W})]$ exhibits a non-separable covariance structure. This structure implies that, even though a separable covariance is assumed for the latent matrix $\bm{W}$, the marginal covariance of $\bm{Y}$ can still be non-separable and spatially varying. As a result, our model retains the flexibility to capture heterogeneous covariance both across response types and over space. Second, for any spatial location $\bm{s}$, the univariate distribution of $Y_j$ fully determines both $\textrm{E}[Y_j(\bm{s})]$ and $\textrm{E}[Y_j^2(\bm{s})]$. Consequently, it directly follows that the off-diagonal elements of $\bm{\Sigma}$ do not influence the marginal means or variances. Their contribution is restricted to the dependence structure across locations. Therefore, the marginal variability coincides with that obtained from a model in which $\bm{\Sigma}_{ij} = 0$ for $i \neq j$. Third, since $\bm{g}$ and $\mathrm{cov}[\textrm{vec}(\bm{Y}) \mid \textrm{vec}(\bm{W})]$ are generally nonlinear and non-constant (for example $\bm{g}$ can be a vectorized function of logit, log-link, etc.), both $\textrm{E}[\textrm{vec}(\bm{Y})]$ and $\textrm{E}[\mathrm{cov}[\textrm{vec}(\bm{Y}) \mid \textrm{vec}(\bm{W})]]$ depend on $\bm{\beta}$ and the diagonal elements of $\bm{\Sigma}$. Fourth, because $\mathrm{var}[Y_j(\bm{s})] = \psi_{j} \textrm{E}\big[b_{j}^{''}(W_j(\bm{s}))]$ increases with $\psi_j$ while $\mathrm{cov}[\bm{g}^{-1}(\textrm{vec}(\bm{W}))]$ is independent of $\bm{\psi} = (\psi_1, \ldots, \psi_q)$, for any two locations $\bm{s}$ and $\bm{s}'$, the magnitude of correlation $\mathrm{corr}[Y_j(\bm{s}), Y_k(\bm{s}')]$ decreases as $\psi_j$ or $\psi_k$ increase. Intuitively, since different response types are conditionally uncorrelated, a large $\psi_j$ indicates that much of the variation in $Y_j(\bm{s})$ is independent of variation in the other responses. We provide a general definition of weak identifiability of the model parameters, as given in \cite{Ekvall22}, and derive a few results establishing key theoretical properties of our model.
\begin{definition}\label{def:identifiability}
    Suppose a model $\mc{P}_{\bm{\gamma}} = \{F_{\gamma} : \bm{\gamma} \in \bm{\Gamma} \}$ is a class of distribution functions $F_{\bm{\gamma}}$ parameterized by $\bm{\gamma} \in   \bm{\Gamma} \subseteq \mbb{R}^r$. Then $\mc{P}_{\bm{\gamma}}$ is said to be weakly identifiable for any component of $\bm{\gamma}$ with size $k$, say ${\bm{\gamma}}_k : k \in \{1,2, \ldots, d\},$ if the mapping ${\bm{\gamma}}_k \mapsto F_{{\gamma}_k}$ is injective.
\end{definition}
Definition~\ref{def:identifiability} means that two distinct subsets of ${\bm{\gamma}}_k : k \in \{1,2, \ldots, d\}$ should identify two different probability distributions, i.e., if $ \bm{\gamma}_k \neq {\bm{\gamma}}^{*}_k $ then $F_{{\bm{\gamma}}_k} \neq F_{{\bm{\gamma}}^{*}_k}$. Non-identifiable parameters often lead to inconsistent estimation.  Our model on observed locations $\mc{S}$ is parameterized by $[\textrm{vec}(\bm{W})^\top, \textrm{vec}(\bm{B})^\top, \textrm{vec}(\bm{\Sigma})^\top, \phi]^\top \in \mbb{R}^r$ where, $r = nq + pq + q(q+1)/2 +1$. The number of parameters in the model, $r$, clearly dominates the dimension of data of size $nq$, so parameter identifiability plays an important role in our analysis. We provide a result on the non-identifiability of the variance component for Binomial responses with a logit link.
\begin{theorem}\label{logit_identifiability}
Suppose the $j$-th response type in \eqref{our_model} is Binomial with a logit link. Then, model parameters are identifiable if $\bm{\Sigma}_{jj}$ is assumed to be fixed.
\end{theorem}

\begin{proof}
The proof is along with the lines of \citet[Theorem 1]{Ekvall22}. Suppose we fix a location $\bm{s} \in \mc{D}$. For the $j$-th response, 
$$ Y_j(\bm{s}) \;|\; W_j(\bm{s}) \overset{\mathrm{ind}}{\sim} \mathrm{Bin} \left(1,\dfrac{1}{1+\exp\{-W_j(\bm{s})\}}\right).$$
Here, $W_j(\bm{s}) \sim \mc{N}(\mc{X}(\bm{s})^{\top} \bm{\beta}_j, \bm{\Sigma}_{jj})$ with $\bm{\beta}_j$ denoting the regression coefficient vector corresponding to the $j$-th type response of length $p$. The marginal expectation of $Y_j(\bm{s})$ is
$$ \textrm{E}[Y_j(\bm{s})]
=
\textrm{E}_{W_j(\bm{s})} \big[\,\textrm{E}[Y_j(\bm{s})\mid W_j(\bm{s})]\,\big]
=
\int_{\mathbb{R}} \dfrac{1}{1+\exp\{-w\}} \,\pi(w)\,dw.$$
Here, $\pi(w)$ denotes the density of $W_j(\bm{s})$ evaluated at $w$. Under the model assumptions, we may write 
$$ W_j(\bm{s}) = \mc{X}(\bm{s})^{\top}\bm{\beta}_{j}+ \sqrt{\bm{\Sigma}_{jj}} Z,\;\; \text{with}\;\; Z \sim \mc{N}(0, 1).  $$
After a change of variables, the marginal success probability can be expressed as 
$$
\textrm{E}[Y_j(\bm{s})]
=
\int_{\mathbb{R}} \dfrac{1}{1+\exp{-\mc{X}(\bm{s})^{\top}\bm{\beta}_{j} - \sqrt{\bm{\Sigma}_{jj}}z }} \varphi(z)\,dz,$$ where $\varphi(z)$ is the density of a standard normal distribution evaluated at $z$.
Fixing all among $p$ coordinates of $\mc{X}(\bm{s})$ except one, say the $k$-th coordinate, and denoting $\bm{\beta}_{jk}$ the corresponding regression coefficient, we obtain the marginal success probability as a function of $(\bm{\beta}_{jk}, \bm{\Sigma}_{jj} ) $ as
\begin{equation}\label{success_prob}
    \delta(\bm{\beta}_{jk}; \bm{\Sigma}_{jj})
:= \textrm{E}[Y_j(\bm{s})] =
\int_{\mathbb{R}} \dfrac{1}{1+\exp\left\{-\sum_{l=1 (\neq k) }^{p} x_l(\bm{s}) \bm{\beta}_{jl} - \bm{\beta}_{jk} x_k(\bm{s}) - \sqrt{\bm{\Sigma}_{jj}} z\right\}} \;\varphi(z)dz.
\end{equation}
Here, $\sum_{l=1 (\neq k)}^{p} x_l(\bm{s}) \bm{\beta}_{jl}$ denotes the contribution of fixed effects except from the $k$-th component of $\mc{X}(\bm{s})$.  Without loss of generality, let $x_k(\bm{s})>0$ for the argument below. Our assumption is valid as our model includes an intercept term, for which $x_1(\bm{s}) = 1$, at least. The inverse-logit function $u\mapsto g(u):=(1+\exp\{-u\})^{-1}$ is smooth and strictly increasing with derivative $g'(u)=g(u)\{1-g(u)\}>0$. Differentiating with respect to $\bm{\beta}_{jk}$ in \eqref{success_prob} under the integral sign (justified by dominated convergence since $0<g(\cdot)<1$) yields
\begin{equation*}
    \begin{split}
        \dfrac{\partial}{\partial \bm{\beta}_{jk}} \delta(\bm{\beta}_{jk};\bm{\Sigma}_{jj}) &= \int_{\mathbb{R}} g' \big(\sum_{l=1 (\neq k)}^{p} x_l(\bm{s}) \bm{\beta}_{jl}+\bm{\beta}_{jk} x_k(\bm{s}) + \sqrt{\bm{\Sigma}_{jj}} z \big)\, x_k \;\varphi(z)dz\\
        &= x_k \int_{\mathbb{R}} g\big(\sum_{l=1 (\neq k)}^{p} x_l(\bm{s}) \bm{\beta}_{jl}+\bm{\beta}_{jk} x_k(\bm{s}) + \sqrt{\bm{\Sigma}_{jj}} z\big) \\
        & \qquad \times \{1-g\big(\sum_{l=1 (\neq k)}^{p} x_l(\bm{s}) \bm{\beta}_{jl}+\bm{\beta}_{jk} x_k(\bm{s}) + \sqrt{\bm{\Sigma}_{jj}} z \big)\}\,\varphi(z)dz. 
    \end{split}
\end{equation*}
Because $x_k>0$ and $g'(u) > 0, \;\text{for all}\; u \in \mbb{R}$, we have $ \dfrac{\partial}{\partial \bm{\beta}_{jk}} \delta(\bm{\beta}_{jk}; \bm{\Sigma}_{jj}) >0$, so $ \delta(\cdot; \bm{\Sigma}_{jj})$ is strictly increasing and continuous in $\bm{\beta}_{jk}$ for each fixed $\bm{\Sigma}_{jj}$. We next evaluate the limits of $ \delta(\bm{\beta}_{jk}; \bm{\Sigma}_{jj})$ as $\bm{\beta}_{jk}\to\pm\infty$. 
For any fixed $\bm{z}$,
$$ \lim_{\bm{\beta}_{jk}\to -\infty} g\left(\sum_{l=1 (\neq k)}^{p} x_l(\bm{s}) \bm{\beta}_{jl}+\bm{\beta}_{jk} x_k(\bm{s}) + \sqrt{\bm{\Sigma}_{jj}} z\right) = 0, \;\text{and}$$
$$
\lim_{\bm{\beta}_{jk}\to +\infty} g\left(\sum_{l=1 (\neq k)}^{p} x_l(\bm{s}) \bm{\beta}_{jl}+\bm{\beta}_{jk} x_k(\bm{s}) + \sqrt{\bm{\Sigma}_{jj}} z \right) = 1.$$
Since $Z$ is integrable, dominated convergence yields $$ \lim_{\bm{\beta}_{jk}\to -\infty} \delta(\bm{\beta}_{jk};\bm{\Sigma}_{jj})=0,
\qquad
\lim_{\bm{\beta}_{jk}\to +\infty} \delta(\bm{\beta}_{jk}; \bm{\Sigma}_{jj})=1.$$

To prove non-identifiability of $\bm{\Sigma}_{jj}$ we now fix two particular values of $\bm{\Sigma}_{jj}$ such that $\bm{\Sigma}_{jj}^{(1)} \neq \bm{\Sigma}_{jj}^{(2)}$.  By the continuity and strict monotonicity of $\bm{\beta}_{jk}\mapsto \delta(\bm{\beta}_{jk}; \bm{\Sigma}_{jj})$, for a unique $\bm{\beta}_{jk}$, the model can admit the same marginal probability $p_0:=f(\bm{\beta}_{jk} ;\bm{\Sigma}_{jj}^{(i)})\in(0,1) \;\text{for}\; i = 1,2$ and its limits $0$ and $1$ when $\bm{\beta}_{jk} \to \pm\infty$.
Thus, $(\bm{\beta}_{jk},\bm{\Sigma}_{jj}^{(1)})$ distinct from $(\bm{\beta}_{jk},\bm{\Sigma}_{jj}^{(2)})$ yield the same marginal success probability $p_0$ at location $\bm{s}$. Hence, $\delta$ is not an injective function of $\bm{\Sigma}_{jj}$ and thus not identifiable in $\bm{\Sigma}_{jj}$. \end{proof}

Theorem \ref{logit_identifiability} directly shows that different values of $\bm{\Sigma}_{jj}$ lead to identical marginal Binomial probabilities at location $\bm{s}$. Consequently, without an additional constraint, the variance component $\bm{\Sigma}_{jj}$ is not identifiable from marginal success probabilities. The choice of $\bm{\Sigma}_{jj}$ is study-specific and we discuss our choices in Section~\ref{sec:simulation_studies}. For our proposed Algorithm~\ref{alg:mixed_model_algorithm}, at each iteration, we constrain MCMC samples $\bm{\Sigma}^{*l}_{ij} = \bm{\Sigma}^{l}_{ij}/\sqrt{\bm{\Sigma}^{l}_{jj}},\; l = 1, \ldots, L$, where $L$ is the number of MCMC samples and use $\bm{\Sigma}^{*l}_{ij} $ for prediction. We next prove the identifiability of the spatial latent process $\bm{W}$, which eventually concludes the identifiability conditions for the model parameters in the next theorem.

\begin{lemma}\label{lemma_w}
For any fixed $\bm{s}$, let $Y_j(\bm{s}) \mid W_j(\bm{s}) \overset{\mathrm{ind}}{\sim}  \mathrm{EF} \left(Y_j(\bm{s}) \mid W_j(\bm{s}), \psi_j \right),\; j = 1, \cdots, q$.  Define $b_j'(w) = {db_j(w)}/{dw}$ to be the derivative of cumulant function at $w$ and the mean-value parameter as
    \begin{equation}
        \mu_j(\bm{s}_i) = \textrm{E}[Y_j(\bm{s}_i) \mid W_j(\bm{s}_i)] = b_j'(W_j(\bm{s}_i)).
    \end{equation}Under our model assumption on strict convexity of $b_j(\cdot)$, the mapping $ W_j(\bm{s}) \mapsto \mu_j(\bm{s}) $ is one-to-one, ensuring that the canonical parameterization of \eqref{spatial_exp_fam} is uniquely identifiable.
\end{lemma}
\begin{proof}
    To establish identifiability, we must show that at a fix location $\bm{s} \in \mc{D}$, $ W_j(\bm{s}) \mapsto \mu_j(\bm{s}) $ is injective, i.e., for any two values $ W_j(\bm{s}) $ and $ \widetilde{W}_j(\bm{s}) $, the equation
    $$ b_j'(W_j(\bm{s})) = b_j'(\widetilde{W}_j(\bm{s})) \implies W_j(\bm{s}) = \widetilde{W}_j(\bm{s}).$$
    Since $b_j$ is strictly convex by assumption, its derivative $ b_j' $ is strictly increasing, which guarantees that $ b_j' $ is injective. Thus, the mapping $ W_j(\bm{s}) \mapsto \mu_j(\bm{s}) $ is injective. Given that the injectivity is sufficient for identifiability, the canonical parameterization is identifiable.
\end{proof}
Lemma \ref{lemma_w} ensures that the latent process matrix of our model is identifiable. Despite identifiability of $\bm{W}$, consistent estimability cannot be guaranteed due to non-replicated spatial data; this is a major challenge that directly influences the quality of estimation of $\bm{W}$ and, consequently, the model parameters. Since the model parameters are hierarchically dependent on $\bm{W}$ from the second layer in the Figure~\ref{fig:dag}, the consequences of bad estimation of $\bm{W}$ are followed to the next layer as well. We now present a result that provides conditions for the identifiability of the model parameters $\bm{B}, \bm{\Sigma}, \phi$, which directly impact inference.
\begin{theorem}\label{B_Sigma_identifiability}
If the covariate matrix $\bm{X}$ over $n$-locations has full column rank $p < n$, the regression coefficient matrix $\bm{B}$ in our model \eqref{our_model} is identifiable. If $\nu$ is known, then $\bm{\Sigma}$ is identifiable up to a multiplicative constant and ${\bm{\Sigma}}/{\phi^{2\nu}}$ is identifiable.
\end{theorem}

\begin{proof}
    We define $\bm{\theta} = [\textrm{vec}(\bm{B})^\top, \textrm{vec}(\bm{\Sigma})^\top, \phi]^\top$ to be the vector of finite-dimensional parameters in \eqref{our_model}. Then the likelihood of $\bm{\theta}$, conditioning on the latent random matrix $\bm{W}$, is given by 
$$ 
L(\bm{\theta} | \bm{W}) = \dfrac{1}{\sqrt{(2\pi)^{nq}|\bm{\Sigma}|^{2n} |\bm{K}_{\phi}|^{2q}}}\exp{-\dfrac{1}{2}  \operatorname{tr} \Big( \bm{\Sigma}^{-1} (\bm{W} - \bm{X}\bm{B})^{\top} 
    \bm{K}_{\phi}^{-1} (\bm{W} - \bm{X}\bm{B}) \Big)}.
$$
We denote twice the negative log-likelihood $-2 \log{L(\bm{\theta} | \bm{W})}$ as
\begin{equation}
    l(\bm{\theta} | \bm{W}) = n\log{|\bm{\Sigma}|} + q\log{|\bm{K}_{\phi}|} + \operatorname{tr} \Big( \bm{\Sigma}^{-1} (\bm{W} - \bm{X}\bm{B})^{\top} 
    \bm{K}_{\phi}^{-1} (\bm{W} - \bm{X}\bm{B}) \Big).
\end{equation}
We first show the identifiability of the regression coefficient matrix $\bm{B}$. Suppose we consider two arbitrary vectors of parameters $\bm{\theta}_1 = [\textrm{vec}(\bm{B}_1)^\top, \textrm{vec}(\bm{\Sigma}_1)^\top, \phi_1]^\top$ and $\bm{\theta}_2 = [\textrm{vec}(\bm{B}_2)^\top, \textrm{vec}(\bm{\Sigma}_2)^\top, \phi_2]^\top$. We must show that if
$\bm{\theta}_1 \neq \bm{\theta}_2$, then $l(\bm{\theta}_1 | \bm{W})$ and $l(\bm{\theta}_2 | \bm{W})$ are different on a set of $\bm{W}$ with positive Lebesgue measure. Since the Matrix-Normal distribution of $\bm{W}$ is characterized by its mean matrix $\bm{X}\bm{B}$ and Kronecker-product derived covariance matrix $\bm{\Omega}_{\phi}$ uniquely, it suffices to show that $\bm{\theta}_1 \neq \bm{\theta}_2$ implies either different means or different covariance matrices. Suppose $\bm{B}_1$ and $\bm{B}_2$ are two distinct values of $\bm{B}$ which corresponds to distinct $\bm{\theta}_1$ and $\bm{\theta}_2$. Since $\bm{X}$ has full column rank, then $\bm{X}\bm{B}_1 \neq \bm{X}\bm{B}_2$. Consequently, $l(\bm{\theta}_1 | \bm{W}) \neq l(\bm{\theta}_2 | \bm{W})$. The mapping of $\bm{B} \rightarrow l(\bm{\theta} | \bm{W})$ is thus injective, and we have already shown $\bm{W}$ is identifiable in Lemma~\ref{lemma_w}. Hence, the regression coefficient matrix $\bm{B}$ is identifiable in our model.

Now we derive the conditions of identifiability of $\bm{\Sigma}$ and $\phi$. Since the random matrix $\bm{W}$ is Matrix-Normal with covariance parameters $\bm{K}_{\phi}$ and $\bm{\Sigma}$ is equivalent to assuming that $\textrm{vec}({\bm{W}})$ is multivariate normal with covariance $\bm{\Sigma}  \otimes \bm{K}_{\phi}$ (follows from \eqref{W_likelihood}). We interpret $\bm{K}_{\phi}$  as the common covariance matrix of the columns of $\bm{W}$ and $\bm{\Sigma}$ as the common covariance matrix of the rows of $\bm{W}$.
Indeed, for any $c > 0$,
$ (c^{-1} \bm{\Sigma}) \otimes (c\bm{K}_{\phi}) = \bm{\Sigma} \otimes  \bm{K}_{\phi}.$ However, without further restrictions, this interpretation is problematic since $\bm{K}_{\phi}$ and $\bm{\Sigma}$ are
only identified up to scaling by a constant. 
 The identifiability of the microergodic parameter ${\bm{\Sigma}}/{\phi^{2\nu}}$ is a direct corollary of Theorem 3 of \cite{bachoc2022asymptotically}.
\end{proof}

As discussed by \cite{Zhang2004inconsistent}, the estimability of weakly identified parameters is not possible in multivariate geostatistics under infill asymptotics, where the number of locations increases in a restricted compact domain. However, according to \cite{Stein1999}, we obtain asymptotically equivalent predictions from a model with such inconsistently estimated parameters. Moreover, if $\bm{K}_{\phi}$ and $\bm{\Sigma}$ are unidentified, one usually has to focus on estimation and interpretation only on their Kronecker product $\bm{\Sigma} \otimes  \bm{K}_{\phi}$. For fitting purposes, an identifiability constraint such as $\norm{\bm{\Sigma}} = 1$ (unit spectral norm) or $\bm{K}_{\phi}^{(1,1)} = 1$ (unit leading entry) is often imposed in the literature. Here, the spatial correlation matrix $\bm{K}_{\phi}$ is indeed with $\bm{K}_{\phi}^{(i, i)} = 1 \;\text{for all}\; i = 1,\cdots, n$, hence we do not require any identifiability constraint unless we have a logit link for $j$-th response. The following subsection discusses a scalable approach we adopt for large spatial datasets.

\subsection{Extension to high-dimensions: Vecchia approximation}\label{sec:Vecchia_appx}
A major drawback of GP-based methods is that inference time typically grows cubically with the size of the training set due to the need to invert a dense covariance matrix, making them impractical for large spatial datasets. To enable scalable inference, we adopt the Vecchia approximation \citep{Vecchia88}, which approximates the joint distribution of the latent process $\bm{W}$ by factorizing it into conditionals that depend only on a small subset of ordered locations. Let $\mc{N}(i) = \{1,\ldots, i - 1\} $ be the set of “preceding” ordered indices for the $i$-th ordered location, which is the ``full'' conditioning set. We induce sparsity into the spatial random effect $\bm{W}$ by defining a reduced conditioning index set $ \mc{M}(i) \subset \mc{N}(i) $ for the $i$-th ordered location. Here, we restrict $\mc{M}(i)$ to be a set of size at most $m$ for the $i$-th location. The latent process thus at the $i$-th ordered location is conditionally dependent only on the preceding indices that are contained in $ \mc{M}(i)$ of size less than or equal to $m$. We denote the corresponding latent process as $\mathcal{W}_{\mathcal{M}(i)} = \{ \mathcal{W}(\bm{s}_j) : j \in \mathcal{M}(i) \}.$
The Vecchia approximation to the joint density of $\bm{W}$ with sparsity parameter $m$ is given by  
\begin{equation}\label{eq:Vecchia-cond}
    \widetilde{\pi}_m(\bm{W} \mid \bm{\theta}) 
    = \prod_{i = 1}^{n} \pi\left(\mathcal{W}(\bm{s}_i) \,\middle|\, \mathcal{W}_{\mathcal{M}(i)}, \bm{\theta}\right).
\end{equation}
As the number of nearest neighbors $m$ increases, $\widetilde{\pi}_m(\bm{W} | \bm{\theta})$ becomes a better approximation to the true joint density, and is exact when $m = n - 1$. \cite{schafer2021a} discuss a detailed theoretical analysis of the Kullback-Leibler loss of the Vecchia approximated density with sparsity parameter $m$. In practice, a small $m$ achieves a favorable balance between accuracy and computational cost. 

We use the max–min ordering of spatial locations \citep{guinness2018permutation} with nearest-neighbor conditioning, so that each $\mathcal{M}(i)$ consists of the $\min\{m, i - 1\}$ closest preceding ordered locations. For notational simplicity, we write $\bm{K}$ for the covariance matrix and omit its indexing by $\phi$. Instead, for index sets $A$ and $B$, we denote by $\bm{K}_{A, B}$ the submatrix of $\bm{K}$ with rows indexed by $A$ and columns indexed by $B$. For each $i$, define
\begin{equation}\label{eq:veccia_coeff}
    \bm{A}_i = \bm{K}_{i,\mc{M}(i)} \bm{K}_{\mc{M}(i),\mc{M}(i)}^{-1} \in \mathbb{R}^{1 \times |\mc{M}(i)|}, 
\;
r_i = \bm{K}_{i,i} - \bm{K}_{i,\mc{M}(i)} \bm{K}_{\mc{M}(i),\mc{M}(i)}^{-1} \bm{K}_{\mc{M}(i),i} \in \mathbb{R}_+.
\end{equation}
The resulting conditional distributions are $\mathcal{W}(\bm{s}_i) \mid \mathcal{W}_{\mathcal{M}(i)} 
\sim \mc{N}_{q}\big( \bm{A}_i \mathcal{W}_{\mathcal{M}(i)}, \, r_i \bm{\Sigma} \big).$
These quantities define a sparse upper-triangular matrix $\bm{U}_{\phi,(m)} = (u_{ij})$ via
\[
u_{ii} = r_i^{-1/2}, \quad 
u_{ij} = -A_{ij}\, r_i^{-1/2} \;\; \text{for } j \in \mathcal{M}(i), \quad 
u_{ij} = 0 \;\; \text{otherwise},
\]
which yields the Vecchia precision matrix $\bm{K}^{-1} = \bm{U}_{\phi,(m)}^{\top}\bm{U}_{\phi,(m)}$. Accordingly, we replace the spatial covariance in \eqref{W_likelihood} with its Vecchia approximation and write
\begin{equation}\label{Vecchia_likelihood}
    \bm{W} \sim \mc{MN}_{n,q}\big(\bm{X}\bm{B}, (\bm{U}_{\phi,(m)}^{\top}\bm{U}_{\phi,(m)})^{-1}, \bm{\Sigma}\big).
\end{equation}
Each row of $\bm{U}_{\phi,(m)}$ can be computed in $\mc{O}(m^3)$ time, leading to an overall cost of matrix inversion $\mc{O}(nm^3)$ when $m \ll \sqrt{n}$. We provide the algorithms for fast Vecchia sampling and likelihood computation in Section~\ref{sec:vecchia_matnorm} for our model.

\section{Bayesian Computation}\label{sec:bayesian_computation}

The number of unknown parameters and latent variables in our model  \eqref{our_model}, i.e., the dimension of $(\bm{W}, \bm{B}, \bm{\Sigma}, \phi)$ grows as $\mc{O}(nq^2)$ as the number of locations $n$ and the number of responses $q$ increase. Jointly updating these many model parameters and the latent spatial effect is cumbersome in high dimensions \citep{Christensen06}. We implement a sparsity-informed component-wise sampler to update the unknown parameters, thereby improving the scalability of our algorithm for large spatial datasets.
We propose a Matrix-Normal Inverse-Wishart (MNIW) blocked-Gibbs sampler to update $(\bm{\Sigma}, \bm{B})$ jointly from the posterior as follows
$$ \bm{\Sigma}\; |\; {\bm{W}}, \phi   \sim \mc{IW}_q(\widetilde{\bm{S}} = \bm{S} + {\bm{W}}^{\top} \bm{K}^{-1} {\bm{W}} + {\bm{M}}^{\top} {\bm{V}}^{-1} {\bm{M}} - \widetilde{\bm{M}}^{\top} \widetilde{\bm{V}}
^{-1} \widetilde{\bm{M}} , \widetilde{v} = v + n),$$
$$ \bm{B}\; |\; \bm{\Sigma}, \phi , {\bm{W}} \sim \mc{MN}_{p,q} (\widetilde{\bm{M}}, \widetilde{\bm{V}}, \bm{\Sigma} ),$$
where $\widetilde{\bm{V}} = (\bm{X}^{\top}\bm{K}^{-1} \bm{X} + \bm{V}^{-1})^{-1}, \;\; \widetilde{\bm{M}} = \widetilde{\bm{V}}(\bm{X}^{\top}\bm{K}^{-1} {\bm{W}} + \bm{V}^{-1} \bm{M})$. The derivations are given in Section~\ref{sec:blocked_Gibbs_MNIW}. We update $\phi$ using a truncated normal random-walk proposal centered at the current value in the domain $(0, b_\phi)$.
Updating the spatial random effects $\bm{W} \in \mathbb{R}^{n \times q}$ via random-walk Metropolis–Hastings is computationally challenging due to strong posterior correlations and the spatial cross-covariance structure. These lead to poor mixing, high autocorrelation, and low acceptance rates, which is consistent with observations in \citet{Rue2007}. Moreover, proposal tuning is complex due to heterogeneity in scale across the components of $\bm{W}$. While samplers like preconditioned Crank–Nicolson \citep{cotter2013mcmc, Rudolf2022robust} can exploit the posterior geometry of $\bm{W}_{n\times q}$ in our model, they require $nq$ tuning parameters. 
We employ elliptical slice sampling \citep{murray2010elliptical}, which is well-suited for models with latent Gaussian priors. To account for the heterogeneous variability across components, we update $\bm{W}$ sequentially through its full conditional distributions. This approach enhances the mixing efficiency of the chain as discussed in Section~\ref{sec:quality_of_samplers}. Elliptical slice sampling is tuning-free, requires no gradients, MAP estimates, or Hessian approximations, and is scalable to high-dimensional settings when used with the Vecchia accelerated sampling from the Gaussian proposal in high dimensions. In our model, given $\bm{B}, \bm{\Sigma}, \phi, \bm{Y}$, the Gaussian prior in the full-conditional posterior of $\bm{W}$ ensures compatibility with elliptical slice sampling, making it a robust and efficient choice. In Step~2 of Algorithm~\ref{alg:mixed_model_algorithm}, we update the model parameters $(\bm{W}, \bm{B}, \bm{\Sigma}, \phi)$ in a sequential ordering based on the dimension of the component. We first update $\phi$, then update $(\bm{\Sigma}, \bm{B})$ jointly, and lastly $\bm{W}$. We recommend a warm start for initializing the parameters, so that the chain quickly explores the high-probability region.  We refer to \ref{sec:ess_W} for the detailed steps of the elliptical slice sampling algorithm for updating $\bm{W}$.

\section{Predictive Modeling}\label{sec:predictive_modeling}

An important objective of spatial models is to predict responses at unobserved locations $\mc{U} = \{\bm{s}_1^{*}, \bm{s}_2^{*}, \ldots, \bm{s}_u^{*} \}$ distinct from $\mc{S}$, in the presence of covariates. Let us denote the stacked matrices for responses, latent process, and covariates over the unobserved locations $\mc{U}$, respectively, as
$$
\bm{Y}^{*} = 
\begin{bmatrix}
\mc{Y}(\bm{s}_1^{*})^\top \\
\mc{Y}(\bm{s}_2^{*})^\top \\
\vdots \\
\mc{Y}(\bm{s}_u^{*})^\top
\end{bmatrix}_{u \times q},
\;\;
\bm{W}^{*} =
\begin{bmatrix}
\mc{W}(\bm{s}_1^{*})^\top \\
\mc{W}(\bm{s}_2^{*})^\top \\
\vdots \\
\mc{W}(\bm{s}_u^{*})^\top
\end{bmatrix}_{u \times q},
\;\; \text{ and } \;\;
\bm{X}^{*} =
\begin{bmatrix}
\mc{X}(\bm{s}_1^{*})^\top \\
\mc{X}(\bm{s}_2^{*})^\top \\
\vdots \\
\mc{X}(\bm{s}_u^{*})^\top
\end{bmatrix}_{u \times p}
.$$
The goal in predictive modeling is to compute the posterior predictive distribution (PPD) of $\bm{Y}^{*}$ given the observed data matrix $\bm{Y}$. We denote the PPD as
\begin{equation}\label{eq:ppd}
\pi(\bm{y}^{*} | \bm{y}) = \int \hspace{-0.2cm} \int \pi(\bm{y}^{*} | \bm{w}^{*}) \pi(\bm{w}^{*} | \bm{w}, \bm{B}, \bm{\Sigma}, \phi) \pi(\bm{w} | \bm{B}, \bm{\Sigma}, \phi)  \pi(\bm{w}, \bm{B}, \bm{\Sigma}, \phi | \bm{y})\;  d\bm{w}^{*}\; d\bm{w}\;  d\bm{B}\; d\bm{\Sigma}\; d\phi.
\end{equation}
Under the Gaussian process model, the joint distribution of $[\bm{W}^{\top}, \bm{W}^{*\top}]^{\top}$ for locations in $\mathcal{S}\cup\mathcal{U}$ is
$$ \begin{bmatrix}\bm{W}\\\bm{W}^*\end{bmatrix}
\sim \mc{MN}_{u+n,q}\left(\,
\begin{bmatrix}\bm{X}\bm{B}\\\bm{X}^*\bm{B}\end{bmatrix},\;
\bm{K}^{(n+u)}=
\begin{bmatrix}\bm{K}^{(n,n)} & \bm{K}^{(n,u)}\\
\bm{K}^{(u,n)} & \bm{K}^{(u,u)}\end{bmatrix},\;
\bm{\Sigma}\right).$$
 For our proposed model, denoting $\bm{M}_{W*} = \bm{X}^{*} \bm{B} + \bm{K}^{(u,n)} (\bm{K}^{(n,n)})^{-1} (\bm{W} - \bm{X} \bm{B}) $ and $\bm{K}_{W*} = \bm{K}^{(u,u)} - \bm{K}^{(u,n)} (\bm{K}^{(n,n)})^{-1} \bm{K}^{(n,u)}$, the joint Gaussian distribution of $[\bm{W}^{*\top}, \bm{W}^{\top}]^{\top}$ leads to the following conditional distribution of latent process
\begin{equation}\label{krigging_eq}
\bm{W}^{*} \;|\; \bm{W}, \bm{B}, \bm{\Sigma}, \phi \sim \mc{MN}_{u, q}\left(\bm{M}_{W*}, \bm{K}_{W*} ,\bm{\Sigma}
\right).
\end{equation}
However, direct evaluation of the covariance blocks over $\mc{U} \cup \mc{S}$ becomes computationally prohibitive for large $n$ and $u$. To address this, we extend the Vecchia approximation \eqref{eq:Vecchia-cond} to the augmented set of locations and work with the corresponding precision matrix $\bm{Q}^{(u+n)} = (\bm{K}^{(u+n)})^{-1}$. Let
\begin{equation}\label{precision_parameterization}
\bm{Q}^{(u+n)} =
\begin{bmatrix}
\bm{Q}^{(u,u)} & \bm{Q}^{(u,n)}\\
\bm{Q}^{(n,u)} & \bm{Q}^{(n,n)}
\end{bmatrix}.
\end{equation}
Then \eqref{krigging_eq} can be equivalently written as
\begin{equation}\label{pred_eq}
\bm{W}^{*} \mid \bm{W}, \bm{B}, \bm{\Sigma}, \phi \sim \mc{MN}_{u, q}\left(
\bm{X}^{*} \bm{B} - (\bm{Q}^{(u,u)})^{-1} \bm{Q}^{(u,n)} (\bm{W} - \bm{X} \bm{B}),
(\bm{Q}^{(u,u)})^{-1},
\bm{\Sigma}
\right).
\end{equation}
 We provide the derivation of the posterior latent predictive process in Section~\ref{sec:latent_ppd}. Under the Vecchia approximation, the precision matrix admits the sparse factorization $\bm{Q}^{(u+n)} \approx \bm{U}_{\phi}^{(u+n)\top}\bm{U}_{\phi}^{(u+n)}$, where $\bm{U}_{\phi}^{(u+n)}$ is constructed from the Vecchia coefficients. We further leverage the Vecchia coefficients obtained from the ordered prediction locations, thus using the non-zero entries of the sparse $\bm{U}_{\phi}^{(u+n)}$, and avoid explicit manipulation of dense covariance matrices. Posterior samples from $\bm{Y}^{*} \mid \bm{Y}$ are obtained via Monte Carlo integration in \eqref{eq:ppd}.

For model comparison, we consider the expected log joint predictive density (ELJPD; \citealp{cooper2025cross}) over $\mc{U}$, defined as
$
\mathrm{ELJPD} 
= \mathbb{E}_{\bm{Y}^* \sim \pi_0}
\left[ \log \pi(\bm{Y}^* \mid \bm{Y}) \right],
$
where $\pi_0$ denotes the true data-generating distribution, which is generally not known. In practice, we approximate this quantity using $L$ posterior predictive samples $\{\bm{W}^{*(l)}\}_{l=1}^L$ and compute the Monte Carlo estimator of ELJPD by
\begin{equation*}
    \widehat{\mathrm{ELJPD}} 
= \log \left( \frac{1}{L} \sum_{l=1}^L \exp \left( \sum_{i=1}^u \sum_{j=1}^q 
\log \pi\big(Y_j(\bm{s}_i^*) \mid W_j^{*(l)}(\bm{s}_i^*)\big) \right) \right).
\end{equation*}
Higher ELJPD values indicate a better model selection rule.

The overall computational complexity presented in Algorithm~\ref{alg:mixed_model_algorithm} is broken down across three primary stages. In {Step 1}, we carry out pre-computing foundations for the Vecchia approximation on the training set $\mc{S}$ and the test set $\mc{U}$. We perform fast max-min ordering as suggested in \cite{guinness2018permutation} with leading computational costs of $\mathcal{O}(n^{*2}\log{n^{*}}),$ where $n^{*} = n + u$. Here, Step~2 integrates posterior inference and prediction within a unified MCMC framework as discussed in Section~\ref{sec:bayesian_computation} and Section~\ref{sec:predictive_modeling}. All updates are expressed through local Vecchia regression coefficients as given in \eqref{eq:veccia_coeff}, leading to a per-iteration cost of $\mathcal{O}(nm^3)$ for updating $\phi$ and $\bm{W}$, instead of the $\mathcal{O}(n^3)$ cost of full Gaussian process models.  We generate predictive samples at a set of new spatial locations $\mc{U}$, where all responses are unknown, using the same local conditioning structure via $(\bm{A}_i^*, r_i^*)$. This workflow avoids dense covariance operations and requires $\mathcal{O}(u m q)$ computations. Our algorithm achieves linear scaling in $n$ and $u$ by using a fixed number of local conditioning sets of size $m \ll \text{min}(n,u)$, ensuring practical scalability even in high-dimensional spatial domains.

\begin{algorithm}[]
\caption{Posterior inference and prediction for mixed-type response model}
\label{alg:mixed_model_algorithm}

\begin{adjustbox}{max width=\textwidth}
\begin{minipage}{\textwidth}

\begin{algorithmic}[1]
\vspace{0.5ex}
\Statex \noindent\textbf{Pre-computing steps for Vecchia approximation}
\vspace{-1.5ex}
\Statex \noindent\rule{\linewidth}{0.4pt}
\State \textbf{Require:} $\texttt{family}$ of length $q$, $\mc{S}$, $\mc{U}$ of size $n$ and $u$, $m$
\State $n^* = n + u$; obtain $\texttt{max-min}(\mathcal{S})$ and $\texttt{max-min}(\mathcal{S} \cup \mathcal{U})$ \hfill $\mathcal{O}(n^{*2} \log n^*)$
\State Construct $\texttt{NNarray}_{\mc{S} \cup \mc{U}}$ and store distances $D_{\mc{S} \cup \mc{U}}$
\vspace{-1ex}
\Statex \noindent\rule{\linewidth}{0.4pt}
\vspace{-3ex}
\Statex \noindent\textbf{MCMC workflow of estimation and prediction}
\vspace{-1.5ex}
\Statex \noindent\rule{\linewidth}{0.4pt}

\State Initialize $\{\bm{W}^{(0)}, \bm{B}^{(0)}, \bm{\Sigma}^{(0)}, \phi^{(0)}\}$
\State Compute Vecchia coefficients from $\texttt{NNarray}_{\mc{U} \cup \mc{S}}$ (Section~\ref{sec:vecchia_matnorm})

\For{$l=1,\ldots,L$}

\State Update $\phi^{(l)}$ via Metropolis--Hastings with truncated-normal proposal \hfill $\mathcal{O}(n m^3)$

\For{$i=1,\ldots,n$}
\State Recompute $(\bm{A}_i, r_i)$ given $\phi^{(l)}$ \hfill $\mathcal{O}(n m^3)$
\State Compute $(\bm{U}_{\phi^{(l)}} \bm{W})_i = r_i^{-1/2}\big(\bm{W}_i - \bm{A}_i \bm{W}_{\mc{N}(i)}\big)$ \hfill $\mathcal{O}(n m q)$
\State Compute $(\bm{U}_{\phi^{(l)}} \bm{X})_i = r_i^{-1/2}\big(\bm{X}_i - \bm{A}_i \bm{X}_{\mc{N}(i)}\big)$ \hfill $\mathcal{O}(n m p)$
\EndFor

\State $\widetilde{\bm{V}} = \big((\bm{U}_{\phi^{(l)}}\bm{X})^{\top}(\bm{U}_{\phi^{(l)}}\bm{X}) + \bm{U}_{\bm{V}}^{\top}\bm{U}_{\bm{V}}\big)^{-1}$; $\bm{U}_{\widetilde{\bm{V}}} = \texttt{chol}(\widetilde{\bm{V}})$ \hfill $\mathcal{O}(p^3)$

\State $\widetilde{\bm{M}} = \widetilde{\bm{V}}\big((\bm{U}_{\phi^{(l)}}\bm{X})^{\top}(\bm{U}_{\phi^{(l)}}\bm{W}) + \bm{U}_{\bm{V}}^{\top}\bm{U}_{\bm{V}}\bm{M}\big)$ \hfill $\mathcal{O}(n p q)$

\State $\widetilde{\bm{S}} = \bm{S} + (\bm{U}_{\phi^{(l)}}\bm{W})^{\top}(\bm{U}_{\phi^{(l)}}\bm{W}) + (\bm{U}_{V}\bm{M})^{\top}(\bm{U}_{V}\bm{M}) - (\bm{U}_{\widetilde{V}}\widetilde{\bm{M}})^{\top}(\bm{U}_{\widetilde{V}}\widetilde{\bm{M}})$ \hfill $\mathcal{O}(n q^2)$

\State $\widetilde{v} = v + n$ \hfill $\mathcal{O}(1)$

\State Sample $\bm{\Sigma}^{(l)} \sim \mathcal{IW}(\widetilde{\bm{S}},\widetilde{v})$; $\bm{U}_{\bm{\Sigma}^{(l)}} = \texttt{chol}(\bm{\Sigma}^{(l)})$ \hfill $\mathcal{O}(q^3)$

\For{$j=1,\ldots,q$}
\State Enforce constraints on $\bm{\Sigma}^{(l)}$ if $\texttt{family}[j]=\texttt{Binomial}$
\EndFor

\State Sample $\bm{Z}\sim \mathcal{MN}_{p,q}(0,\bm{I}_p,\bm{I}_q)$ \hfill $\mathcal{O}(pq)$

\State $\bm{B}^{(l)} = \widetilde{\bm{M}} + \bm{U}_{\widetilde{\bm{V}}}\bm{Z}\bm{U}_{\bm{\Sigma}^{(l)}}^{\top}$ \hfill $\mathcal{O}(pq^2+q^2)$

\State Update $\bm{W}^{(l)}$ via component-wise elliptical slice sampling (Section~\ref{sec:ess_W}) \hfill $\mathcal{O}(n m^3)$

\State \textbf{Save} $\{(\bm{W}^{(l)}, \bm{B}^{(l)}, \bm{\Sigma}^{(l)}, \phi^{(l)})\}$

\For{$i=1,\ldots,u$}
\State Extract $(\bm{A}_i^*, r_i^*)$ from $\texttt{NNarray}_{\mc{U}\cup\mc{S}}$ \hfill $\mathcal{O}(u m^2)$
\State Sample $\bm{Z}_i \sim \mc{N}_q(\bm{0},\bm{\Sigma}^{(l)})$ \hfill $\mathcal{O}(u q^2)$
\State $\bm{W}_i^* = \bm{X}_i^*\bm{B}^{(l)} + \bm{A}_i^*(\bm{W}_{\mc{N}(i)}^* - \bm{X}_{\mc{N}(i)}^*\bm{B}^{(l)}) + \sqrt{r_i^*}\bm{Z}_i$ \hfill $\mathcal{O}(u m q)$
\State Generate $Y_j(\bm{s}_i^*) \mid W_j(\bm{s}_i^*) \sim \mathrm{EF}(W_j(\bm{s}_i^*), \psi_j)$ \hfill $\mathcal{O}(u q)$
\State $\log \pi(\bm{Y}^* \mid \bm{W}^{*(l)}) = \sum_{i=1}^u \sum_{j=1}^q \log \pi(Y_j(\bm{s}_i^*) \mid W_j^{*(l)}(\bm{s}_i^*))$ \hfill $\mathcal{O}(u q)$
\EndFor

\EndFor

\State $\widehat{\mathrm{ELJPD}} = \log\left(\frac{1}{L} \sum_{l=1}^L \exp\big(\log \pi(\bm{Y}^* \mid \bm{W}^{*(l)})\big)\right)$ \hfill $\mathcal{O}(L)$

\State \textbf{Save} $\widehat{\mathrm{ELJPD}}$
\vspace{1ex}
\end{algorithmic}

\end{minipage}
\end{adjustbox}

\end{algorithm}

\section{Data analysis}\label{sec:data_analysis}
We evaluate the statistical performance of our proposed joint Bayesian model for mixed-type spatial responses by comparing it with a commonly used alternative of separate modeling. By separate modeling, we refer to the approach detailed in Section~\ref{sec:sep_model}, in which each response is modeled independently, even when the true data-generating process exhibits cross-response dependence (i.e., $\bm{\Sigma} _ {ij} \neq 0$ for $i \neq j$). Under a common spatial covariance structure with shared $\phi$ across components, a separate model inherently assumes $\bm{\Sigma}_{ij} = 0$ and fits independent univariate latent GP with a flat prior on $\bm{\Sigma}_{jj} \overset{\text{ind}}{\sim} \mc{IG}(a = 0.01, b = 0.01)$, for each $j$-th response type. Each component is thus endowed with its own mean structure, thereby ignoring potential dependence between different response types.

We use weak prior information for the regression coefficient matrix $\bm{B}$ by specifying the mean matrix of $\bm{B}$ as $\bm{M} = \bm{0}$ and the covariance matrix $\bm{V} = 100\bm{I}_p$. 
For the cross-response covariance $\bm{\Sigma}$, we choose $\bm{S} = \bm{I}_q$ and $v = q + 1$, providing a weakly-informative prior. The full conditional posterior of $\bm{\Sigma}$ admits a valid density, and thus the posterior is proper. We choose the smoothness parameter $\nu = 0.5$ for the Mat\'ern kernel in \eqref{eq:matern_kernel} based on the empirical semivariogram from exploratory data analysis. We use $m = 20$ as the size of the conditioning set in the Vecchia approximation. We refer the reader to Section~\ref{sec:choice_of_m} for a detailed justification of our choice. The hyperparameter $b_{\phi}$ for the spatial range parameter, $\phi$, is chosen so that the effective range corresponds to a correlation of $0.05$ at the domain diameter $\Delta$. The \texttt{R} codes to reproduce the simulation studies and real data analysis are available on GitHub\footnote{Available at \url{https://github.com/ArghyaStat/Bayesian_mixed_type_spatial_model}.}.

\subsection{Simulation studies}\label{sec:simulation_studies}
 Our analysis is based on studying the quality of estimation and the model performance in terms of the strength of prediction. We consider three bivariate mixed-type case studies on ``Binomial-Gaussian'', ``Binomial-Poisson'', and ``Gaussian-Poisson'' (in Section~\ref{sec:gp_simulations}) response models along with a trivariate model involving ``Binomial--Gaussian--Poisson'' responses. For each case, we generate gridded spatial locations over the unit square $\mc{D} = [0,1]^2$, under both low-dimensional ($n = 100$) and high-dimensional ($n = 2500$) regimes.  We leave out $20\%$ of the randomly chosen locations for predictive analysis. For each of the following cases, we segment our analysis further on weak and strong spatial correlation $\phi_0 \in \{0.1, 0.3\}$, and both independent and dependent structures for the cross-covariance matrix, $\bm{\Sigma}^{(0)}$. The specification of $\bm{\Sigma}^{(0)}$ is tailored to each response configuration, and we provide general guidelines for its construction. At any location $\bm{s}$, we consider the linear covariates as $\mc{X}(\bm{s}) = [1, \textrm{lon}(\bm{s}), \textrm{lat}(\bm{s})]^{\top}$, which is a standard choice in spatial data analysis. We specify the true regression coefficient matrix as $\bm{B}^{(0)} = [(1.0, -0.5)^{\top},  (3,  1.5)^{\top}, (-1.2,  0)^{\top}]$ for all bivariate spatial response models, and $\bm{B}^{(0)} = [(1.0, -0.5, 0.8)^{\top},  (3,  1.5, -2.0)^{\top}, (-1.2,  0, 0.7)^{\top}]$ for ``Binomial-Gaussian-Poisson'' response model. For each model, we implement Algorithm~\ref{alg:mixed_model_algorithm} and evaluate both model fit and predictive performance over $50$ independently generated datasets. The results are summarized through the following metrics:
\begin{enumerate}
    \item Tables detailing the estimation summaries for the cross-covariance parameters in $\bm{\Sigma}$,
    \item Comparative boxplots of the marginal posterior variance of regression coefficients, contrasting the proposed joint model with separate models,
    \item Tables showing the differences in estimated expected log predictive density ($\widehat{\mathrm{ELJPD}}$) between the joint and separate models.
\end{enumerate}

In our simulation studies, we specify $\bm{\Sigma}^{(0)}$ to induce strong cross-response dependence across different model settings in accordance with \citet[Examples 1 and 3]{Ekvall22}. We choose $\mathrm{vec}(\bm{\Sigma}^{(0)}) = [9, 5, 5, 3]^\top$ for ``Binomial-Gaussian'' model, $\mathrm{vec}(\bm{\Sigma}^{(0)}) = [9, 4, 4, 2]^\top$ for ``Binomial-Poisson'' model, $\mathrm{vec}(\bm{\Sigma}^{(0)}) = [3, 9/4, 9/4, 2]^\top$ for ``Gaussian-Poisson'' model, and $\mathrm{vec}(\bm{\Sigma}^{(0)}) = [9, 5, 4, 5, 3, 9/4, 4, 9/4, 2]^\top$ for ``Binomial–Gaussian–Poisson'' model. 


We summarize the performance of the proposed model through posterior inference on the off-diagonal elements of $\bm{\Sigma}$, which characterize cross-dependence among response types. For bivariate models with ``Binomial-Gaussian'' and ``Binomial-Poisson'' response types, this reduces to a single parameter $\bm{\Sigma}_{12}$. In contrast, the trivariate ``Binomial-Gaussian-Poisson'' model involves $\bm{\Sigma}_{12}$, $\bm{\Sigma}_{13}$, and $\bm{\Sigma}_{23}$. Results in Tables~\ref{tab:crosscov_combined}, \ref{tab:crosscov_bgp} show that, when cross-dependence is present, i.e., when $\bm{\Sigma}_{ij} \neq 0,\; i \neq j,$  the posterior credible intervals for the off-diagonal components exclude zero and the corresponding posterior means, averaged over 50 replicated datasets, are consistently bounded away from zero with satisfactory coverage. This finding indicates that our proposed framework effectively recovers joint dependence across response types and highlights the benefits of joint modeling in such settings.
When the true cross-covariance structure is diagonal, the model adapts accordingly. Posterior means remain close to zero, and the credible intervals exhibit good coverage of the true value. This finding demonstrates that the proposed approach remains well-calibrated even in the absence of cross-dependence, providing a flexible alternative that performs comparably to separate modeling strategies in such scenarios.
\begin{table}[!h]
\centering
\caption{Posterior estimation summary of $\bm{\Sigma}_{12}$ (posterior mean, 95\% posterior credible interval, and empirical coverage at the nominal 95\% level) across 50 replicated datasets under varying spatial correlation ($\phi_0$, the true value of $\phi$) and varying cross-covariance matrix ($\bm{\Sigma}^{(0)}$, the true value of $\bm{\Sigma}$) for the Binomial-Gaussian and Binomial-Poisson models. For the Binomial-Gaussian model, we choose the diagonal entries of $\bm{\Sigma}^{(0)}$ to be $\bm{\Sigma}^{(0)}_{11} = 9$ and $\bm{\Sigma}^{(0)}_{22} = 3$. For the Binomial-Poisson model, we choose $\bm{\Sigma}^{(0)}_{11} = 9$ and $\bm{\Sigma}^{(0)}_{22} = 2$.}
\label{tab:crosscov_combined}
\renewcommand{\arraystretch}{1}
\small
\begin{tabular}{c c c c c}
\hline
\noalign{\vskip 1pt}
$\bm{\Sigma}_{12}^{(0)}$ & $\phi_0$ & Posterior mean \text{\scriptsize(SE)} & Credible interval \text{\scriptsize(SE)} & Coverage \text{\scriptsize(SE)} \\[2pt]
\hline
\noalign{\vskip 2pt}
\multicolumn{5}{c}{Binomial--Gaussian} \\[2pt]

\rowcolor{lightgray}
\multicolumn{5}{c}{Sample size: $n = 100$} \\
\multirow{2}{*}{$5$}
& $0.3$ & $3.687\!\,\text{\scriptsize(0.130)}$ & $[2.217\!\,\text{\scriptsize(0.131)},\; 5.397\!\,\text{\scriptsize(0.145)}]$ & $0.68\!\,\text{\scriptsize(0.071)}$ \\
& $0.1$ & $4.608\!\,\text{\scriptsize(0.105)}$ & $[3.206\!\,\text{\scriptsize(0.103)},\; 6.250\!\,\text{\scriptsize(0.122)}]$ & $0.90\!\,\text{\scriptsize(0.043)}$ \\

\rowcolor{lightviolet}
\multicolumn{5}{c}{Sample size: $n = 2500$} \\
\multirow{2}{*}{$5$}
& $0.3$ & $4.052\!\,\text{\scriptsize(0.061)}$ & $[3.538\!\,\text{\scriptsize(0.064)},\; 4.871\!\,\text{\scriptsize(0.038)}]$ & $0.64\!\,\text{\scriptsize(0.061)}$ \\
& $0.1$ & $4.628\!\,\text{\scriptsize(0.049)}$ & $[4.200\!\,\text{\scriptsize(0.045)},\; 5.124\!\,\text{\scriptsize(0.049)}]$ & $0.82\!\,\text{\scriptsize(0.070)}$ \\

\rowcolor{lightgray}
\multicolumn{5}{c}{Sample size: $n = 100$} \\
\multirow{2}{*}{$0$}
& $0.3$ & $0.008\!\,\text{\scriptsize(0.233)}$ & $[-2.418\!\,\text{\scriptsize(0.218)},\; 2.429\!\,\text{\scriptsize(0.236)}]$ & $0.84\!\,\text{\scriptsize(0.052)}$ \\
& $0.1$ & $-0.304\!\,\text{\scriptsize(0.228)}$ & $[-2.849\!\,\text{\scriptsize(0.225)},\; 2.291\!\,\text{\scriptsize(0.234)}]$ & $0.86\!\,\text{\scriptsize(0.050)}$ \\

\rowcolor{lightviolet}
\multicolumn{5}{c}{Sample size: $n = 2500$} \\
\multirow{2}{*}{$0$}
& $0.3$ & $0.023\!\,\text{\scriptsize(0.076)}$ & $[-0.691\!\,\text{\scriptsize(0.075)},\; 0.694\!\,\text{\scriptsize(0.073)}]$ & $0.82\!\,\text{\scriptsize(0.055)}$ \\
& $0.1$ & $-0.037\!\,\text{\scriptsize(0.062)}$ & $[-0.663\!\,\text{\scriptsize(0.062)},\; 0.592\!\,\text{\scriptsize(0.067)}]$ & $0.90\!\,\text{\scriptsize(0.043)}$ \\

\hline
\noalign{\vskip 2pt}
\multicolumn{5}{c}{Binomial--Poisson} \\[2pt]

\rowcolor{lightgray}
\multicolumn{5}{c}{Sample size: $n = 100$} \\
\multirow{2}{*}{$4$}
& $0.3$ & $3.002\!\,\text{\scriptsize(0.115)}$ & $[1.574\!\,\text{\scriptsize(0.136)},\; 4.564\!\,\text{\scriptsize(0.126)}]$ & $0.64\!\,\text{\scriptsize(0.069)}$ \\
& $0.1$ & $3.881\!\,\text{\scriptsize(0.135)}$ & $[2.560\!\,\text{\scriptsize(0.107)},\; 5.517\!\,\text{\scriptsize(0.196)}]$ & $0.92\!\,\text{\scriptsize(0.039)}$ \\

\rowcolor{lightviolet}
\multicolumn{5}{c}{Sample size: $n = 2500$} \\
\multirow{2}{*}{$4$}
& $0.3$ & $3.523\!\,\text{\scriptsize(0.046)}$ & $[3.092\!\,\text{\scriptsize(0.055)},\; 4.090\!\,\text{\scriptsize(0.031)}]$ & $0.66\!\,\text{\scriptsize(0.068)}$ \\
& $0.1$ & $3.795\!\,\text{\scriptsize(0.044)}$ & $[3.419\!\,\text{\scriptsize(0.045)},\; 4.222\!\,\text{\scriptsize(0.042)}]$ & $0.72\!\,\text{\scriptsize(0.064)}$ \\

\rowcolor{lightgray}
\multicolumn{5}{c}{Sample size: $n = 100$} \\
\multirow{2}{*}{$0$}
& $0.3$ & $-0.207\!\,\text{\scriptsize(0.176)}$ & $[-2.255\!\,\text{\scriptsize(0.172)},\; 1.836\!\,\text{\scriptsize(0.185)}]$ & $0.88\!\,\text{\scriptsize(0.046)}$ \\
& $0.1$ & $-0.009\!\,\text{\scriptsize(0.194)}$ & $[-2.237\!\,\text{\scriptsize(0.198)},\; 2.118\!\,\text{\scriptsize(0.201)}]$ & $0.86\!\,\text{\scriptsize(0.050)}$ \\

\rowcolor{lightviolet}
\multicolumn{5}{c}{Sample size: $n = 2500$} \\
\multirow{2}{*}{$0$}
& $0.3$ & $0.008\!\,\text{\scriptsize(0.049)}$ & $[-0.500\!\,\text{\scriptsize(0.050)},\; 0.505\!\,\text{\scriptsize(0.052)}]$ & $0.92\!\,\text{\scriptsize(0.039)}$ \\
& $0.1$ & $0.016\!\,\text{\scriptsize(0.047)}$ & $[-0.461\!\,\text{\scriptsize(0.052)},\; 0.497\!\,\text{\scriptsize(0.046)}]$ & $0.86\!\,\text{\scriptsize(0.050)}$ \\

\hline
\end{tabular}
\end{table}

\clearpage
\thispagestyle{empty}
\begin{table}[!h]
\centering
\caption{Posterior estimation summary of cross-covariances (posterior mean, 95\% posterior credible interval, and empirical coverage at the nominal 95\% level) across 50 replicated datasets for Binomial-Gaussian-Poisson response model. For the Binomial-Gaussian-Poisson model, we choose the diagonal entries of $\bm{\Sigma}^{(0)}$ to be $\bm{\Sigma}^{(0)}_{11} = 9$,  $\bm{\Sigma}^{(0)}_{22} = 3$, and $\bm{\Sigma}^{(0)}_{33} = 2$.}
\label{tab:crosscov_bgp}
\renewcommand{\arraystretch}{1}
\small
\begin{tabular}{c c c c c}
\hline
 & $\phi_0$ & Posterior mean \text{\scriptsize(SE)} & Credible interval \text{\scriptsize(SE)} & Coverage \text{\scriptsize(SE)} \\[2pt]
\hline
\noalign{\vskip 2pt}
$\bm{\Sigma}_{12}^{(0)}$ & \multicolumn{4}{c}{\textbf{$\bm{\Sigma}_{12}$}} \\[2pt]

\rowcolor{lightgray}
\multicolumn{5}{c}{Sample size: $n = 100$} \\
\multirow{2}{*}{$5$}
& $0.3$ & $3.905\!\,\text{\scriptsize(0.120)}$ & $[2.515\!\,\text{\scriptsize(0.122)},\; 5.546\!\,\text{\scriptsize(0.136)}]$ & $0.70\!\,\text{\scriptsize(0.065)}$ \\
& $0.1$ & $4.854\!\,\text{\scriptsize(0.098)}$ & $[3.542\!\,\text{\scriptsize(0.087)},\; 6.411\!\,\text{\scriptsize(0.116)}]$ & $0.92\!\,\text{\scriptsize(0.039)}$ \\

\rowcolor{lightviolet}
\multicolumn{5}{c}{Sample size: $n = 2500$} \\
\multirow{2}{*}{$5$}
& $0.3$ & $3.974\!\,\text{\scriptsize(0.052)}$ & $[3.462\!\,\text{\scriptsize(0.052)},\; 4.895\!\,\text{\scriptsize(0.041)}]$ & $0.64\!\,\text{\scriptsize(0.050)}$ \\
& $0.1$ & $4.669\!\,\text{\scriptsize(0.043)}$ & $[4.242\!\,\text{\scriptsize(0.038)},\; 5.151\!\,\text{\scriptsize(0.048)}]$ & $0.70\!\,\text{\scriptsize(0.065)}$ \\[4pt]

\rowcolor{lightgray}
\multicolumn{5}{c}{Sample size: $n = 100$} \\
\multirow{2}{*}{$0$}
& $0.3$ & $0.003\!\,\text{\scriptsize(0.209)}$ & $[-2.241\!\,\text{\scriptsize(0.194)},\; 2.262\!\,\text{\scriptsize(0.205)}]$ & $0.92\!\,\text{\scriptsize(0.039)}$ \\
& $0.1$ & $-0.140\!\,\text{\scriptsize(0.219)}$ & $[-2.448\!\,\text{\scriptsize(0.202)},\; 2.142\!\,\text{\scriptsize(0.225)}]$ & $0.86\!\,\text{\scriptsize(0.050)}$ \\

\rowcolor{lightviolet}
\multicolumn{5}{c}{Sample size: $n = 2500$} \\
\multirow{2}{*}{$0$}
& $0.3$ & $-0.007\!\,\text{\scriptsize(0.061)}$ & $[-0.630\!\,\text{\scriptsize(0.061)},\; 0.629\!\,\text{\scriptsize(0.068)}]$ & $0.94\!\,\text{\scriptsize(0.034)}$ \\
& $0.1$ & $0.040\!\,\text{\scriptsize(0.049)}$ & $[-0.534\!\,\text{\scriptsize(0.046)},\; 0.632\!\,\text{\scriptsize(0.062)}]$ & $0.94\!\,\text{\scriptsize(0.034)}$ \\[2pt]
\hline
\noalign{\vskip 2pt}
$\bm{\Sigma}_{13}^{(0)}$ & \multicolumn{4}{c}{\textbf{$\bm{\Sigma}_{13}$}} \\[2pt]

\rowcolor{lightgray}
\multicolumn{5}{c}{Sample size: $n = 100$} \\
\multirow{2}{*}{$4$}
& $0.3$ & $3.198\!\,\text{\scriptsize(0.094)}$ & $[2.008\!\,\text{\scriptsize(0.098)},\; 4.670\!\,\text{\scriptsize(0.111)}]$ & $0.76\!\,\text{\scriptsize(0.061)}$ \\
& $0.1$ & $4.069\!\,\text{\scriptsize(0.121)}$ & $[2.850\!\,\text{\scriptsize(0.086)},\; 5.659\!\,\text{\scriptsize(0.227)}]$ & $0.94\!\,\text{\scriptsize(0.034)}$ \\

\rowcolor{lightviolet}
\multicolumn{5}{c}{Sample size: $n = 2500$} \\
\multirow{2}{*}{$4$}
& $0.3$ & $3.221\!\,\text{\scriptsize(0.045)}$ & $[2.807\!\,\text{\scriptsize(0.045)},\; 3.803\!\,\text{\scriptsize(0.036)}]$ & $0.66\!\,\text{\scriptsize(0.052)}$ \\
& $0.1$ & $3.798\!\,\text{\scriptsize(0.034)}$ & $[3.457\!\,\text{\scriptsize(0.032)},\; 4.195\!\,\text{\scriptsize(0.038)}]$ & $0.78\!\,\text{\scriptsize(0.059)}$ \\[4pt]

\rowcolor{lightgray}
\multicolumn{5}{c}{Sample size: $n = 100$} \\
\multirow{2}{*}{$0$}
& $0.3$ & $0.131\!\,\text{\scriptsize(0.164)}$ & $[-1.821\!\,\text{\scriptsize(0.183)},\; 2.055\!\,\text{\scriptsize(0.146)}]$ & $0.88\!\,\text{\scriptsize(0.046)}$ \\
& $0.1$ & $0.006\!\,\text{\scriptsize(0.166)}$ & $[-1.943\!\,\text{\scriptsize(0.163)},\; 1.953\!\,\text{\scriptsize(0.174)}]$ & $0.88\!\,\text{\scriptsize(0.046)}$ \\

\rowcolor{lightviolet}
\multicolumn{5}{c}{Sample size: $n = 2500$} \\
\multirow{2}{*}{$0$}
& $0.3$ & $0.066\!\,\text{\scriptsize(0.061)}$ & $[-0.471\!\,\text{\scriptsize(0.061)},\; 0.616\!\,\text{\scriptsize(0.066)}]$ & $0.90\!\,\text{\scriptsize(0.043)}$ \\
& $0.1$ & $-0.057\!\,\text{\scriptsize(0.042)}$ & $[-0.515\!\,\text{\scriptsize(0.049)},\; 0.402\!\,\text{\scriptsize(0.042)}]$ & $0.80\!\,\text{\scriptsize(0.057)}$ \\[2pt]
\hline
\noalign{\vskip 2pt}
$\bm{\Sigma}_{23}^{(0)}$ & \multicolumn{4}{c}{\textbf{$\bm{\Sigma}_{23}$}} \\[2pt]

\rowcolor{lightgray}
\multicolumn{5}{c}{Sample size: $n = 100$} \\
\multirow{2}{*}{$2.25$}
& $0.3$ & $1.536\!\,\text{\scriptsize(0.080)}$ & $[0.787\!\,\text{\scriptsize(0.052)},\; 2.761\!\,\text{\scriptsize(0.126)}]$ & $0.68\!\,\text{\scriptsize(0.067)}$ \\
& $0.1$ & $2.328\!\,\text{\scriptsize(0.092)}$ & $[1.368\!\,\text{\scriptsize(0.057)},\; 3.808\!\,\text{\scriptsize(0.161)}]$ & $0.92\!\,\text{\scriptsize(0.039)}$ \\

\rowcolor{lightviolet}
\multicolumn{5}{c}{Sample size: $n = 2500$} \\
\multirow{2}{*}{$2.25$}
& $0.3$ & $1.515\!\,\text{\scriptsize(0.036)}$ & $[1.179\!\,\text{\scriptsize(0.031)},\; 2.031\!\,\text{\scriptsize(0.036)}]$ & $0.64\!\,\text{\scriptsize(0.050)}$ \\
& $0.1$ & $2.021\!\,\text{\scriptsize(0.035)}$ & $[1.689\!\,\text{\scriptsize(0.029)},\; 2.421\!\,\text{\scriptsize(0.043)}]$ & $0.76\!\,\text{\scriptsize(0.061)}$ \\[4pt]

\rowcolor{lightgray}
\multicolumn{5}{c}{Sample size: $n = 100$} \\
\multirow{2}{*}{$0$}
& $0.3$ & $0.033\!\,\text{\scriptsize(0.062)}$ & $[-0.649\!\,\text{\scriptsize(0.078)},\; 0.751\!\,\text{\scriptsize(0.078)}]$ & $0.82\!\,\text{\scriptsize(0.055)}$ \\
& $0.1$ & $0.086\!\,\text{\scriptsize(0.060)}$ & $[-0.669\!\,\text{\scriptsize(0.073)},\; 0.858\!\,\text{\scriptsize(0.070)}]$ & $0.86\!\,\text{\scriptsize(0.050)}$ \\

\rowcolor{lightviolet}
\multicolumn{5}{c}{Sample size: $n = 2500$} \\
\multirow{2}{*}{$0$}
& $0.3$ & $0.004\!\,\text{\scriptsize(0.020)}$ & $[-0.219\!\,\text{\scriptsize(0.022)},\; 0.220\!\,\text{\scriptsize(0.022)}]$ & $0.92\!\,\text{\scriptsize(0.039)}$ \\
& $0.1$ & $-0.004\!\,\text{\scriptsize(0.018)}$ & $[-0.211\!\,\text{\scriptsize(0.019)},\; 0.192\!\,\text{\scriptsize(0.021)}]$ & $0.94\!\,\text{\scriptsize(0.034)}$ \\

\hline
\end{tabular}
\end{table}
\clearpage
The advantages of joint modeling are most evident when cross-dependence is present. In this setting, we observe clear reductions in posterior variability of the regression coefficients $\bm{B}$ across all response types, as illustrated by the boxplots in the left panels of Figure~\ref{fig:beta_var_all}, along with improved predictive performance measured by ELJPD (top rows for each model, Table~\ref{tab:elpd_diff_all}). In contrast, when responses are effectively independent, joint and separate models exhibit comparable estimation (right panel, Figure~\ref{fig:beta_var_all}) and predictive performance (bottom rows for each model, Table~\ref{tab:elpd_diff_all}), indicating robustness of the proposed framework to the underlying dependence structure. We refer the reader to Section~\ref{sec:gp_simulations} for a detailed discussion on results obtained for the ``Gaussian-Poisson'' model and Section~\ref{sec:phi_summaries} for an insightful analysis of the Mat\'{e}rn range parameter, $\phi$, corresponding to all the models in the main article. 
\begin{figure}[]
    \centering
    \includegraphics[width=0.9\linewidth]{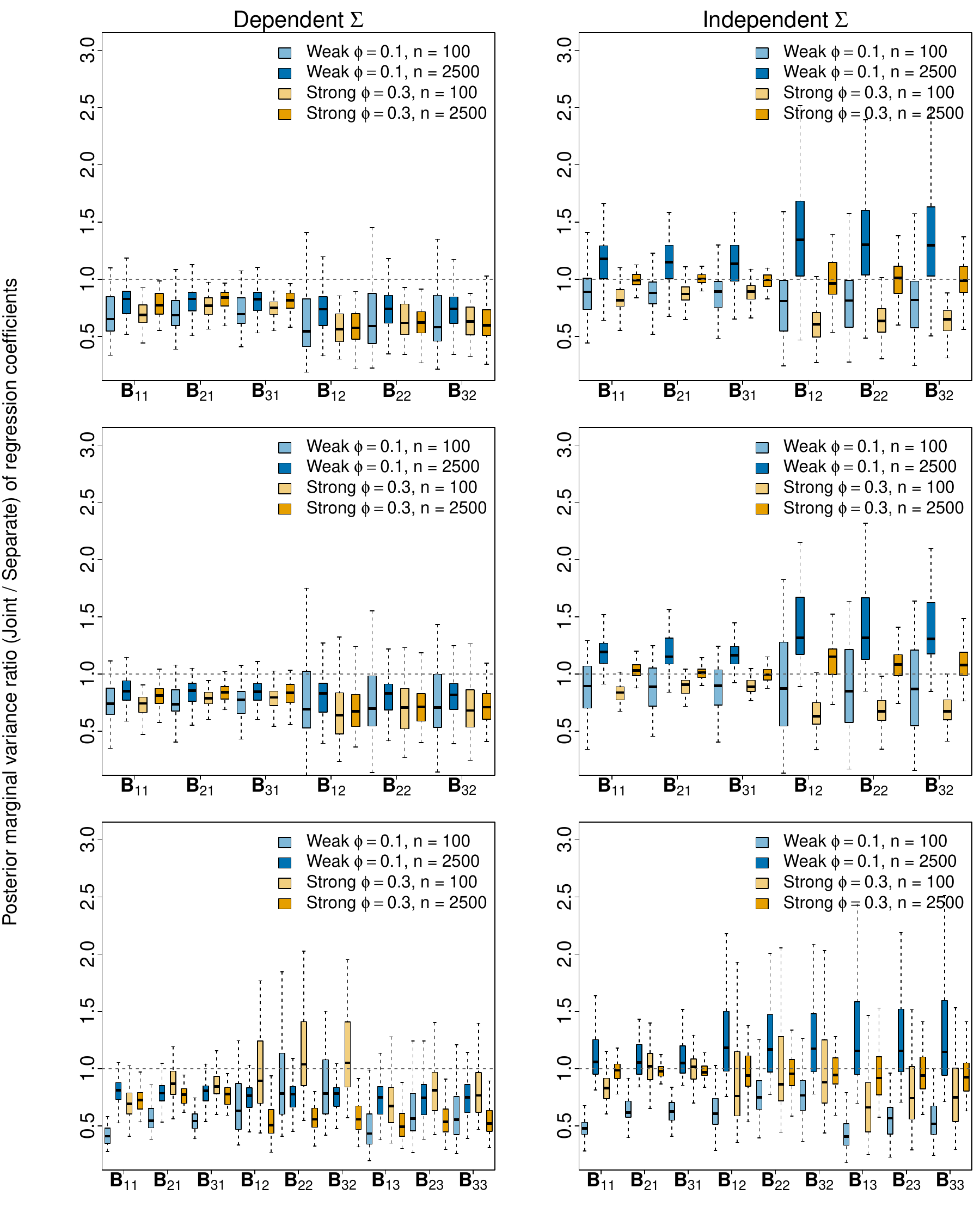}
    \caption{Efficacy in variance reduction in regression coefficient matrix $\bm{B}$ components in the joint model in comparison to the separate model for the Binomial-Gaussian model (top), Binomial-Poisson (middle), and Binomial-Gaussian-Poisson (bottom) case study on 50 replicated datasets.}
    \label{fig:beta_var_all}
\end{figure}
\begin{table}[!h]
\centering
\caption{ELJPD differences of Joint and Separate models across 50 replications for Binomial-Gaussian (top block), Binomial-Poisson (middle block), and Binomial-Gaussian-Poisson (bottom block) response types. Here, $\textrm{vec}(\bm{\Sigma}^{(0)})$ denotes the vectorized form of the true cross-covariance matrix $\bm{\Sigma}^{(0)}$.}
\label{tab:elpd_diff_all}
\renewcommand{\arraystretch}{1.1}
\small
\begin{tabular}{c c c c c}
\hline
\text{Model} & $\textrm{vec}(\bm{\Sigma}^{(0)})$ & $\phi_0$ & \cellcolor{lightgray}{$n = 100$} & \cellcolor{lightviolet}{$n = 2500$} \\ [2pt]
\hline

\multirow{4}{*}{\shortstack{Binomial\\Gaussian}}
& \multirow{2}{*}{$[9,5,5,3]^{\top}$} 
& $0.3$ & $\hphantom{-}1.485\!\,\text{\scriptsize(0.304)}$ & $\hphantom{-}2.364\!\,\text{\scriptsize(0.962)}$ \\
&  & $0.1$ & $\hphantom{-}3.263\!\,\text{\scriptsize(0.519)}$ & $\hphantom{-}9.103\!\,\text{\scriptsize(2.185)}$ \\[2pt]
\cline{2-5}

& \multirow{2}{*}{$[9,0,0,3]^{\top}$} 
& $0.3$ & $-0.577\!\,\text{\scriptsize(0.213)}$ & $-0.350\!\,\text{\scriptsize(0.948)}$ \\
&  & $0.1$ & $\hphantom{-}0.149\!\,\text{\scriptsize(0.430)}$ & $-2.732\!\,\text{\scriptsize(2.379)}$ \\

\hline

\multirow{4}{*}{\shortstack{Binomial\\Poisson}}
& \multirow{2}{*}{$[9,4,4,2]^{\top}$} 
& $0.3$ & $\hphantom{-}1.587\!\,\text{\scriptsize(0.362)}$ & $\hphantom{-}4.135\!\,\text{\scriptsize(3.367)}$ \\
&  & $0.1$ & $\hphantom{-}3.219\!\,\text{\scriptsize(0.673)}$ & $\hphantom{-}14.366\!\,\text{\scriptsize(3.663)}$ \\
\cline{2-5}

& \multirow{2}{*}{$[9,0,0,2]^{\top}$} 
& $0.3$ & $-0.558\!\,\text{\scriptsize(0.366)}$ & $-3.458\!\,\text{\scriptsize(2.172)}$ \\
&  & $0.1$ & $-1.480\!\,\text{\scriptsize(0.765)}$ & $\hphantom{-}8.680\!\,\text{\scriptsize(4.041)}$ \\

\hline

\multirow{4}{*}{\shortstack{Binomial\\Gaussian\\Poisson}}
& \multirow{2}{*}{$[9, 5, 4, 5, 3, \tfrac{9}{4}, 4, \tfrac{9}{4}, 2]^{\top}$} 
& $0.3$ & $\hphantom{-}1.742\!\,\text{\scriptsize(0.748)}$ & $\hphantom{-}13.054\!\,\text{\scriptsize(2.921)}$ \\
&  & $0.1$ & $\hphantom{-}2.519\!\,\text{\scriptsize(1.275)}$ & $\hphantom{-}32.948\!\,\text{\scriptsize(5.226)}$ \\[2pt]
\cline{2-5}

& \multirow{2}{*}{$[9, 0, 0, 0, 3, 0, 0, 0, 2]^{\top}$} 
& $0.3$ & $-5.548\!\,\text{\scriptsize(0.723)}$ & $\hphantom{-}3.704\!\,\text{\scriptsize(2.078)}$ \\
&  & $0.1$ & $-11.022\!\,\text{\scriptsize(1.334)}$ & $\hphantom{-}5.601\!\,\text{\scriptsize(5.367)}$ \\

\hline
\end{tabular}
\end{table}

\subsection{Real data analysis}\label{sec:real_data_analysis}

Wildfires present significant risks to ecosystems, human life, and infrastructure, with far-reaching social and economic consequences. They are also a substantial source of carbon dioxide, contributing to the intensification of the global greenhouse effect. In the United States, the area affected by wildfires has increased nearly fourfold over the past 40 years \citep{iglesias2022us}, resulting in a sharp rise in federal spending on fire control efforts. These trends underscore the need for flexible statistical methods to accurately predict extreme wildfire events across space, an essential component of fire management that informs resource allocation, risk mitigation, and recovery planning.

 We analyze gridded wildfire data comprising aggregated monthly counts of wildfire occurrences ($\mathrm{CNT}$), and corresponding burnt areas ($\mathrm{BA}$), within each cell of a regular grid spanning the mainland United States. The dataset, provided for the Extreme Value Analysis (EVA) 2021 data challenge \citep{Optitz23}, includes monthly observations from 3,503 grid cells over the mainland United States at a spatial resolution of $0.5^\circ \times 0.5^\circ$. It covers a 23-year period (1993–2015) and records data for seven months each year (March through September). Despite the discrete nature of wildfire counts (CNTs), much of the existing literature has relied on transformation-based methods to facilitate joint modeling using a bivariate GP. For instance, \cite{cisneros2023combined} modeled $\log(1 + \textrm{CNT}(\bm{s}))$ combining random forest with bivariate GP. While several alternative approaches, such as multi-stage models, have been developed to accommodate zero inflation \citep{zhang2023joint}, we do not aim to compare our model against those. Nonetheless, we acknowledge that more tailored models may improve both fitting and predictive performance on this dataset. Our primary inferential goal is to evaluate the performance of our proposed multivariate process model relative to conventional separate modeling approaches, detailed in Section~\ref{sec:sep_model}, outlined as follows:
\begin{align*}\label{competing_models}
\mc{M}_1 &: \text{Joint Gaussian-Poisson model: } 
\big(\log(1+\mathrm{BA}(\bm{s})), \mathrm{CNT}(\bm{s}) \big) 
\; \text{with}\; \eqref{our_model},  \\
\mc{M}_2 &: \text{Separate Gaussian-Poisson model: }  \log(1+\mathrm{BA}(\bm{s})) \;\text{and} \;
\mathrm{CNT}(\bm{s}) \; \text{with}\; \eqref{eq:sep_model}.
\end{align*}
For illustration, we focus on time index 140, corresponding to September 2012, which aligns well with our model assumptions regarding spatial range and cross-covariance separability. A detailed statistical summary, including the number of fires and acres burned, is provided in the US monthly climate reports\footnote{Available at \url{https://www.ncei.noaa.gov}.}, which state that ``fire seasons'' mostly span from July to September.
\begin{figure}[ht]
    \centering
    \centering
    \begin{subfigure}[t]{0.49\textwidth}
        \centering
        \includegraphics[width=\linewidth]{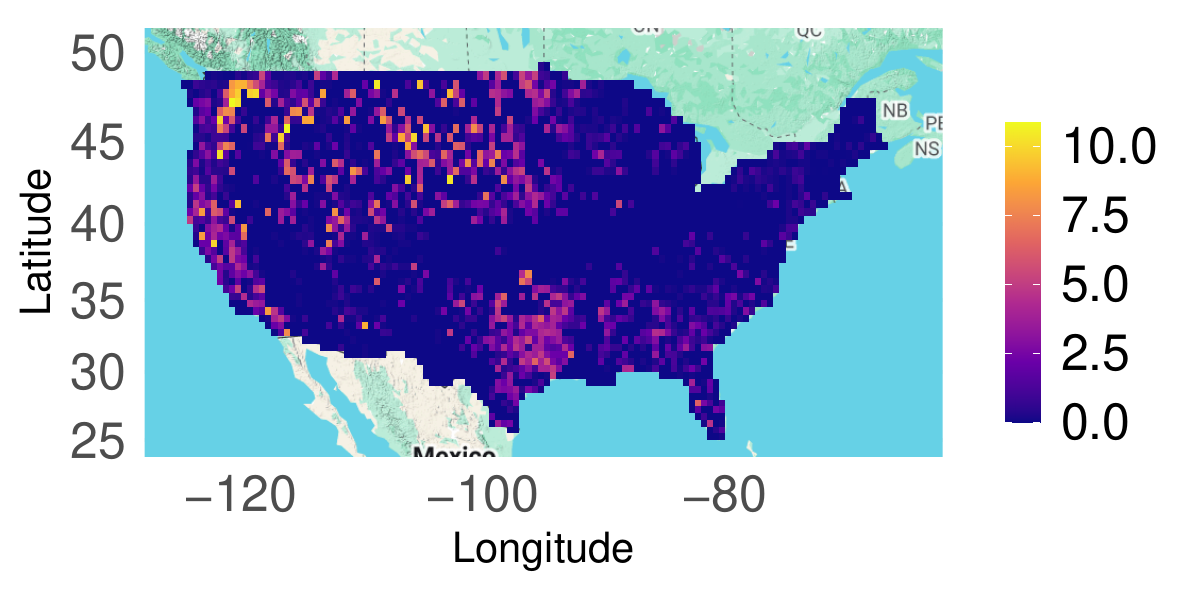}
    \end{subfigure}
    \begin{subfigure}[t]{0.49\textwidth}
        \centering
        \includegraphics[width=\linewidth]{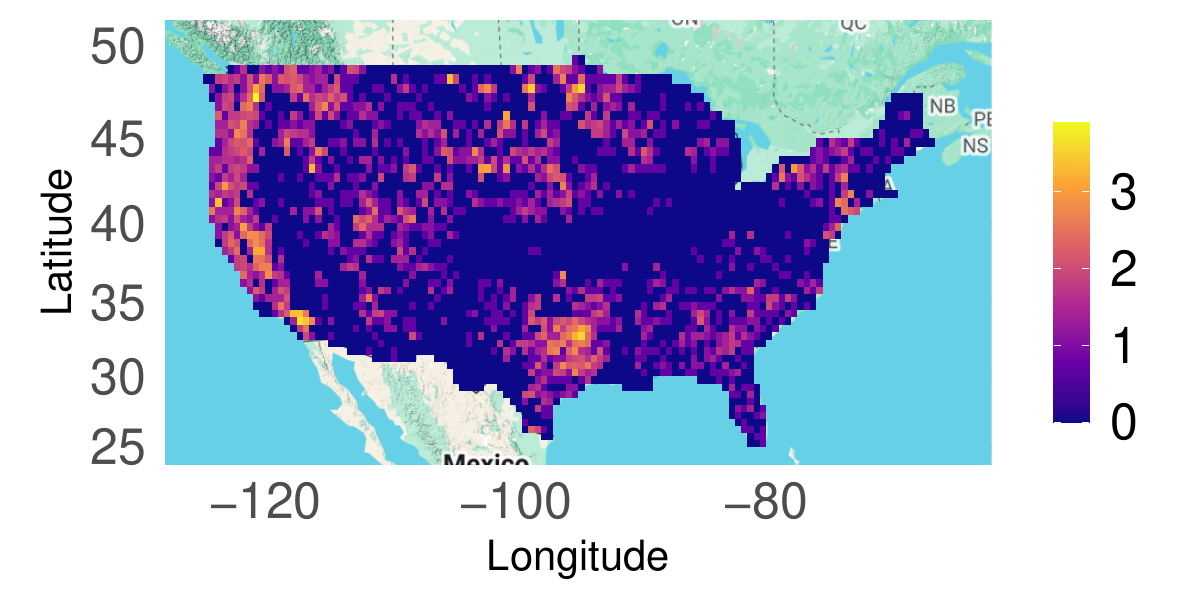}
    \end{subfigure}
     \caption{Spatial maps of $\log{(1+ \textrm{BA})}$ (left panel) and $\log{(1+ \textrm{CNT})}$ (right panel) of September 2012.}
    \label{fig:spatial map}
\end{figure}
The spatial maps in Figure~\ref{fig:spatial map} show that, with high correlations, the responses are spatially localized and concentrated within the domain. 
Several meteorological and land cover variables can be potentially used as covariates if they are useful and contribute significantly to our data analysis. However, including them in the model can be computationally challenging. We have compared the adjusted-$R^2$ by fitting a simple linear regression model with $\log(1 + \textrm{BA}(\bm{s}))$ and $\log(1 + \textrm{CNT}(\bm{s}))$ on the all the covariates and with latitude and longitude as covariates and obtained a mere increase of 0.32 from 0.28 from the later. \cite{cisneros2023combined} study the importance of covariates in their model and demonstrate that predicting $\textrm{BA}$, no other variables except $\textrm{CNT}$ have a significant effect, and vice versa. Hence, we are not including any other variables, except geographical coordinates, as covariates in our analysis. To assess spatial dependence, we plot empirical semivariograms of residuals obtained by separately regressing $\log(1 + \textrm{BA}(\bm{s}))$ and $\log(1 + \textrm{CNT}(\bm{s}))$ on the covariates $[1, \textrm{lon}(\bm{s}), \textrm{lat}(\bm{s})]^{\top}$. The estimated range parameters $\phi$ for burnt area and count responses are $2.622$ and $2.646$, respectively, values that are reasonably consistent with the separability assumption. We fix the smoothness parameter $\nu = 0.5$, as suggested by the fitted variograms, for subsequent analysis.

\begin{figure}[ht]
    \centering
    \includegraphics[width=1\linewidth]{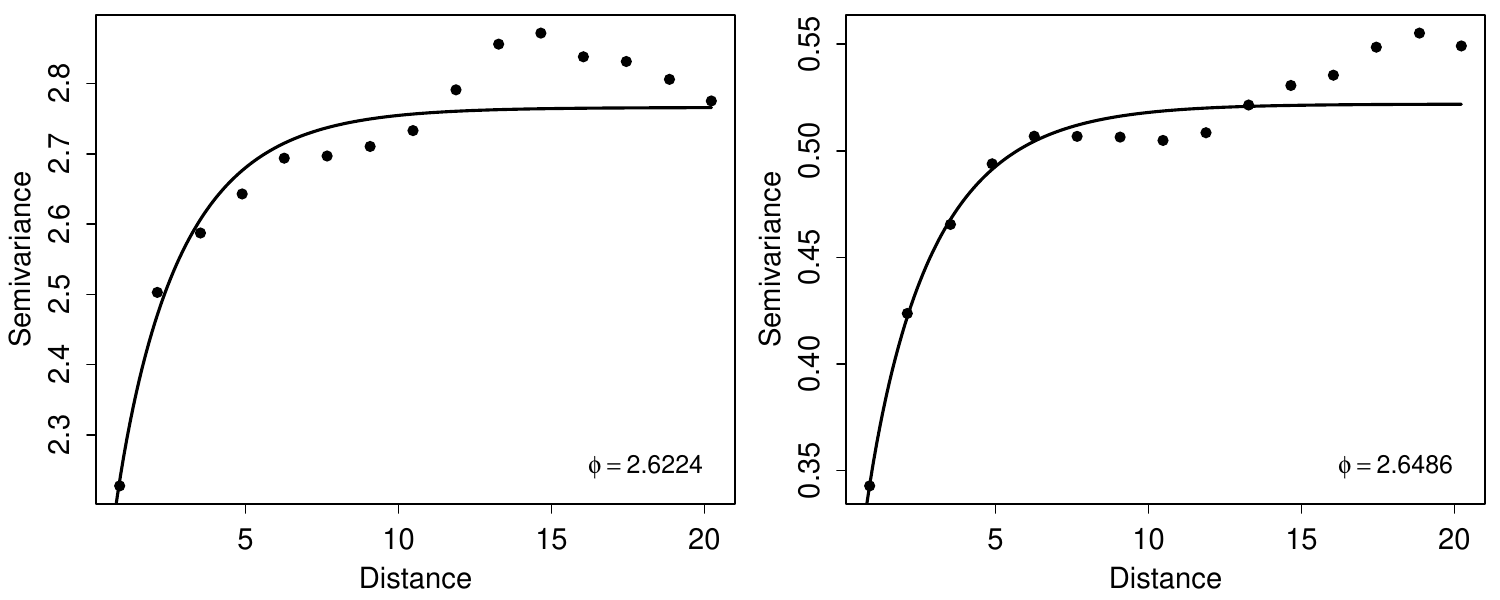}
    \caption{Semivariogram of residuals $\log{(1+ \textrm{BA})}$ (left panel) and $\log{(1+ \textrm{CNT})}$ (right panel) regression on locations as covariates for the dataset in September 2012.}
    \label{fig:semivariogram}
\end{figure}
The spatial variability pattern of the residuals of $\log(1 + \textrm{BA}(\bm{s}))$ and $\log(1 + \textrm{CNT}(\bm{s}))$ in Figure~\ref{fig:semivariogram} clearly illustrates that our proposed model may be a good fit. We present the estimation summary of real data analysis of our model $\mc{M}_1$ and the separate model $\mc{M}_2$ in Table~\ref{tab:real_data_estimation}. From Table~\ref{tab:real_data_estimation}, the posterior credible intervals at the $95\%$ level indicate that all the covariates are meaningful for the Gaussian response $\log(1 + \textrm{BA})$, while only intercept is statistically significant for $\log(1 + \textrm{CNT})$ under both $\mc{M}_1$ and $\mc{M}_2$. This finding indeed suggests that these covariates play an important role, irrespective of whether the responses are modeled jointly or separately. 
The posterior summaries of the covariance matrix $\bm{\Sigma}$ under $\mc{M}_1$ indicate a strong relationship between the two response types. In particular, the $95\%$ credible interval for the cross-covariance parameter $\bm{\Sigma}_{12}$ excludes zero, suggesting that the dependence is statistically significant. Such dependence is explicitly captured by the joint model $\mc{M}_1$, whereas it is inherently disregarded under the separate model $\mc{M}_2$. These findings highlight the practical advantage of joint modeling in the presence of dependence among bivariate ``continuous-count'' mixed-type responses by our approach. We also note that inference on the spatial range parameter $\phi$ is influenced by the truncated prior, resulting in estimates that differ from those obtained from empirical semivariogram analysis of residuals from a fitted linear model.
\begin{table}[!h]
\centering
\caption{Posterior means of the model parameters along with their standard errors for $\mc{M}_1$ and $\mc{M}_2$, based on using 10-fold cross-validation.}
\label{tab:real_data_estimation}
\begin{tabular}{lll}
  \hline
Parameters & $\mc{M}_1$ & $\mc{M}_2$ \\ 
  \hline
$\bm{B}_{11}$ & $\mathbf{-1.893}\!\;\text{\scriptsize($-3.9752$,\,$\hphantom{-}0.2212$)}$ & $\mathbf{-2.165}\!\;\text{\scriptsize($-3.0555$,\,$-1.2740$)}$ \\ 
  $\bm{B}_{12}$ & $\mathbf{-3.717}\!\;\text{\scriptsize($-5.9603$,\,$-1.4834$)}$ & $\mathbf{-4.100}\!\;\text{\scriptsize($-5.0956$,\,$-3.1047$)}$ \\ 
  $\bm{B}_{13}$ & $\mathbf{-0.021}\!\;\text{\scriptsize($-0.0372$,\,$-0.0045$)}$ & $\mathbf{-0.023}\!\;\text{\scriptsize($-0.0299$,\,$-0.0163$)}$ \\ 
  $\bm{B}_{21}$ & $\mathbf{-0.022}\!\;\text{\scriptsize($-0.0398$,\,$-0.0045$)}$ & $\mathbf{-0.027}\!\;\text{\scriptsize($-0.0349$,\,$-0.0194$)}$ \\ 
  $\bm{B}_{22}$ & $\hphantom{-}0.015\!\;\text{\scriptsize($-0.0242$,\,$\hphantom{-}0.0532$)}$ & $\hphantom{-}0.018\!\;\text{\scriptsize($\hphantom{-}0.0007$,\,$\hphantom{-}0.0349$)}$ \\ 
  $\bm{B}_{23}$ & $\hphantom{-}0.017\!\;\text{\scriptsize($-0.0245$,\,$\hphantom{-}0.0588$)}$ & $\hphantom{-}0.017\!\;\text{\scriptsize($-0.0019$,\,$\hphantom{-}0.0360$)}$ \\ 
  $\bm{\Sigma}_{11}$ & $\hphantom{-}1.919\!\;\text{\scriptsize($\hphantom{-}1.6114$,\,$\hphantom{-}2.3752$)}$ & $\hphantom{-}1.385\!\;\text{\scriptsize($\hphantom{-}1.2594$,\,$\hphantom{-}1.5276$)}$ \\ 
  $\bm{\Sigma}_{12}$ & $\hphantom{-}\mathbf{1.636}\!\;\text{\scriptsize($\hphantom{-}1.4135$,\,$\hphantom{-}1.9110$)}$ & -- (--) \\ 
  $\bm{\Sigma}_{22}$ & $\hphantom{-}2.135\!\;\text{\scriptsize($\hphantom{-}1.7275$,\,$\hphantom{-}2.5284$)}$ & $\hphantom{-}1.630\!\;\text{\scriptsize($\hphantom{-}1.4037$,\,$\hphantom{-}1.8168$)}$ \\ 
  $\phi$ & $\hphantom{-}1.143\!\;\text{\scriptsize($\hphantom{-}0.9435$,\,$\hphantom{-}1.3734$)}$ & $\hphantom{-}0.446\!\;\text{\scriptsize($\hphantom{-}0.4116$,\,$\hphantom{-}0.4934$)}$ \\ 
   \hline
\end{tabular}
\end{table}

\begin{table}[!h]
\centering
\caption{Performances of $\mc{M}_1$ and $\mc{M}_2$ based on ELJPD and the coverage for both components over all locations in the test sets from 10-fold cross-validation.} 
\label{tab:real_data_prediction}
\begin{tabular}{lll}
  \hline
Model & ELJPD \text{\scriptsize(SE)} & Coverage \text{\scriptsize(SE)} \\ 
  \hline
$\mc{M}_1$ & $\mathbf{-1355.146}\!\,\text{\scriptsize(30.759)}$ & $0.950\!\,\text{\scriptsize(0.002)}$ \\ 
  $\mc{M}_2$ & $-1514.617\!\,\text{\scriptsize(30.413)}$ & $0.964\!\,\text{\scriptsize(0.002)}$ \\ 
   \hline
\end{tabular}
\end{table}
We evaluate the predictive performance using $10$-fold cross-validation on randomly partitioned data. While both models show good prediction coverage (Table~\ref{tab:real_data_prediction}), the results indicate that $\mc{M}_1$ dominates $\mc{M}_2$ with higher ELJPD, indicating improved out-of-sample prediction across $10$-fold cross-validation. This finding suggests that accounting for cross-response association enhances predictive accuracy beyond separate modeling. In the context of wildfire management, this can support more efficient resource allocation across and within states.

\section{Discussions}\label{sec:discussions}
We propose a novel joint Bayesian hierarchical model for high-dimensional spatial data with mixed-type responses. This framework captures the cross-dependence among different response types and the spatial dependence across the domain. To our knowledge, this is the first model to offer flexibility, interpretability, and computational scalability.
Unlike the \texttt{GPVecchia} package in \texttt{R} \citep{katzfuss2021general}, which does not extend to matrix-Gaussian approximation, our method can be used
by modeling a sparsity-aware approximation of a multivariate latent spatial random effect with a separable covariance structure. Our approach enables fully Bayesian inference in high-dimensional spatial settings, making it particularly well-suited for modeling spatially indexed point process data over large and complex geographical regions. In contrast, traditional univariate process models fail to account for cross-response correlations, as demonstrated in our simulation-based predictive analysis. To ensure computational efficiency, we impose a sparse structure on the Cholesky factor of the precision matrix via the Vecchia approximation. Additionally, we incorporate a component-wise elliptical slice sampler for the latent GP and a blocked Gibbs sampler for the regression and cross-covariance matrices, which improves the mixing and convergence of the MCMC algorithm. However, there are methodologies for mixed-type response non-Gaussian spatial data \citep{zhang2025bayesian} and spatio-temporal data \citep{pan2024bayesian} that bypass the MCMC algorithm, providing faster methods based on Bayesian predictive stacking. Another domain of interest could be analyzing model selection for mixed-type multivariate spatial data, as in \cite{ghosh2025high}, a recent work in a non-spatial setup.

We have also established key theoretical properties of the model, including identifiability, a known challenge in non-replicated multivariate spatial settings. In our real-data application, we analyze US wildfire data over $3503$ spatial grid cells covering mainland US, demonstrating that our model consistently outperforms independent univariate process models. This superiority is evident through standard Bayesian model comparison metrics such as ELJPD. The EVA 2021 data challenge, however, involves spatio-temporal data with multiple land cover and meteorological variables; several avenues can also be explored to extend our model to temporally replicated spatial data \citep{Zhu05}, spatio-temporal data \citep{zhang2023joint}, or in Bayesian model selection.

While our current formulation assumes spatial isotropy and separability, violations from these assumptions may lead to model misspecification. Extending the framework for nonstationary spatial dependence remains a substantial challenge. Flexible cross-covariance can be used, for instance, through linear models of coregionalization (LMC), which express $\mc{W}(\cdot)$ as a linear combination of latent spatial processes \citep{gelfand2004nonstationary, schmidt2003bayesian}, as well as through more general constructions based on latent dimensions \citep{apanasovich2010cross}. However, such non-separable covariance formulations often introduce over-parameterization and raise identifiability concerns \citep{Genton15}, particularly in the absence of temporally replicated data. For areal or lattice spatial data, an alternative direction can be explored using multivariate conditional autoregressive (MCAR) models \citep{mardia1988multi, gelfand2003proper}. Future work could explore scalable approximations to the high-dimensional posterior of the latent GP, such as domain partitioning or approximate Bayesian fusion \citep{dai2023bayesian} within the mixed-type response modeling framework. 



\bibliography{reference_paper}


\clearpage

\appendix

\setcounter{equation}{0}
\setcounter{figure}{0}
\setcounter{table}{0}
\setcounter{algorithm}{0}

\section*{Supplementary Materials}
\addcontentsline{toc}{section}{Supplementary Materials}
\renewcommand{\thesection}{S\arabic{section}}
\renewcommand{\thesubsection}{S\arabic{section}.\arabic{subsection}}

\addcontentsline{toc}{section}{Supplementary Materials}

\renewcommand{\theequation}{E\S\arabic{equation}}   
\renewcommand{\thefigure}{F\S\arabic{figure}}      
\renewcommand{\thetable}{T\S\arabic{table}}        
\renewcommand{\thealgorithm}{A\S\arabic{algorithm}}

\section{Exponential family distributions}
In Section~\ref{sec:data_analysis} of the main article, we consider models with Gaussian, Binomial, and Poisson response types. However, our spatial mixed-type model accommodates other well-known members of the exponential family. Table~\ref{tab:exp_family_distributions} presents a list of well-known distributions along with their corresponding parameters.
\begin{table}[!h]
\centering
\caption{Exponential family distributions with canonical link functions}
\label{tab:exp_family_distributions}
\renewcommand{\arraystretch}{1}
\setlength{\tabcolsep}{3pt}
\begin{tabular}{llcccc}
\toprule
\vspace{0.2cm}
\text{Distribution} & \text{Params} & \text{Canonical Link} & $\psi$ & $b(w)$ & $h(y, \psi)$ \\
\midrule
\vspace{0.2cm}
\text{Gaussian} & $\mu, \sigma^2$ & $w = \mu$ & $\sigma^2$ & $\dfrac{w^2}{2}$ & $\dfrac{1}{\sqrt{2\pi\sigma^2}} \exp\!\left(-\dfrac{y^2}{2\sigma^2}\right)$ \\
\vspace{0.2cm}
\text{Binomial} & $p$ & $w = \log\!\left(\dfrac{p}{1 - p}\right)$ & $1$ & $\log(1 + e^{w})$ & $1$ \\
\vspace{0.2cm}
\text{Poisson} & $\lambda$ & $w = \log(\lambda)$ & $1$ & $e^{w}$ & $\dfrac{1}{y!}$ \\
\vspace{0.2cm}
\text{Binomial} & $p, m$ & $w = \log\!\left(\dfrac{p}{1 - p}\right)$ & $1$ & $m \log(1 + e^{w})$ & $\binom{m}{y}$ \\
\vspace{0.2cm}
\text{Gamma} & $\theta, \alpha$ & $w = \dfrac{1}{\theta}$ & $\alpha$ & $-\log(-w)$ & $\dfrac{y^{\alpha - 1}}{\Gamma(\alpha)} \, \mathbf{1}_{\{y>0\}}$ \\
\vspace{0.2cm}
Neg-Binomial & $p, r$ & $w = \log\!\left(\dfrac{p}{1 - p}\right)$ & $r$ & $-r \log(1 - e^{w})$ & $\binom{y+r-1}{y}$ \\
\bottomrule
\end{tabular}
\end{table}

 Although Table~\ref{tab:exp_family_distributions} presents canonical link functions, our model is not limited to these specifications. In particular, it can accommodate alternative link functions, such as the probit and complementary log-log links for Binomial responses and the log link for Gamma models.
\section{Matrix-Normal Vecchia approximation}\label{sec:vecchia_matnorm}

In this section, we outline the definition of the Matrix-Normal distribution and a scalable sampling procedure based on the Vecchia approximation, used in inference and prediction in Algorithm~\ref{alg:mixed_model_algorithm} in the main paper.

\begin{definition}
    A random matrix $\bm{Z}_{d_1 \times d_2}$ is said to follow a Matrix-Normal distribution with mean matrix $\bm{M}$, row-wise covariance being $\bm{V}$ and column-wise covariance being $\bm{U}$ and denoted by $\bm{Z} \sim \mc{MN}_{d_1, d_2}(\bm{M}, \bm{U}, \bm{V})$ if its probability density function is given by
\begin{equation*}
    \pi(\bm{Z} \mid \bm{M}, \bm{U}, \bm{V}) =
    \dfrac{1}{ (2\pi)^{d_1 d_2/2} \, |\bm{U}|^{d_2/2} \, |\bm{V}|^{d_1/2} }
    \exp\!\Bigg\{ -\dfrac{1}{2}\, \mathrm{tr}\!\big[ \bm{V}^{-1} (\bm{Z}-\bm{M})^{\top} \bm{U}^{-1} (\bm{Z}-\bm{M}) \big] \Bigg\},
\end{equation*}
where $\bm{M}$ is an $d_1 \times d_2$ mean matrix, $\bm{U}$ is an $d_1 \times d_1$ row covariance matrix, and $\bm{V}$ is an $d_2 \times d_2$ column covariance matrix. 
Equivalently, if $\bm{Z} \sim \mc{MN}_{d_1, d_2}(\bm{M}, \bm{U}, \bm{V})$, then 
$\mathrm{vec}(\bm{Z}) \sim \mc{N}_{d_1 d_2}\!\big(\mathrm{vec}(\bm{M}), \bm{V} \otimes \bm{U}\big)$,
where $\otimes$ denotes the Kronecker product.
\end{definition}
 In our model, we assume that $\bm{U}$ is the spatial covariance matrix that grows with the data size $d_1$. Instead of using sparse cholesky factor of the precision matrix $\bm{U}^{-1}$, we let $\texttt{NNarray}$ encode the ordered neighbor sets, with
\[
\mc{M}(i) = \{j < i : j \in \texttt{NNarray}[i,\cdot]\}, \quad |\mc{M}(i)| = m \ll d_1.
\]
For each $i$, define
\[
\bm{A}_i = \bm{K}_{i,\mc{M}(i)} \bm{K}_{\mc{M}(i),\mc{M}(i)}^{-1}, 
\qquad
r_i = \bm{K}_{i,i} - \bm{K}_{i,\mc{M}(i)} \bm{K}_{\mc{M}(i),\mc{M}(i)}^{-1} \bm{K}_{\mc{M}(i),i},
\]
where $\bm{K}$ is the full underlying covariance kernel. Here $\bm{A}_i \in \mathbb{R}^{1 \times |\mc{M}(i)|}$ and $r_i$ is scalar.
\begin{algorithm}[H]
\caption{Vecchia sampling from $\mc{MN}_{d_1,d_2}(\bm{M}, \bm{U}, \bm{V})$}
\label{alg:vecchia_sampling}
\begin{algorithmic}[1]
\Require $\bm{M}$, $\bm{V}$, $(\bm{A}_i, r_i, \texttt{NNarray})_{i=1}^{d_1}$

\State Sample $\bm{Z}_i \overset{\text{i.i.d.}}{\sim} \mc{N}_{d_2}(\bm{0}, \bm{V})$ for $i=1,\ldots,d_1.$

\State Initialize $\bm{W} \gets \bm{0}$.

\For{$i = 1,\ldots,d_1$}
    \State $\mc{M}(i) \gets \{j < i : j \in \texttt{NNarray}[i,\cdot]\}.$ 
    
    \State $\bm{W}_i \gets \bm{A}_i \bm{W}_{\mc{M}(i)} + \sqrt{r_i}\,\bm{Z}_i.$
\EndFor

\State \textbf{Return:} $\bm{W} = \bm{M} + \bm{W}$

\end{algorithmic}
\end{algorithm}
\noindent Under the Vecchia approximation, the joint distribution factorizes as
\[
\bm{W}_i \mid \bm{W}_{\mc{M}(i)} \sim \mc{N}_n\!\big(\bm{A}_i \bm{W}_{\mc{M}(i)}, \; r_i \bm{V}\big), \quad i=1,\ldots,m,
\]
where $(\bm{A}_i, r_i)$ are obtained from local covariance blocks. This approximation yields linear complexity in $m$ while avoiding dense matrix factorizations. We provide efficient Vecchia-based sampling in Algorithm~\ref{alg:vecchia_sampling} and fast likelihood evaluation for the latent layer of our model~\ref{our_model} in Algorithm~\ref{alg:vecchia_likelihood} \citep{guinness2018permutation}, both of which are used in the inference workflow of our proposed methodology.
\begin{algorithm}[H]
\caption{Fast Vecchia log-likelihood evaluation}
\label{alg:vecchia_likelihood}
\begin{algorithmic}[1]
\Require $\bm{W}$, $\bm{M}$, $\bm{V}$, $(\bm{A}_i, r_i, \texttt{NNarray})_{i=1}^{d_1}$

\State Initialize $\log L \gets 0$.

\For{$i = 1,\ldots,d_1$}
    \State $\mc{M}(i) \gets \{j < i : j \in \texttt{NNarray}[i,\cdot]\}.$
    
    \State $\bm{\mu}_i \gets \bm{M}_i + \bm{A}_i (\bm{W}_{\mc{M}(i)} - \bm{M}_{\mc{M}(i)}).$
    
    \State $\bm{e}_i \gets \bm{W}_i - \bm{\mu}_i.$
    
    \State $\log L \gets \log L 
    - \dfrac{1}{2} \Big[ 
    n \log(2\pi r_i) + \log|\bm{V}| 
    + \dfrac{1}{r_i} \bm{e}_i^\top \bm{V}^{-1} \bm{e}_i 
    \Big].$
\EndFor

\State \textbf{Return:} $\log L$
\end{algorithmic}
\end{algorithm}

\section{MCMC computations}

We present the posterior simulation framework for our proposed model, discussing a blocked Gibbs sampler for $(\bm{B}, \bm{\Sigma})$ with a component-wise elliptical slice sampler algorithm for updating the latent spatial random effect $\bm{W}$.
\subsection{Blocked-Gibbs sampler}\label{sec:blocked_Gibbs_MNIW}
We provide a detailed derivation of the steps of the Matrix-Normal Inverse-Wishart (MNIW) Blocked-Gibbs sampler mentioned in Section~\ref{sec:bayesian_computation}. 
The prior density of the random effect matrix ${\bm{W}}$ is given by
\begin{equation*}
    \begin{aligned}[t]
        & \pi({\bm{W}}\;|\; \bm{B}, \bm{\Sigma}, \phi ) \\
        &= (2\pi)^{-nq/2} {\lvert {\bm{K}} \rvert}^{-q/2} {\lvert \bm{\Sigma} \rvert}^{-n/2}  \exp{-\dfrac{1}{2}\textrm{tr}\left[ \bm{\Sigma}^{-1} ({\bm{W}} - \bm{X}\bm{B})^{\top} \bm{K}^{-1}({\bm{W}} - \bm{X}\bm{B})\right]}.
    \end{aligned}
\end{equation*}
Similarly the prior density of $\bm{\Sigma}$ is given by
$$ \pi(\bm{\Sigma}) = \dfrac{1}{2^{vq/2} \Gamma(v/2)} {\lvert \bm{S} \rvert}^{v/2} {\lvert \bm{\Sigma} \rvert}^{-(v+q+1)/2}\exp{-\dfrac{1}{2}\tr{\bm{S}\bm{\Sigma}^{-1}}}.$$
The marginal posterior density of $\bm{\Sigma}\; |\; {\bm{W}}, \phi , \bm{Y}$ is obtained by
\begin{equation*}
    \begin{split}
        \pi(\bm{\Sigma} \;|\; {\bm{W}}, \phi , \bm{Y}) &= \int \pi(\bm{\Sigma}, \bm{B} \;|\; {\bm{W}}, \phi , \bm{Y}) d\bm{B}\\
        &= \int \pi({\bm{W}} \;|\; \bm{B}, \bm{\Sigma}) \pi(\bm{B} \;|\; \bm{M}, \bm{V}, \bm{\Sigma}) \pi(\bm{\Sigma} \;|\; \bm{S}, v).\;\;\\
    \end{split}
\end{equation*}
In the blocked Gibbs sampler, instead of sampling from the full conditional posterior of $\bm{\Sigma}$ given $\bm{B} \mid \bm{W}, \phi, \bm{Y}$, we integrate out $\bm{B}$ and draw from the marginal full conditional posterior density $\pi(\bm{\Sigma}, \bm{B} \mid \bm{W}, \phi, \bm{Y})$. We outline the derivation of $\pi(\bm{\Sigma}, \bm{B} \mid \bm{W}, \phi , \bm{Y})$ as given below
\begin{equation*}
\begin{aligned}
\pi(\bm{\Sigma}, \bm{B} \mid \bm{W}, \phi , \bm{Y}) &\propto 
    |\bm{K}|^{-q/2} |\bm{\Sigma}|^{-n/2}
    \exp\!\left\{-\tfrac{1}{2}\,\tr\!\big[ \bm{\Sigma}^{-1} (\bm{W}-\bm{X}\bm{B})^{\top}\bm{K}^{-1}(\bm{W}-\bm{X}\bm{B}) \big]\right\} \\[4pt]
&\quad \times | \bm{V} |^{-q/2} |\bm{\Sigma}|^{-p/2}
    \exp\!\left\{-\tfrac{1}{2}\,\tr\!\big[ \bm{\Sigma}^{-1} (\bm{B}-\bm{M})^{\top} \bm{V}^{-1} (\bm{B}-\bm{M}) \big]\right\} \\[4pt]
&\quad \times |\bm{\Sigma}|^{-(v+q+1)/2}
    \exp\!\left\{-\tfrac{1}{2}\,\tr\!\big[ \bm{\Sigma}^{-1}\bm{S} \big]\right\}.
\end{aligned}
\end{equation*}
Substituting $\widetilde{\bm{V}} = \big(\bm{X}^{\top}\bm{K}^{-1}\bm{X} + \bm{V}^{-1}\big)^{-1},\;
\widetilde{\bm{M}} = \widetilde{\bm{V}}\big(\bm{X}^{\top}\bm{K}^{-1}\bm{W} + \bm{V}^{-1}\bm{M}\big)$ in the joint full conditional posterior, we get the simplified expression below
\begin{equation}\label{blocked_gibbs}
\begin{aligned}[t]
& \pi(\bm{\Sigma}, \bm{B} \mid \bm{W}, \phi , \bm{Y}) \\
&\propto\; |\bm{\Sigma}|^{-(v+p+q+n+1)/2}
   \exp\!\left\{-\tfrac{1}{2}\,\tr\!\Big[\bm{\Sigma}^{-1}\big(
   \bm{S}+\bm{W}^{\top}\bm{K}^{-1}\bm{W}
   +\bm{M}^{\top}\bm{V}^{-1}\bm{M}
   -\widetilde{\bm{M}}^{\top}\widetilde{\bm{V}}^{-1}\widetilde{\bm{M}} \big)\Big]\right\}\\[4pt]
&\propto\; 
   |\bm{\Sigma}|^{-(v+p+q+n+1)/2}
   \exp\!\left\{-\tfrac{1}{2}\,\tr\!\big[\bm{\Sigma}^{-1}(\bm{B}-\widetilde{\bm{M}})^{\top}
   \widetilde{\bm{V}}^{-1}(\bm{B}-\widetilde{\bm{M}})\big]\right\}.
\end{aligned}
\end{equation}
Integrating out the density of $\bm{B}$ in \eqref{blocked_gibbs} we obtain, 
$$ \bm{\Sigma} \;|\; {\bm{W}}, \phi   \sim \mc{IW}_q(\widetilde{\bm{S}} = \bm{S} + {\bm{W}}^{\top} \bm{K}^{-1} {\bm{W}} + {\bm{M}}^{\top} {\bm{V}}^{-1} {\bm{M}} - \widetilde{\bm{M}}^{\top} \widetilde{\bm{V}}^{-1} \widetilde{\bm{M}} , \widetilde{v} = v + n).$$
The full conditional distribution of $\bm{B} \;|\; \bm{\Sigma}, \phi , {\bm{W}}, \bm{Y}$ is obtained by
\begin{equation}\label{posterior_B}
\begin{aligned}[t]
&\pi(\bm{B} \mid \bm{\Sigma}, \phi , {\bm{W}}, \bm{Y}) 
\\
& \propto\; \pi({\bm{W}} \mid \bm{B}, \bm{\Sigma}) \, \pi(\bm{B} \mid \bm{M}, \bm{V}, \bm{\Sigma})\\
&\propto\; (2\pi)^{-nq/2} \, {\lvert \bm{K} \rvert}^{-q/2} \, {\lvert \bm{\Sigma} \rvert}^{-n/2}  
   \exp\!\left\{-\tfrac{1}{2}\,\tr\!\Big[\bm{\Sigma}^{-1} ({\bm{W}} - \bm{X}\bm{B})^{\top} 
   \bm{K}^{-1}({\bm{W}} - \bm{X}\bm{B})\Big]\right\}\\
&\quad\times (2\pi)^{-pq/2} \, {\lvert \bm{V} \rvert}^{-q/2} \, {\lvert \bm{\Sigma} \rvert}^{-p/2}
   \exp\!\left\{-\tfrac{1}{2}\,\tr\!\Big[\bm{\Sigma}^{-1} (\bm{B}- \bm{M})^{\top} 
   \bm{V}^{-1}(\bm{B} - \bm{M})\Big]\right\}\\
&\propto\; \exp\!\left\{-\tfrac{1}{2}\,\tr\!\Big[ \bm{\Sigma}^{-1} 
   \big(\bm{B}^{\top}\bm{X}^{\top} \bm{K}^{-1} \bm{X}\bm{B} 
   - \bm{B}^{\top}\bm{X}^{\top}\bm{K}^{-1}{\bm{W}} 
   - {\bm{W}}^{\top}\bm{K}^{-1} \bm{X}\bm{B}  
   + {\bm{W}}^{\top} \bm{K}^{-1} {\bm{W}} ) \big)\Big]\right\}\\
&\quad\times 
   \exp\!\left\{-\tfrac{1}{2}\,\tr\!\Big[ \bm{\Sigma}^{-1} 
   \big(\bm{B}^{\top}\bm{V}^{-1}\bm{B} - \bm{B}^{\top}\bm{V}^{-1}\bm{M} 
   - \bm{M}^{\top}\bm{V}^{-1} \bm{B} + \bm{M}^{\top}\bm{V}^{-1}\bm{M}\big)\Big]\right\}\\
&\propto\; \exp\!\left\{-\tfrac{1}{2}\,\tr\!\Big[\bm{\Sigma}^{-1} 
   \big(\bm{B}^{\top} (\bm{X}^{\top} \bm{K}^{-1} \bm{X} + \bm{V}^{-1}) \bm{B} 
   - 2\bm{B}^{\top} (\bm{X}^{\top}\bm{K}^{-1} {\bm{W}} 
   + \bm{V}^{-1} \bm{M})\big)\Big]\right\}\\
&\propto\; \exp\!\left\{-\tfrac{1}{2}\,\tr\!\Big[\bm{\Sigma}^{-1} 
   (\bm{B} - \widetilde{\bm{M}})^{\top} 
   \widetilde{\bm{V}}^{-1} (\bm{B} - \widetilde{\bm{M}})\Big]\right\}.
\end{aligned}
\end{equation}
Given $\bm{\Sigma}, \phi , {\bm{W}}$, we obtain the Matrix-Normal full conditional posterior of $\bm{B}$ in \eqref{posterior_B} as $\bm{B} \;|\; \bm{\Sigma}, \phi , {\bm{W}} \sim \mc{MN}_{p,q} (\widetilde{\bm{M}}, \widetilde{\bm{V}}, \bm{\Sigma} ).$

\subsection{Elliptical Slice Sampler}\label{sec:ess_W}
We outline a scalable sampling algorithm for the latent spatial effect $\bm{W}$ in our model, as shown in Algorithm~\ref{alg:mixed_model_algorithm} in the main article.
\begin{algorithm}[H]
\caption{Elliptical Slice Sampler for column-wise update of $\bm{W}$}
\label{alg:ESS_W_conditional}
\begin{algorithmic}[1]
\Require $\bm{W}$, $\bm{B}$, $\bm{\Sigma}$, $(\bm{A}, \bm{r}, \texttt{NNarray})$, $\bm{Y}$, $\bm{X}$, family

\State Denote: $\bm{\mu}_{\bm{W}} \gets \bm{X}\bm{B}$,\quad $\log L \gets \log L(\bm{W})$

\State Sample $\bm{Z} \sim \mathcal{MN}_{n,q}(\bm{0}, \bm{I}_n, \bm{I}_q)$
\State Assign $\bm{\nu} \gets \bm{0}$.
\For{$i = 1,\ldots,n$}
\State $\bm{\nu}_i \gets \bm{A}_i \bm{\nu}_{\mc{M}(i)} + \sqrt{r_i}\,\bm{Z}_i$, \quad where $\mc{M}(i) = \{j < i : j \in \texttt{NNarray}[i,\cdot]\}$.

\EndFor

\For{$j = 1,\ldots,q$}
  
    \State Let $\mathcal{I}_{-j} = \{1,\ldots,q\}\setminus\{j\}$
\State Compute conditional mean: $\bm{\mu}_j \gets \bm{\mu}_{\bm{W},j} + (\bm{W}_{\cdot,-j} - \bm{\mu}_{\bm{W},-j})\bm{\Sigma}_{-j,-j}^{-1}\bm{\Sigma}_{-j,j}$
\State Compute conditional variance: $\sigma_j^2 \gets \Sigma_{j,j} - \bm{\Sigma}_{j,-j}\bm{\Sigma}_{-j,-j}^{-1}\bm{\Sigma}_{-j,j}$
    
    \State $\bm{W}_{\text{prior},j} \gets \sqrt{\sigma_j^2}\,\bm{\nu}_{\cdot j}$
    
    \State Sample $\gamma \sim \mathcal{U}(0,2\pi)$, 
    \State Set $(\gamma_{\min}, \gamma_{\max}) \gets (\gamma-2\pi, \gamma)$
    \State $\log y \gets \log L + \log u$, \quad $u \sim \mathcal{U}(0,1)$
    
    \Repeat
        \State $\bm{W}_{\cdot j}^{\text{cand}} \gets \bm{\mu}_j + (\bm{W}_{\cdot j}-\bm{\mu}_j)\cos\gamma + \bm{W}_{\text{prior},j}\sin\gamma$
        
        \State Form $\bm{W}^{\text{cand}}$ and compute $\log L_{\text{cand}}$
        
        \If{$\log L_{\text{cand}} > \log y$}
            \State $\bm{W}_{\cdot j} \gets \bm{W}_{\cdot j}^{\text{cand}}$, \quad $\log L \gets \log L_{\text{cand}}$
            \State \textbf{break}
        \Else
            \State Shrink bracket:
            \[
            (\gamma_{\min}, \gamma_{\max}) \gets
            \begin{cases}
            (\gamma, \gamma_{\max}) & \gamma < 0 \\
            (\gamma_{\min}, \gamma) & \gamma \ge 0
            \end{cases}
            \]
            \State $\gamma \sim \mathcal{U}(\gamma_{\min}, \gamma_{\max})$
        \EndIf
    \Until{accepted}
    
\EndFor

\State \Return $\bm{W}$
\end{algorithmic}
\end{algorithm}

\section{Posterior latent predictive process}\label{sec:latent_ppd}

In this section, we derive the posterior predictive distribution of the latent process in the precision parameterization \eqref{precision_parameterization}, consistent with the predictive modeling formulation in Section~\ref{sec:predictive_modeling}. This representation is essential for efficient sampling under the Vecchia approximation. Let $\bm{W}^{*}$ and $\bm{X}^{*}$ denote the stacked latent process and covariate matrices over the prediction locations $\mc{U}$, defined as in Section~\ref{sec:predictive_modeling}.

\begin{theorem}
Assume the joint distribution of the latent process over $\mc{S} \cup \mc{U}$ is given by
\[
\begin{bmatrix}\bm{W}\\\bm{W}^*\end{bmatrix}
\sim \mc{MN}_{n+u,q}\left(\,
\begin{bmatrix}\bm{X}\bm{B}\\\bm{X}^*\bm{B}\end{bmatrix},\;
\bm{K}^{(u+n)}=
\begin{bmatrix}\bm{K}^{(n,n)} & \bm{K}^{(n,u)}\\
\bm{K}^{(u,n)} & \bm{K}^{(u,u)}\end{bmatrix},\;
\bm{\Sigma}\right).
\]
Then the conditional distribution of $\bm{W}^{*}$ given $\bm{W}$ admits the precision-form representation
\[
\bm{W}^{*} \mid \bm{W}, \bm{B}, \bm{\Sigma}, \phi 
\sim \mc{MN}_{u,q}\Big(
\bm{X}^{*}\bm{B} - (\bm{Q}^{(u,u)})^{-1}\bm{Q}^{(u,n)}(\bm{W}-\bm{X}\bm{B}),
(\bm{Q}^{(u,u)})^{-1},
\bm{\Sigma}
\Big),
\]
where $\bm{Q}^{(u+n)} = (\bm{K}^{(u+n)})^{-1}$ is partitioned as in \eqref{precision_parameterization}.
\end{theorem}

\begin{proof}
From standard Gaussian conditioning, we have
\[
\bm{W}^* \mid \bm{W} \sim \mc{MN}_{u,q}\big(\bm{M}_{W*}, \bm{K}_{W*}, \bm{\Sigma}\big),
\]
where
\[
\bm{M}_{W*} = \bm{X}^*\bm{B} + \bm{K}^{(u,n)}(\bm{K}^{(n,n)})^{-1}(\bm{W}-\bm{X}\bm{B}),
\quad
\bm{K}_{W*} = \bm{K}^{(u,u)} - \bm{K}^{(u,n)}(\bm{K}^{(n,n)})^{-1}\bm{K}^{(n,u)}.
\]
Let $\bm{Q}^{(u+n)} = (\bm{K}^{(u+n)})^{-1}$ with block partition as in \eqref{precision_parameterization}. Using the identity
\[
\bm{K}^{(u+n)}\bm{Q}^{(u+n)} = \bm{I}_{n+u},
\]
block multiplication yields
\begin{align}
\bm{K}^{(u,u)}\bm{Q}^{(u,u)} + \bm{K}^{(u,n)}\bm{Q}^{(n,u)} &= \bm{I}_u, \label{eq1}\\
\bm{K}^{(n,u)}\bm{Q}^{(u,u)} + \bm{K}^{(n,n)}\bm{Q}^{(n,u)} &= \bm{0}. \label{eq2}
\end{align}
From \eqref{eq2}, assuming $\bm{K}^{(n,n)}$ is invertible, we have $\bm{Q}^{(n,u)} = -(\bm{K}^{(n,n)})^{-1}\bm{K}^{(n,u)}\bm{Q}^{(u,u)}.$
Substituting into \eqref{eq1} gives $\left(\bm{K}^{(u,u)} - \bm{K}^{(u,n)}(\bm{K}^{(n,n)})^{-1}\bm{K}^{(n,u)}\right)\bm{Q}^{(u,u)} = \bm{I}_u,$
which implies
\[
\bm{K}_{W*} = \bm{K}^{(u,u)} - \bm{K}^{(u,n)}(\bm{K}^{(n,n)})^{-1}\bm{K}^{(n,u)}
= (\bm{Q}^{(u,u)})^{-1}.
\]
Next, using another block identity, we obtain $\bm{Q}^{(u,n)} = -(\bm{Q}^{(u,u)})\,\bm{K}^{(u,n)}(\bm{K}^{(n,n)})^{-1},$
which implies
\[
\bm{K}^{(u,n)}(\bm{K}^{(n,n)})^{-1}
= -(\bm{Q}^{(u,u)})^{-1}\bm{Q}^{(u,n)}.
\]
Substituting into $\bm{M}_{W*}$ yields
\[
\bm{M}_{W*}
= \bm{X}^*\bm{B} - (\bm{Q}^{(u,u)})^{-1}\bm{Q}^{(u,n)}(\bm{W}-\bm{X}\bm{B}).
\]
Combining these expressions establishes the precision-form representation of the predictive distribution.
\end{proof}

\section{Separate model}\label{sec:sep_model}
We specify the traditional mixed-type separate model, in which, for each response type, a practitioner fits an independent latent univariate Gaussian process to each margin. Under our notations in model specification, with a shared parameter $\phi$ for each response type, the hierarchical structure for a separate model is given by
\begin{equation}\label{eq:sep_model}
\begin{aligned}
\text{Data level:} \quad
& Y_j(\bm{s}) \mid W_j(\bm{s}) \;\overset{\mathrm{ind}}{\sim}\; \mathrm{EF}\big(Y_j(\bm{s}) \mid W_j(\bm{s}), \psi_j \big), 
\quad j = 1, \ldots, q,\quad \bm{s} \in \mc{D},\\[0.3em]
\text{Process level:} \quad
& \mc{W}_j(\cdot) \mid \bm{\beta}_{j}, \bm{\Sigma}_{jj}, \phi \;\overset{\mathrm{ind}}{\sim}\; \mc{GP}\big(\mc{X}^\top(\cdot) \bm{\beta}_{j},\; \bm{\Sigma}_{jj}\mc{K}_\phi(\cdot, \cdot) \big),\quad j = 1, \ldots, q,\\[0.3em]
\text{Parameter level:} \quad
& \bm{\beta}_{j} \mid \bm{\Sigma}_{jj} \;\overset{\mathrm{ind}}{\sim}\; \mathcal{N}_{p}\left(\bm{M}_j, \bm{\Sigma}_{jj}\bm{V}\right),\quad j = 1, \ldots, q, \\
& \bm{\Sigma}_{jj} \;\overset{\mathrm{ind}}{\sim}\; \mathcal{IG}(a, b),\quad j = 1, \ldots, q, \\
& \phi \;\sim\; \mathcal{U}(0, b_{\phi}).
\end{aligned}
\end{equation}
The directed acyclic graph in Figure~\ref{fig:dag_separate} reflects the assumption of independence across $\bm{\Sigma}_{jj}$, each assigned an inverse-gamma prior. At the process level, each $\mc{W}_j(\cdot)$ is modeled as a univariate GP with its own mean function $\mc{X}^\top(\cdot)$. At the same time, its variability is controlled by $\bm{\Sigma}_{jj}$, and the shared range parameter $\phi$. This separate modeling framework is used in Section~\ref{sec:simulation_studies} of the main paper as a baseline for comparison with the proposed joint model.
\begin{figure}[]
\centering
\begin{tikzpicture}[node distance=8mm and 8mm, >=stealth]

\node[param] (Wj) {$W_j(\bm{s}_i)$};
\node[obs, right=of Wj] (Yj) {$Y_j(\bm{s}_i)$};

\node[param, above left=5mm and 5mm of Wj] (beta) {$\bm{\beta}_j$};
\node[param, below left=5mm and 5mm of Wj] (sigmajj) {$\Sigma_{jj}$};

\node[param, below=of Wj] (phi) {$\phi$};
\node[hyper, below=of phi] (bphi) {$b_{\phi}$};

\node[hyper, left=of beta] (Mj) {$\bm{M}_j$};
\node[hyper, below=of Mj] (V) {$\bm{V}$};

\node[hyper, left=of sigmajj] (a) {$a$};
\node[hyper, below=of a] (b) {$b$};


\draw[->] (Mj) -- (beta);
\draw[->] (V) -- (beta);
\draw[->] (sigmajj) -- (beta);

\draw[->] (a) -- (sigmajj);
\draw[->] (b) -- (sigmajj);

\draw[->] (bphi) -- (phi);

\draw[->] (beta) -- (Wj);
\draw[->] (sigmajj) -- (Wj);
\draw[->] (phi) -- (Wj);

\draw[->] (Wj) -- (Yj);

\node[
  draw,
  fit=(Wj)(Yj),
  inner sep=6pt,
  label=below right:{\normalsize $i = 1,\ldots,n$}
] (plate_i) {};

\node[
  draw,
  dashed,
  fit=(plate_i)(beta)(sigmajj)(phi),
  inner sep=8pt,
  label=below right:{\normalsize $j = 1,\ldots,q$}
] (plate_j) {};

\end{tikzpicture}
\caption{Directed acyclic graph for the separate model. The outer plate over $j$ explicitly encodes conditional independence across response types, with each response modeled using its univariate latent GP.}
\label{fig:dag_separate}
\end{figure}

\section{Additional simulation results}\label{sec:additional simulation results}

This section presents additional simulation results that complement the main findings in Section~\ref{sec:simulation_studies}. We report detailed posterior summaries of the spatial range parameter $\phi$ across different mixed-type response models, along with additional diagnostics for the Gaussian-Poisson model, including reductions in regression coefficient variability, predictive performance, and cross-response dependence.

\subsection{Choice of $m$ in Vecchia approximation}\label{sec:choice_of_m}
We first examine the choice of the neighbor size $m$ in the Vecchia approximation. Existing implementations in \texttt{R}, such as \texttt{GpGp} and \texttt{GPvecchia}, provide heuristic guidelines for selecting $m$. In the context of nearest neighbor Gaussian process models, \cite{Datta16} recommend choosing $m \in \{10, \ldots, 50\}$, based on minimizing the residual mean squared predictive error. For our proposed mixed-type multivariate spatial model, however, directly evaluating predictive loss across different values of $m$ is computationally infeasible. Instead, we assess the quality of the approximation through the covariance structure. In particular, \cite{schafer2021a} derive theoretical bounds on the Kullback-Leibler divergence between the true GP with full conditioning set and an approximated univariate GP when the covariance kernel is approximated using the Vecchia approximation; the authors consider different choices of $m$. Motivated by these results, we evaluate the approximation error using the relative Frobenius norm of the covariance matrix,
\[
\frac{\norm{\bm{K} - \widetilde{\bm{K}}(m)}_{F}}{\norm{\bm{K}}_{F}},
\]
where $\bm{K}$ denotes the true full covariance matrix and $\widetilde{\bm{K}}(m)$ denotes its Vecchia approximation with neighbor size $m$.
\begin{figure}[!h]
    \centering
    \centering
    \begin{subfigure}[t]{0.49\textwidth}
        \centering
        \includegraphics[width=\linewidth]{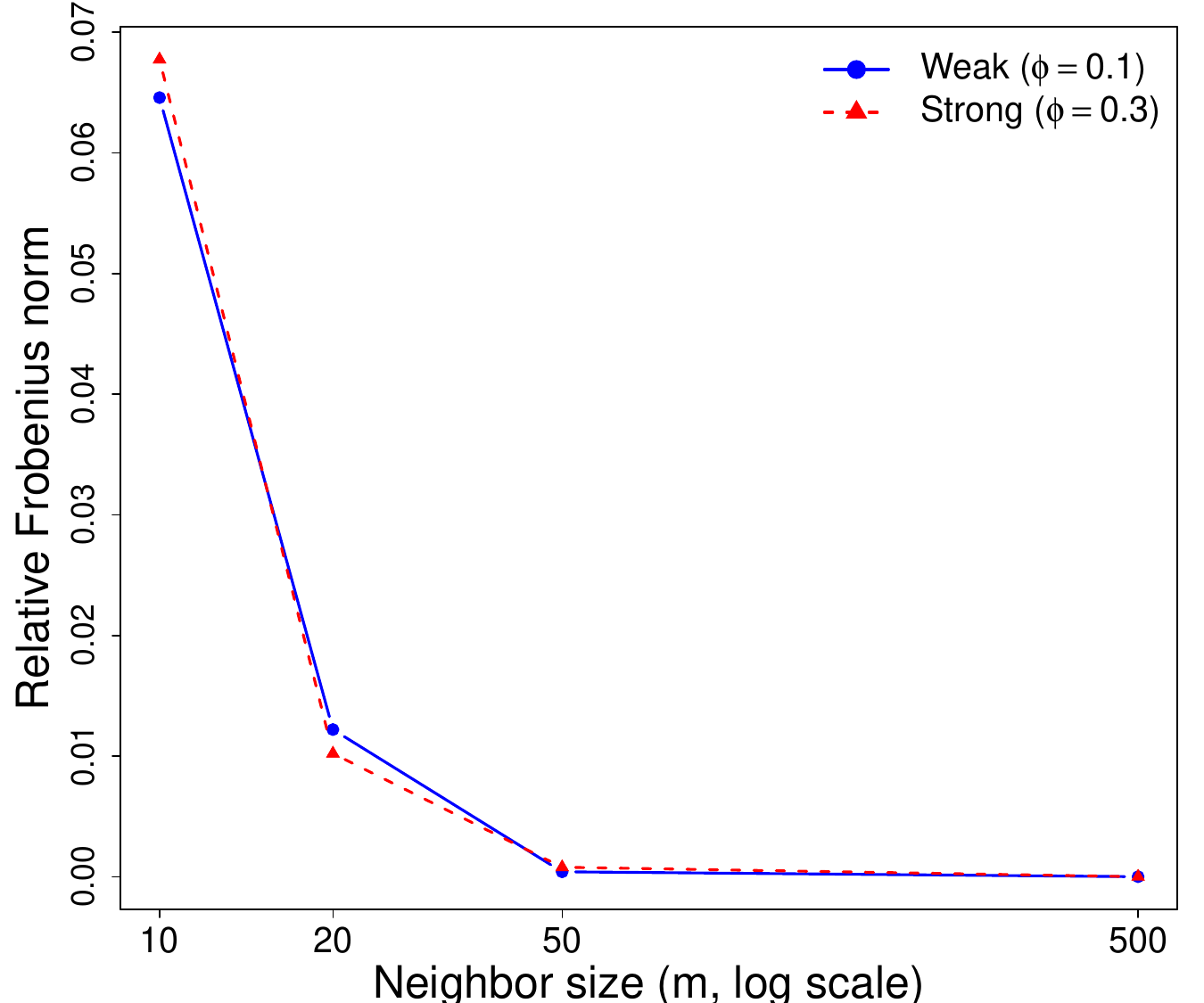}
    \end{subfigure}
    \begin{subfigure}[t]{0.49\textwidth}
        \centering
        \includegraphics[width=\linewidth]{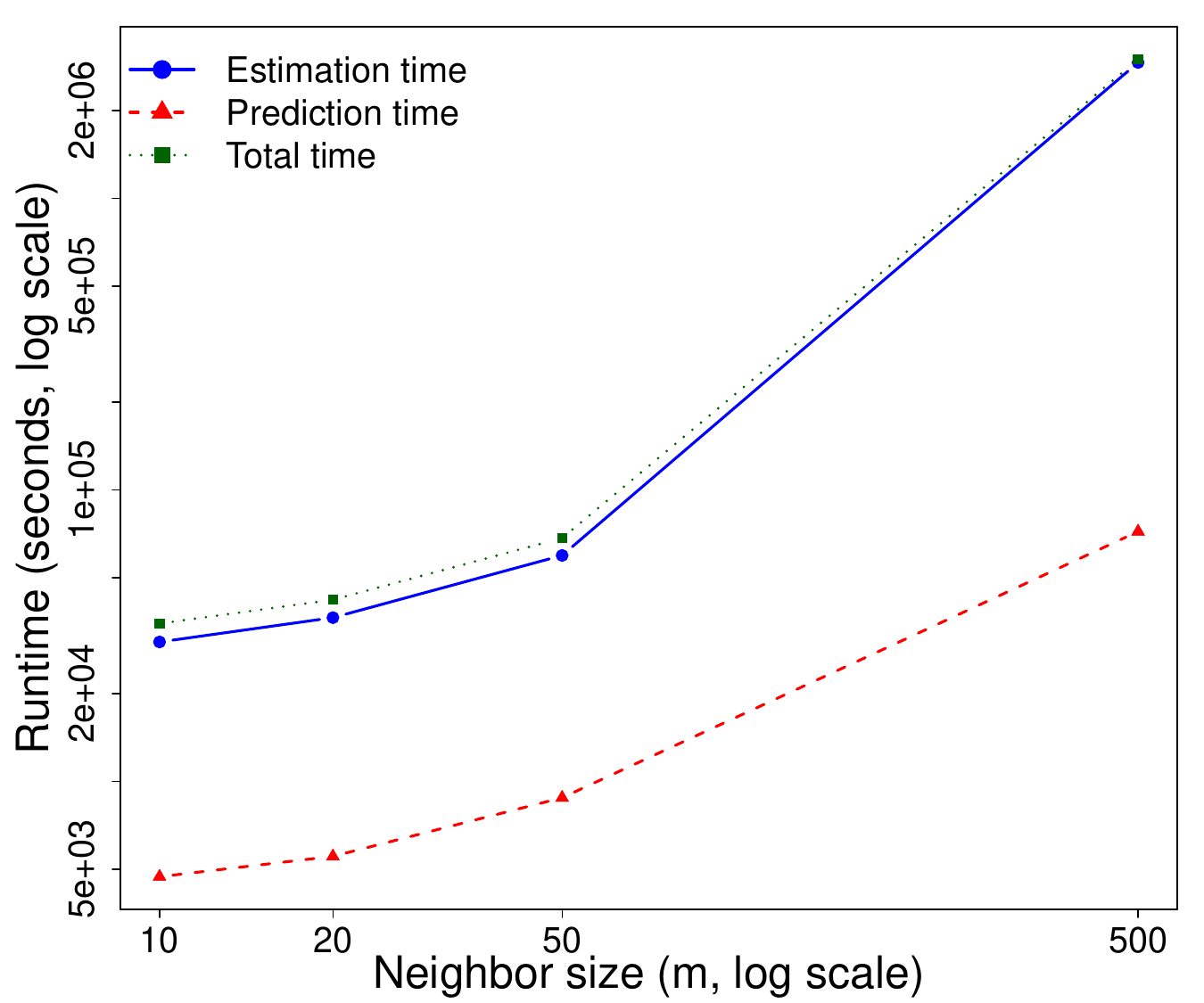}
    \end{subfigure}
    \caption{Elbow plots of relative Frobenius norms under weak ($\phi = 0.1$) and strong ($\phi = 0.3$) spatial dependence (left panel), and wall-clock times (in seconds) for estimation, prediction, and total computation (right panel), across different neighbor sizes $m = 10, 20, 50, 500$ in the Vecchia approximation for the Binomial--Gaussian--Poisson model with $n = 500$ locations for $10^4$ MCMC iterations.}
    \label{fig:vecchia_comparison}
\end{figure}
Since, for each MCMC iteration, $\phi$ and consequently, the large covariance matrix $\bm{K}$ are also updated. Thus, we have to choose the conditioning set size so that it well approximates our model within a feasible time in our complete workflow. From the elbow plot, we see that the relative Frobenius norm decreases rapidly to zero. On the other hand, the clock-wall time increases as $m$ increases. From the elbow plot in Figure~\ref{fig:vecchia_comparison}, we choose $m = 20$ as a reasonable choice to proceed with further analysis.

\begin{table}[!h]
\centering
\caption{Wall-clock times (in seconds) for pre-computation, estimation, and prediction over 100 MCMC iterations of Algorithm~\ref{alg:mixed_model_algorithm} under the Vecchia approximation ($m=20$) and full GP. Results for the full GP at $n = 2500$ are omitted due to prohibitive memory requirements.}
\label{tab:wall_time}
\renewcommand{\arraystretch}{1.1}
\begin{tabular}{ccccccc}
   \hline
 & \multicolumn{2}{c}{\text{Pre-computation}} & \multicolumn{2}{c}{\text{Estimation}} & \multicolumn{2}{c}{\text{Prediction}} \\[2pt]
\cline{2-3} \cline{4-5} \cline{6-7} 
 $n$ & \text{Vecchia} & \text{Full} & \text{Vecchia} & \text{Full} & \text{Vecchia} & \text{Full} \\[2pt]
\hline
100 & 0.04 & 0.08 & 0.75 & 1.15 & 0.04 & 0.05 \\ 
  500 & 0.16 & 1.59 & 3.13 & 128.64 & 0.05 & 0.13 \\ 
  2500 & 3.99 & 13.81 & 15.06 &  & 0.17 &  \\ 
   \hline
\end{tabular}
\end{table}
Table~\ref{tab:wall_time} summarizes the wall-clock time from the Vecchia approximation relative to the full GP. While pre-computation costs are comparable for small $n$, the estimation phase exhibits substantial speedups under the Vecchia approximation with $m = 20$, particularly as $n$ increases. For moderate sample sizes ($n=500$), the reduction in estimation time is over an order of magnitude, and for larger datasets ($n=2500$), the full GP becomes computationally infeasible. These results demonstrate that the Vecchia approximation enables scalable inference with minimal loss of accuracy, supporting our choice of $m=20$ as neighbor size in our data analysis in Section~\ref{sec:data_analysis} in the main article.

\subsection{Quality assessment of samplers}\label{sec:quality_of_samplers}
In the literature on spatial latent Gaussian process models, the Metropolis-Hastings algorithm and its variants are typically used to update $\bm{W}$. While easy to implement, they suffer from slow mixing in high dimensions and thus fail to properly explore the posterior landscape due to the strong correlation in the latent random effect $\bm{W}$. To assess sampler performance in our model, we consider four methods for updating $\bm{W}$, such as elliptical slice sampling with component-wise updates, elliptical slice sampling with joint updates, the preconditioned Crank–Nicolson (pCN) sampler with component-wise updates, and a random-walk Metropolis–Hastings sampler with joint updates. The comparison is carried out on a simulated Binomial–Gaussian–Poisson model with $n = 500$ locations, so that $\bm{W}$ is of dimension $500 \times 3$. Each sampler is run for $10^5$ iterations. We summarize performance using trace plots of the unnormalized posterior density and box plots of log-transformed effective sample sizes computed across $1500$ components. The trace plots in Figure~\ref{fig:comp_samplers} (left panel) indicate that the proposed sampler explores the posterior distribution more effectively than the competing methods. While a well-tuned pCN sampler can achieve comparable behavior with a component-wise elliptical slice sampler, such tuning requires additional effort, often based on information gathered during a warm-up phase. In contrast, our chosen sampler performs reliably without extensive tuning. This is further supported by consistently higher effective sample sizes as shown in Figure~\ref{fig:comp_samplers} (right panel), indicating improved sampling efficiency.
\begin{figure}[]
    \centering
    \centering
    \begin{subfigure}[t]{0.49\textwidth}
        \centering
        \includegraphics[width=\linewidth]{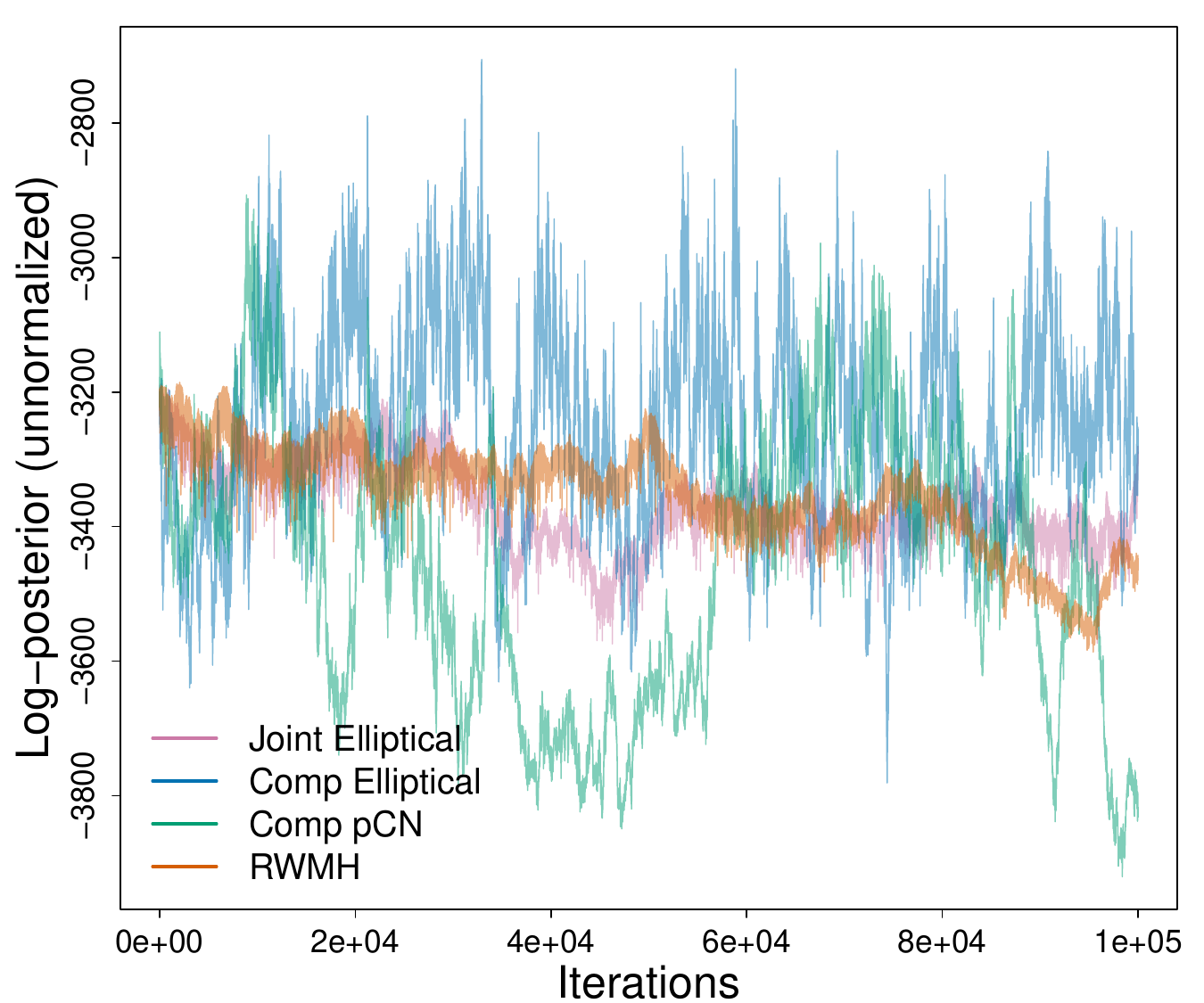}
    \end{subfigure}
    \begin{subfigure}[t]{0.49\textwidth}
        \centering
        \includegraphics[width=\linewidth]{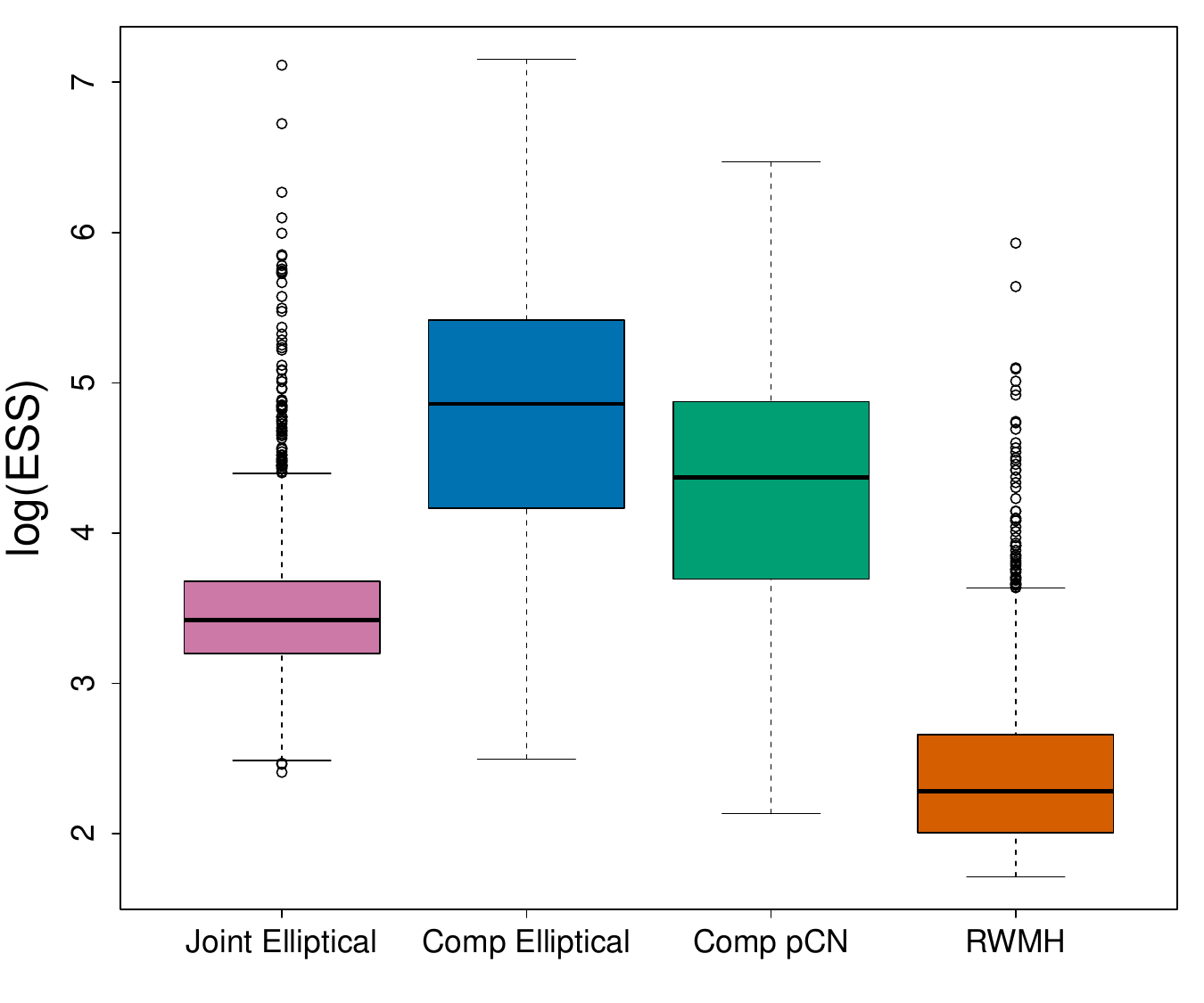}
    \end{subfigure}
    \caption{Traceplots of the unnormalized log-posterior (left panel) and boxplots of $\log(\mathrm{ESS})$ (right panel) of all components of $\bm{W}$ based on $10^5$ iterations, comparing joint elliptical slice sampling, component-wise elliptical slice sampling, component-wise preconditioned Crank--Nicolson (pCN), and joint random-walk Metropolis--Hastings samplers.}
    \label{fig:comp_samplers}
\end{figure}

\subsection{Simulation results for Gaussian-Poisson model}\label{sec:gp_simulations}

We provide a complementary analysis of Section~\ref{sec:simulation_studies} by providing a thorough simulation study of a Gaussian-Poisson model. We summarize cross-covariance estimation in Table~\ref{tab:crosscov_gp}, predictive performance in Table~\ref{tab:elpd_diff_gp}, and the variability of regression coefficients in Figure~\ref{fig:gp_beta_var}.

Table~\ref{tab:crosscov_gp} shows that the cross-covariance parameter $\bm{\Sigma}_{12}$ is well estimated under the joint model across all scenarios. When $\bm{\Sigma}$ is dependent, the estimated posterior mean moves closer to the true value  $1.735$ to $2.070$ for $\phi_0=0.3$ as the sample size increases from $n = 100$ to $n = 2500$, with a noticeable reduction in uncertainty reflected through tighter credible intervals. Coverage also improves from $0.80$ to $0.92$, indicating better estimation in larger samples. Under weak spatial correlation for $\phi_0=0.1$, the estimation remains stable with high coverage ($0.96$ and $0.94$). When $\bm{\Sigma}$ is independent, the estimated posterior mean centers around zero (e.g., $0.013$, $-0.006$), and the credible intervals shrink substantially as $n$ increases, while maintaining high coverage (up to $0.98$).
\begin{table}[!h]
\centering
\caption{Posterior estimation summary of $\bm{\Sigma}_{12}$ (posterior mean, 95\% posterior credible interval, and empirical coverage at the nominal 95\% level) across 50 replicated datasets under varying spatial correlation ($\phi_0$, the true value of $\phi$) and varying cross-covariance matrix ($\bm{\Sigma}^{(0)}$, the true value of $\bm{\Sigma}$) for Gaussian-Poisson model. We choose the diagonal entries of $\bm{\Sigma}^{(0)}$ to be $\bm{\Sigma}^{(0)}_{11} = 3$ and $\bm{\Sigma}^{(0)}_{22} = 2$.}
\label{tab:crosscov_gp}
\renewcommand{\arraystretch}{1}
\begin{tabular}{c c c c c}
\hline
$\bm{\Sigma}_{12}^{(0)}$ & $\phi_0$ & Posterior mean \text{\scriptsize(SE)} & Credible interval \text{\scriptsize(SE)} & Coverage \text{\scriptsize(SE)} \\[2pt]
\rowcolor{lightgray}
\multicolumn{5}{c}{Sample size: $n = 100$} \\ 
\multirow{2}{*}{$2.25$} & $0.3$ & $1.735\!\,\text{\scriptsize(0.081)}$ & $[0.824\!\,\text{\scriptsize(0.048)},\; 3.273\!\,\text{\scriptsize(0.127)}]$ & $0.800\!\,\text{\scriptsize(0.057)}$ \\
& $0.1$ & $2.488\!\,\text{\scriptsize(0.105)}$ & $[1.337\!\,\text{\scriptsize(0.051)},\; 4.799\!\,\text{\scriptsize(0.227)}]$ & $0.960\!\,\text{\scriptsize(0.028)}$ \\
\rowcolor{lightviolet}
\multicolumn{5}{c}{Sample size: $n = 2500$} \\ 
\multirow{2}{*}{$2.25$} & $0.3$ & $2.070\!\,\text{\scriptsize(0.065)}$ & $[1.266\!\,\text{\scriptsize(0.042)},\; 3.093\!\,\text{\scriptsize(0.067)}]$ & $0.920\!\,\text{\scriptsize(0.039)}$ \\
& $0.1$ & $2.467\!\,\text{\scriptsize(0.066)}$ & $[1.810\!\,\text{\scriptsize(0.035)},\; 3.649\!\,\text{\scriptsize(0.139)}]$ & $0.940\!\,\text{\scriptsize(0.034)}$ \\
\rowcolor{lightgray}
\multicolumn{5}{c}{Sample size: $n = 100$} \\ 
\multirow{2}{*}{$0$} & $0.3$ & $0.013\!\,\text{\scriptsize(0.065)}$ & $[-0.754\!\,\text{\scriptsize(0.088)},\; 0.812\!\,\text{\scriptsize(0.083)}]$ & $0.840\!\,\text{\scriptsize(0.052)}$ \\
& $0.1$ & $0.066\!\,\text{\scriptsize(0.056)}$ & $[-0.742\!\,\text{\scriptsize(0.074)},\; 0.892\!\,\text{\scriptsize(0.063)}]$ & $0.940\!\,\text{\scriptsize(0.034)}$ \\
\rowcolor{lightviolet}
\multicolumn{5}{c}{Sample size: $n = 2500$} \\ 
\multirow{2}{*}{$0$} & $0.3$ & $-0.006\!\,\text{\scriptsize(0.019)}$ & $[-0.192\!\,\text{\scriptsize(0.023)},\; 0.186\!\,\text{\scriptsize(0.020)}]$ & $0.900\!\,\text{\scriptsize(0.043)}$ \\
& $0.1$ & $-0.005\!\,\text{\scriptsize(0.014)}$ & $[-0.183\!\,\text{\scriptsize(0.016)},\; 0.175\!\,\text{\scriptsize(0.016)}]$ & $0.980\!\,\text{\scriptsize(0.020)}$ \\
\hline
\end{tabular}
\end{table}

In contrast to the models discussed in the main article, the Gaussian-Poisson model exhibits more nuanced behavior in posterior uncertainty for $\bm{B}$, as illustrated in the left panel of Figure~\ref{fig:gp_beta_var}. When spatial correlation is strong, the joint model achieves a clear reduction in posterior variance compared to the separate model (seen through the yellow boxplots in Figure~\ref{fig:gp_beta_var}), indicating effective borrowing of strength across responses for both $\bm{\Sigma}$ being dependent and independent scenarios. In contrast, under weak spatial correlation with the true value at $\phi = 0.1$, the variability remains comparable between the two modeling approaches (seen through the blue boxplots in Figure~\ref{fig:gp_beta_var}), suggesting that the gains from joint modeling are limited when dependence is weak.
\begin{figure}[ht]
    \centering
    \includegraphics[width=\linewidth]{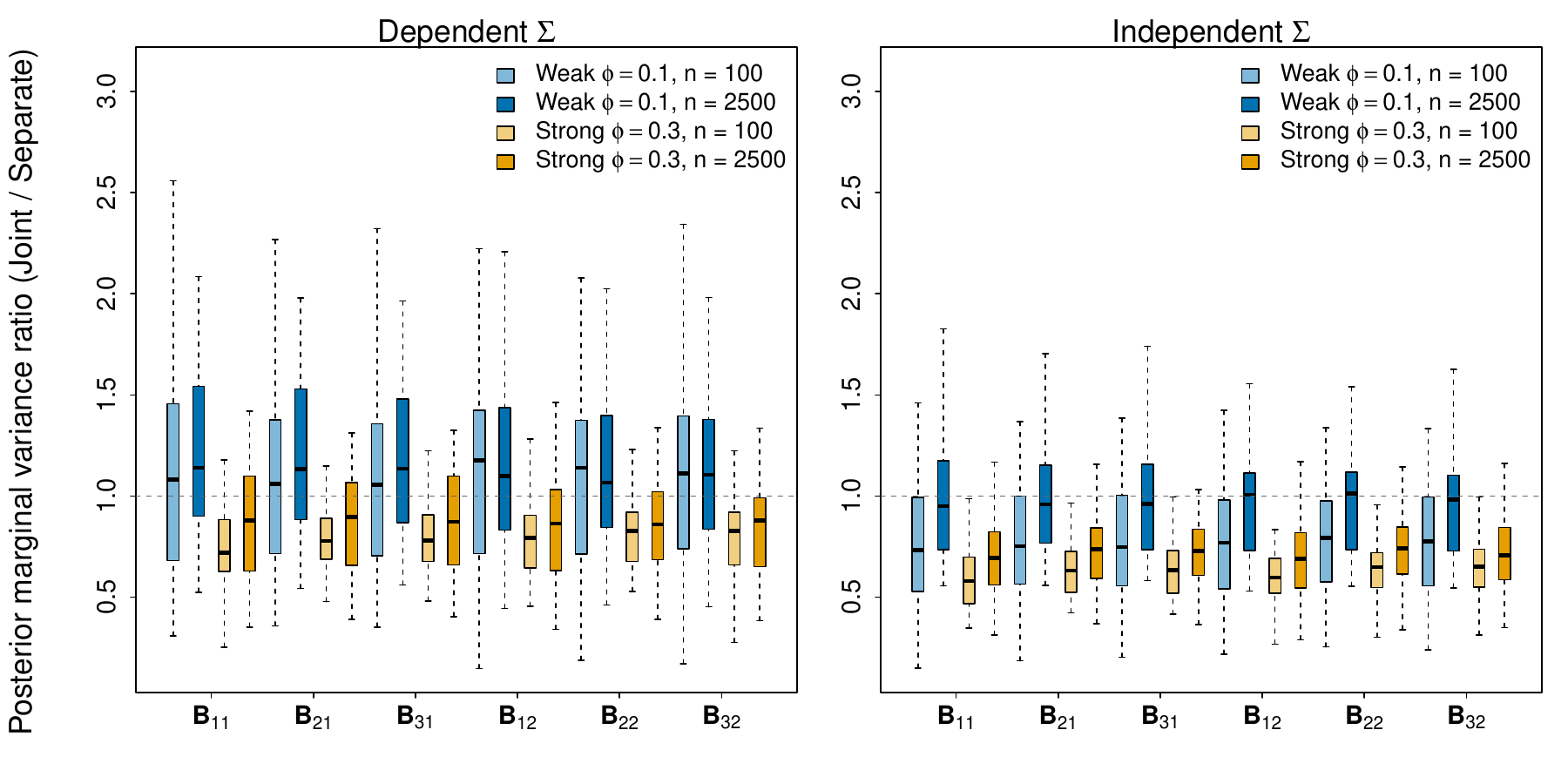}
    \caption{Efficacy in variance reduction in regression coefficient matrix $\bm{B}$ components in the joint model in comparison to the separate model for the Gaussian-Poisson case study on 50 replicated datasets.}
    \label{fig:gp_beta_var}
\end{figure}

These findings are further supported by predictive performance, as demonstrated in Table~\ref{tab:elpd_diff_gp}. When $\bm{\Sigma}$ is dependent, the joint model consistently outperforms the separate model, with positive ELJPD differences (e.g., $4.601$ and $10.913$ for $n=100$, increasing further for $n=2500$), and the gains are particularly pronounced under weak spatial correlation. When $\bm{\Sigma}$ is independent, however, the ELJPD differences become negative, indicating that the joint model offers no clear advantage in the absence of cross-dependence and may introduce slight inefficiency.
\begin{table}[!h]
\centering
\caption{ELJPD differences of joint and separate models across 50 replications for Gaussian-Poisson response types.}
\label{tab:elpd_diff_gp}
\renewcommand{\arraystretch}{1}
\small
\begin{tabular}{c c c c}
\hline
$\textrm{vec}(\bm{\Sigma}^{(0)})$ & $\phi_0$ & \cellcolor{lightgray}{$n = 100$} & \cellcolor{lightviolet}{$n = 2500$} \\ [2pt]
\multirow{2}{*}{$[3,\tfrac{9}{4}, \tfrac{9}{4}, 2]^{\top}$} & $0.3$ & $\hphantom{-}4.601\!\,\text{\scriptsize(0.697)}$ & $\hphantom{-}4.269\!\,\text{\scriptsize(2.277)}$ \\
& $0.1$ & $\hphantom{-}10.913\!\,\text{\scriptsize(1.240)}$ & $\hphantom{-}27.863\!\,\text{\scriptsize(4.152)}$ \\
\hline
\multirow{2}{*}{$[3,0,0,2]^{\top}$} & $0.3$ & $-0.335\!\,\text{\scriptsize(0.465)}$ & $-3.783\!\,\text{\scriptsize(1.950)}$ \\
& $0.1$ & $-0.262\!\,\text{\scriptsize(0.875)}$ & $-9.444\!\,\text{\scriptsize(5.179)}$ \\
\hline
\end{tabular}
\end{table}

Despite accurate recovery of the cross-covariance parameter $\bm{\Sigma}_{12}$ (Table~\ref{tab:crosscov_gp}), the inherent variability in count responses leads to comparable uncertainty levels between joint and separate models. Nevertheless, the joint model continues to provide clear advantages in overall inference and predictive performance, supporting the use of our methodology in these settings.

\subsection{Estimation of Mat\'{e}rn range parameter}\label{sec:phi_summaries}
We report posterior inference for the spatial range parameter $\phi$ across mixed-type models introduced in Section~\ref{sec:simulation_studies}. As discussed in the main article, estimating $\phi$ is inherently challenging under fixed-domain asymptotics \citep{Zhang2004inconsistent}. Nevertheless, the joint model consistently improves uncertainty quantification by borrowing strength across response types, leading to more reliable coverage compared to separate modeling.

For the Binomial-Gaussian settings on $50$ replicated datasets, Table~\ref{tab:phi_bg} shows that when $\bm{\Sigma}$ is dependent with strong spatial correlation ($\phi_0=0.3$), the joint model yields posterior mean $0.198$ for $n=100$ and $0.254$ for $n=2500$, compared to $0.375$ and $0.347$ under the separate model, indicating reduced upward bias. The bias is significantly reduced when the dimension is high. Under weak correlation ($\phi_0=0.1$), the average posterior means $0.098$ and $0.101$ remain close to the truth. In contrast, the separate model overestimates $0.194$ in a low-dimensional setup while showing an improved estimate of $0.117$ at high dimension. When the true model is dependent with strong correlation being present, on average, the joint model produces wider credible intervals (e.g., $[0.078,0.390]$ in comparison to the separate model $[0.242,0.458]$) for $n = 100$, yielding substantially improved coverage (up to $0.98$ vs $0.68$). At the same time, it is narrower for the $n = 2500$ scenario, but the coverage reduces for both the joint and the separate models. When the off-diagonal entries of $\bm{\Sigma}$ are zeros, the same pattern persists. However, differences in posterior means narrow slightly for large $n$ (e.g., $0.338$ vs $0.343$), with the joint model still achieving better coverage (e.g., $0.78$ vs $0.66$).
\begin{table}[t]
\centering
\caption{Posterior summaries of $\phi$ (posterior mean, 95\% credible interval, and empirical coverage at the nominal 95\% level) across 50 replications for joint and separate models for Binomial-Gaussian settings.}
\label{tab:phi_bg}
\renewcommand{\arraystretch}{1}
\setlength{\tabcolsep}{5pt}   
\small
\begin{tabular}{c c c c c c}
\hline
$\bm{\Sigma}^{(0)}$ & $\phi_0$ & Model & Posterior mean \text{\scriptsize(SE)} & Credible interval \text{\scriptsize(SE)} & Coverage \text{\scriptsize(SE)} \\[2pt]
\rowcolor{lightgray}
\multicolumn{6}{c}{{Sample size: $n = 100$}} \\
\multirow{2}{*}{$\begin{bmatrix}
9 & 5 \\
5 & 3
\end{bmatrix}$} & $0.3$ & Joint & $0.198\!\,\text{\scriptsize(0.007)}$ & $[0.078\!\,\text{\scriptsize(0.005)},\; 0.390\!\,\text{\scriptsize(0.008)}]$ & $\mathbf{0.96}\!\,\text{\scriptsize(0.028)}$ \\
& & Separate & $0.375\!\,\text{\scriptsize(0.010)}$ & $[0.242\!\,\text{\scriptsize(0.012)},\; 0.458\!\,\text{\scriptsize(0.004)}]$ & $0.68\!\,\text{\scriptsize(0.067)}$ \\
& $0.1$ & Joint & $0.098\!\,\text{\scriptsize(0.006)}$ & $[0.029\!\,\text{\scriptsize(0.004)},\; 0.204\!\,\text{\scriptsize(0.010)}]$ & $\mathbf{0.98}\!\,\text{\scriptsize(0.020)}$ \\
& & Separate & $0.194\!\,\text{\scriptsize(0.012)}$ & $[0.075\!\,\text{\scriptsize(0.007)},\; 0.348\!\,\text{\scriptsize(0.016)}]$ & $0.74\!\,\text{\scriptsize(0.063)}$ \\
\hline
\multirow{2}{*}{$\begin{bmatrix}
9 & 0 \\
0 & 3
\end{bmatrix}$} & $0.3$ & Joint & $0.263\!\,\text{\scriptsize(0.008)}$ & $[0.114\!\,\text{\scriptsize(0.007)},\; 0.440\!\,\text{\scriptsize(0.005)}]$ & $\mathbf{0.98}\!\,\text{\scriptsize(0.020)}$ \\
& & Separate & $0.387\!\,\text{\scriptsize(0.006)}$ & $[0.238\!\,\text{\scriptsize(0.009)},\; 0.464\!\,\text{\scriptsize(0.002)}]$ & $0.86\!\,\text{\scriptsize(0.050)}$ \\
& $0.1$ & Joint & $0.142\!\,\text{\scriptsize(0.008)}$ & $[0.045\!\,\text{\scriptsize(0.005)},\; 0.313\!\,\text{\scriptsize(0.015)}]$ & $\mathbf{0.88}\!\,\text{\scriptsize(0.046)}$ \\
& & Separate & $0.188\!\,\text{\scriptsize(0.013)}$ & $[0.073\!\,\text{\scriptsize(0.008)},\; 0.337\!\,\text{\scriptsize(0.015)}]$ & $0.72\!\,\text{\scriptsize(0.064)}$ \\
\rowcolor{lightviolet}
\multicolumn{6}{c}{{Sample size: $n = 2500$}} \\
\multirow{2}{*}{$\begin{bmatrix}
9 & 5 \\
5 & 3
\end{bmatrix}$} & $0.3$ & Joint & $0.254\!\,\text{\scriptsize(0.006)}$ & $[0.198\!\,\text{\scriptsize(0.005)},\; 0.327\!\,\text{\scriptsize(0.007)}]$ & $\mathbf{0.76}\!\,\text{\scriptsize(0.061)}$ \\
& & Separate & $0.347\!\,\text{\scriptsize(0.004)}$ & $[0.292\!\,\text{\scriptsize(0.002)},\; 0.402\!\,\text{\scriptsize(0.006)}]$ & $0.70\!\,\text{\scriptsize(0.065)}$ \\
& $0.1$ & Joint & $0.101\!\,\text{\scriptsize(0.002)}$ & $[0.083\!\,\text{\scriptsize(0.001)},\; 0.123\!\,\text{\scriptsize(0.003)}]$ & $\mathbf{0.88}\!\,\text{\scriptsize(0.046)}$ \\
& & Separate & $0.117\!\,\text{\scriptsize(0.001)}$ & $[0.098\!\,\text{\scriptsize(0.001)},\; 0.136\!\,\text{\scriptsize(0.002)}]$ & $0.62\!\,\text{\scriptsize(0.069)}$ \\
\hline
\multirow{2}{*}{$\begin{bmatrix}
9 & 0 \\
0 & 3
\end{bmatrix}$} & $0.3$ & Joint & $0.338\!\,\text{\scriptsize(0.006)}$ & $[0.264\!\,\text{\scriptsize(0.005)},\; 0.416\!\,\text{\scriptsize(0.007)}]$ & $\mathbf{0.88}\!\,\text{\scriptsize(0.046)}$ \\
& & Separate & $0.343\!\,\text{\scriptsize(0.005)}$ & $[0.289\!\,\text{\scriptsize(0.003)},\; 0.398\!\,\text{\scriptsize(0.007)}]$ & $0.70\!\,\text{\scriptsize(0.065)}$ \\
& $0.1$ & Joint & $0.133\!\,\text{\scriptsize(0.003)}$ & $[0.095\!\,\text{\scriptsize(0.001)},\; 0.179\!\,\text{\scriptsize(0.006)}]$ & $\mathbf{0.78}\!\,\text{\scriptsize(0.068)}$ \\
& & Separate & $0.116\!\,\text{\scriptsize(0.001)}$ & $[0.098\!\,\text{\scriptsize(0.001)},\; 0.133\!\,\text{\scriptsize(0.002)}]$ & $0.66\!\,\text{\scriptsize(0.068)}$ \\
\hline
\end{tabular}
\end{table}

For the Binomial-Poisson settings (see Table~\ref{tab:phi_bp}), under nonzero off-diagonal entries of $\bm{\Sigma}$ and strong spatial correlation $\phi = 0.3$, both the joint model and the separate model produce biased posterior means for $n=100$. At the same time, the biases reduce when $n=2500$, for both the joint and separate models, with average posterior means of $0.261$ and $0.327$, respectively. For weak correlation, the joint model produces nearly unbiased estimates of $\phi$ ($0.097$ and $0.098$) in both low and high-dimensional regimes. In comparison, the separate model remains upwardly biased in the low-dimensional scenario ($0.167$), yielding an improved estimate of $0.108$ at high-dimension with $n = 2500$. The credible intervals for the joint model are moderately wider, leading to near-nominal or perfect coverage (e.g., $1.00$ vs $0.82$) in the presence of weak spatial correlation with $\phi = 0.1$ for $n=100$ and for $n = 2500$, respectively. However, for long-range spatial correlation ($\phi = 0.3$), the truth is not captured with high coverage; the opposite is true when the off-diagonal entries of $\bm{\Sigma}$ are zeros. Posterior means for the two models are $0.350$ and $0.325$ for large $n = 2500$. Still, the joint model provides slightly higher coverage ($0.82$ vs $0.76$) with only a modest increase in the credible interval width.
\begin{table}[t]
\centering
\caption{Posterior summaries of $\phi$ (posterior mean, 95\% credible interval, and empirical coverage at the nominal 95\% level) across 50 replications for joint and separate models for Binomial-Poisson settings.}
\label{tab:phi_bp}
\renewcommand{\arraystretch}{1}
\setlength{\tabcolsep}{5pt}   
\small
\begin{tabular}{c c c c c c}
\hline
$\bm{\Sigma}^{(0)}$ & $\phi_0$ & Model & Posterior mean \text{\scriptsize(SE)} & Credible interval \text{\scriptsize(SE)} & Coverage \text{\scriptsize(SE)} \\[2pt]
\rowcolor{lightgray}
\multicolumn{6}{c}{{Sample size: $n = 100$}} \\
\multirow{2}{*}{$\begin{bmatrix}
9 & 4 \\
4 & 2
\end{bmatrix}$} & $0.3$ & Joint & $0.212\!\,\text{\scriptsize(0.008)}$ & $[0.090\!\,\text{\scriptsize(0.005)},\; 0.393\!\,\text{\scriptsize(0.009)}]$ & $\mathbf{0.88}\!\,\text{\scriptsize(0.046)}$ \\
& & Separate & $0.357\!\,\text{\scriptsize(0.010)}$ & $[0.224\!\,\text{\scriptsize(0.013)},\; 0.454\!\,\text{\scriptsize(0.004)}]$ & $0.80\!\,\text{\scriptsize(0.057)}$ \\
& $0.1$ & Joint & $0.097\!\,\text{\scriptsize(0.005)}$ & $[0.028\!\,\text{\scriptsize(0.004)},\; 0.205\!\,\text{\scriptsize(0.010)}]$ & $\mathbf{1.00}\!\,\text{\scriptsize(0.000)}$ \\
& & Separate & $0.167\!\,\text{\scriptsize(0.012)}$ & $[0.063\!\,\text{\scriptsize(0.008)},\; 0.299\!\,\text{\scriptsize(0.014)}]$ & $0.82\!\,\text{\scriptsize(0.055)}$ \\
\hline
\multirow{2}{*}{$\begin{bmatrix}
9 & 0 \\
0 & 2
\end{bmatrix}$} & $0.3$ & Joint & $0.256\!\,\text{\scriptsize(0.008)}$ & $[0.110\!\,\text{\scriptsize(0.006)},\; 0.434\!\,\text{\scriptsize(0.006)}]$ & $\mathbf{0.96}\!\,\text{\scriptsize(0.028)}$ \\
& & Separate & $0.369\!\,\text{\scriptsize(0.009)}$ & $[0.230\!\,\text{\scriptsize(0.012)},\; 0.459\!\,\text{\scriptsize(0.002)}]$ & $0.76\!\,\text{\scriptsize(0.061)}$ \\
& $0.1$ & Joint & $0.132\!\,\text{\scriptsize(0.008)}$ & $[0.041\!\,\text{\scriptsize(0.005)},\; 0.299\!\,\text{\scriptsize(0.015)}]$ & $\mathbf{0.92}\!\,\text{\scriptsize(0.039)}$ \\
& & Separate & $0.183\!\,\text{\scriptsize(0.013)}$ & $[0.075\!\,\text{\scriptsize(0.010)},\; 0.318\!\,\text{\scriptsize(0.015)}]$ & $0.78\!\,\text{\scriptsize(0.059)}$ \\
\rowcolor{lightviolet}
\multicolumn{6}{c}{{Sample size: $n = 2500$}} \\
\multirow{2}{*}{$\begin{bmatrix}
9 & 4 \\
4 & 2
\end{bmatrix}$} & $0.3$ & Joint & $0.261\!\,\text{\scriptsize(0.006)}$ & $[0.214\!\,\text{\scriptsize(0.006)},\; 0.325\!\,\text{\scriptsize(0.006)}]$ & $\mathbf{0.88}\!\,\text{\scriptsize(0.046)}$\\
& & Separate & $0.327\!\,\text{\scriptsize(0.003)}$ & $[0.289\!\,\text{\scriptsize(0.002)},\; 0.368\!\,\text{\scriptsize(0.005)}]$ & $0.78\!\,\text{\scriptsize(0.059)}$  \\
& $0.1$ & Joint & $0.098\!\,\text{\scriptsize(0.002)}$ & $[0.083\!\,\text{\scriptsize(0.002)},\; 0.118\!\,\text{\scriptsize(0.002)}]$ & $\mathbf{0.82}\!\,\text{\scriptsize(0.055)}$ \\
& & Separate & $0.108\!\,\text{\scriptsize(0.001)}$ & $[0.097\!\,\text{\scriptsize(0.001)},\; 0.121\!\,\text{\scriptsize(0.001)}]$ & $0.76\!\,\text{\scriptsize(0.061)}$ \\
\hline
\multirow{2}{*}{$\begin{bmatrix}
9 & 0 \\
0 & 2
\end{bmatrix}$} & $0.3$ & Joint & $0.350\!\,\text{\scriptsize(0.006)}$ & $[0.271\!\,\text{\scriptsize(0.004)},\; 0.430\!\,\text{\scriptsize(0.006)}]$ & $\mathbf{0.88}\!\,\text{\scriptsize(0.046)}$ \\
& & Separate & $0.325\!\,\text{\scriptsize(0.004)}$ & $[0.287\!\,\text{\scriptsize(0.002)},\; 0.364\!\,\text{\scriptsize(0.005)}]$ & $0.84\!\,\text{\scriptsize(0.052)}$ \\
& $0.1$ & Joint & $0.129\!\,\text{\scriptsize(0.002)}$ & $[0.096\!\,\text{\scriptsize(0.001)},\; 0.167\!\,\text{\scriptsize(0.005)}]$ & $\mathbf{0.82}\!\,\text{\scriptsize(0.055)}$ \\
& & Separate & $0.109\!\,\text{\scriptsize(0.001)}$ & $[0.097\!\,\text{\scriptsize(0.001)},\; 0.121\!\,\text{\scriptsize(0.002)}]$ & $0.76\!\,\text{\scriptsize(0.061)}$ \\
\hline
\end{tabular}
\end{table}

For the Gaussian-Poisson setting, a similar pattern is observed, with the joint model providing more stable inference across all scenarios. When off-diagonal entries of $\bm{\Sigma}$ are nonzero and spatial correlation is strong ($\phi_0 = 0.3$), the separate model shows clear upward bias in the posterior mean ($0.376$ for $n=100$ and $0.334$ for $n=2500$). In contrast, the joint model reduces this bias ($0.262$ and $0.292$), though it tends to underestimate in smaller samples. Under weak correlation ($\phi_0 = 0.1$), the joint model remains reasonably close to the truth ($0.136$, $0.116$), while the separate model is less consistent, slightly overestimates for $n=100$ ($0.148$), and becomes nearly unbiased for $n=2500$ ($0.110$). 
This improvement in point estimation is accompanied by wider credible intervals under the joint model (for instance, $[0.116,0.441]$ compared to $[0.254,0.455]$ when $n=100$), which leads to a clear gain in coverage (up to $1.00$ vs $0.68$). When off-diagonal entries of $\bm{\Sigma}$ are zeros, the same trend continues: for strong correlation, the joint model remains closer to the true value ($0.244$ and $0.259$) than the separate model ($0.379$ and $0.330$). In contrast, for weak correlation, both models become comparable as $n$ increases (e.g., $0.108$ vs $0.110$ for $n=2500$). Even in this setting, the joint model consistently achieves higher coverage (e.g., $0.98$ vs $0.72$), reflecting better-calibrated uncertainty despite only moderately wider intervals.
\begin{table}[t]
\centering
\caption{Posterior summaries of $\phi$ (posterior mean, 95\% credible interval, and empirical coverage at the nominal 95\% level) across 50 replications for joint and separate models for Gaussian-Poisson settings.}
\label{tab:phi_gp}
\renewcommand{\arraystretch}{1}
\setlength{\tabcolsep}{5pt}   
\small
\begin{tabular}{c c c c c c}
\hline
$\bm{\Sigma}^{(0)}$ & $\phi_0$ & Model & Posterior mean \text{\scriptsize(SE)} & Credible interval \text{\scriptsize(SE)} & Coverage \text{\scriptsize(SE)} \\[2pt]
\rowcolor{lightgray}
\multicolumn{6}{c}{{Sample size: $n = 100$}} \\[2pt]
\multirow{2}{*}{$\begin{bmatrix}
3 & \tfrac{9}{4} \\
\tfrac{9}{4} & 2
\end{bmatrix}$} & $0.3$ & Joint & $0.262\!\,\text{\scriptsize(0.008)}$ & $[0.116\!\,\text{\scriptsize(0.006)},\; 0.441\!\,\text{\scriptsize(0.005)}]$ & $\mathbf{1.00}\!\,\text{\scriptsize(0.000)}$ \\
& & Separate & $0.376\!\,\text{\scriptsize(0.009)}$ & $[0.254\!\,\text{\scriptsize(0.010)},\; 0.455\!\,\text{\scriptsize(0.005)}]$ & $0.68\!\,\text{\scriptsize(0.067)}$ \\
& $0.1$ & Joint & $0.136\!\,\text{\scriptsize(0.008)}$ & $[0.048\!\,\text{\scriptsize(0.005)},\; 0.309\!\,\text{\scriptsize(0.015)}]$ & $\mathbf{0.92}\!\,\text{\scriptsize(0.039)}$ \\
& & Separate & $0.148\!\,\text{\scriptsize(0.010)}$ & $[0.069\!\,\text{\scriptsize(0.006)},\; 0.248\!\,\text{\scriptsize(0.013)}]$ & $0.78\!\,\text{\scriptsize(0.059)}$ \\
\hline
\multirow{2}{*}{$\begin{bmatrix}
3 & 0 \\
0 & 2
\end{bmatrix}$} & $0.3$ & Joint & $0.244\!\,\text{\scriptsize(0.007)}$ & $[0.119\!\,\text{\scriptsize(0.005)},\; 0.421\!\,\text{\scriptsize(0.006)}]$ & $\mathbf{0.98}\!\,\text{\scriptsize(0.020)}$ \\
& & Separate & $0.379\!\,\text{\scriptsize(0.009)}$ & $[0.258\!\,\text{\scriptsize(0.011)},\; 0.457\!\,\text{\scriptsize(0.003)}]$ & $0.72\!\,\text{\scriptsize(0.064)}$ \\
& $0.1$ & Joint & $0.111\!\,\text{\scriptsize(0.006)}$ & $[0.047\!\,\text{\scriptsize(0.004)},\; 0.216\!\,\text{\scriptsize(0.011)}]$ & $\mathbf{0.98}\!\,\text{\scriptsize(0.020)}$ \\
& & Separate & $0.158\!\,\text{\scriptsize(0.010)}$ & $[0.077\!\,\text{\scriptsize(0.006)},\; 0.259\!\,\text{\scriptsize(0.013)}]$ & $0.82\!\,\text{\scriptsize(0.055)}$ \\[2pt]
\rowcolor{lightviolet}
\multicolumn{6}{c}{{Sample size: $n = 2500$}} \\[2pt]
\multirow{2}{*}{$\begin{bmatrix}
3 & \tfrac{9}{4} \\
\tfrac{9}{4} & 2
\end{bmatrix}$} & $0.3$ & Joint & $0.292\!\,\text{\scriptsize(0.008)}$ & $[0.178\!\,\text{\scriptsize(0.005)},\; 0.431\!\,\text{\scriptsize(0.007)}]$ & $\mathbf{0.94}\!\,\text{\scriptsize(0.034)}$ \\
& & Separate & $0.334\!\,\text{\scriptsize(0.004)}$ & $[0.293\!\,\text{\scriptsize(0.002)},\; 0.374\!\,\text{\scriptsize(0.006)}]$ & $0.70\!\,\text{\scriptsize(0.065)}$ \\
& $0.1$ & Joint & $0.116\!\,\text{\scriptsize(0.003)}$ & $[0.083\!\,\text{\scriptsize(0.002)},\; 0.175\!\,\text{\scriptsize(0.007)}]$ & $\mathbf{0.96}\!\,\text{\scriptsize(0.028)}$ \\
& & Separate & $0.110\!\,\text{\scriptsize(0.001)}$ & $[0.098\!\,\text{\scriptsize(0.001)},\; 0.122\!\,\text{\scriptsize(0.001)}]$ & $0.66\!\,\text{\scriptsize(0.068)}$ \\
\hline
\multirow{2}{*}{$\begin{bmatrix}
3 & 0 \\
0 & 2
\end{bmatrix}$} & $0.3$ & Joint & $0.259\!\,\text{\scriptsize(0.006)}$ & $[0.165\!\,\text{\scriptsize(0.004)},\; 0.402\!\,\text{\scriptsize(0.008)}]$ & $\mathbf{0.92}\!\,\text{\scriptsize(0.039)}$ \\
& & Separate & $0.330\!\,\text{\scriptsize(0.003)}$ & $[0.292\!\,\text{\scriptsize(0.002)},\; 0.367\!\,\text{\scriptsize(0.005)}]$ & $0.70\!\,\text{\scriptsize(0.065)}$ \\
& $0.1$ & Joint & $0.108\!\,\text{\scriptsize(0.002)}$ & $[0.082\!\,\text{\scriptsize(0.001)},\; 0.150\!\,\text{\scriptsize(0.004)}]$ & $\mathbf{0.98}\!\,\text{\scriptsize(0.020)}$ \\
& & Separate & $0.110\!\,\text{\scriptsize(0.001)}$ & $[0.098\!\,\text{\scriptsize(0.000)},\; 0.122\!\,\text{\scriptsize(0.001)}]$ & $0.72\!\,\text{\scriptsize(0.064)}$ \\
\hline
\end{tabular}
\end{table}

Lastly, we discuss the key insights of Table~\ref{tab:phi_bgp} for the trivariate response model. When off-diagonal entries of $\bm{\Sigma}$ are nonzero, and spatial correlation is also strong, the joint model substantially reduces bias (e.g., $0.183$ vs $0.363$ for $n=100$, $0.228$ vs $0.336$ for $n=2500$). Under weak correlation, the posterior mean estimates of the joint model ($0.086$, $0.096$) remain close to $\phi_0$, while the separate model overestimates ($0.140$, $0.112$) for $n = 100$ and $n = 2500$, respectively. The wider credible intervals under joint modeling (e.g., $[0.086,0.344]$ vs $[0.276,0.470]$) result in significantly improved coverage (e.g., $0.96$ vs $0.64$) for the joint model compared to the separate model. When off-diagonal entries of $\bm{\Sigma}$ are zeros, posterior means across models become comparable for large $n$ (e.g., $0.323$ vs $0.333$). Still, the joint model maintains better-calibrated credible intervals with coverage of $0.84$ for $n = 100$ and $0.90$ for $n = 2500$, vs $0.66$ for $n = 100$ and $0.72$ for $n = 2500$.
\begin{table}[t]
\centering
\caption{Posterior summaries of $\phi$ (posterior mean, 95\% credible interval, and empirical coverage at the nominal 95\% level) across 50 replications for joint and separate models for Binomial-Gaussian-Poisson settings.}
\label{tab:phi_bgp}
\renewcommand{\arraystretch}{1}
\setlength{\tabcolsep}{5pt}   
\small
\begin{tabular}{@{}c c c c c c@{}}
\hline
$\bm{\Sigma}^{(0)}$ & $\phi_0$ & Model & Posterior mean \text{\scriptsize(SE)} & Credible interval \text{\scriptsize(SE)} & Coverage \text{\scriptsize(SE)} \\[2pt]
\rowcolor{lightgray}
\multicolumn{6}{c}{Sample size: $n = 100$} \\

\multirow{2}{*}{$\!\begin{bmatrix}
9 & 5 & 4 \\
5 & 3 & \tfrac{9}{4} \\
4 & \tfrac{9}{4} & 2
\end{bmatrix}\!$}
& $0.3$ & Joint & $0.183\!\,\text{\scriptsize(0.007)}$ & $[0.086\!\,\text{\scriptsize(0.004)},\; 0.344\!\,\text{\scriptsize(0.010)}]$ & $\mathbf{0.78}\!\,\text{\scriptsize(0.067)}$ \\
& & Separate & $0.363\!\,\text{\scriptsize(0.004)}$ & $[0.276\!\,\text{\scriptsize(0.001)},\; 0.470\!\,\text{\scriptsize(0.001)}]$ & $0.64\!\,\text{\scriptsize(0.044)}$ \\

& $0.1$ & Joint & $0.086\!\,\text{\scriptsize(0.005)}$ & $[0.030\!\,\text{\scriptsize(0.004)},\; 0.165\!\,\text{\scriptsize(0.008)}]$ & $\mathbf{0.96}\!\,\text{\scriptsize(0.028)}$ \\
& & Separate & $0.140\!\,\text{\scriptsize(0.003)}$ & $[0.092\!\,\text{\scriptsize(0.001)},\; 0.267\!\,\text{\scriptsize(0.001)}]$ & $0.72\!\,\text{\scriptsize(0.005)}$ \\

\hline

\multirow{2}{*}{$\!\begin{bmatrix}
9 & 0 & 0 \\
0 & 3 & 0 \\
0 & 0 & 2
\end{bmatrix}\!$}
& $0.3$ & Joint & $0.227\!\,\text{\scriptsize(0.008)}$ & $[0.114\!\,\text{\scriptsize(0.006)},\; 0.391\!\,\text{\scriptsize(0.009)}]$ & $\mathbf{0.92}\!\,\text{\scriptsize(0.039)}$ \\
& & Separate & $0.351\!\,\text{\scriptsize(0.004)}$ & $[0.274\!\,\text{\scriptsize(0.001)},\; 0.468\!\,\text{\scriptsize(0.001)}]$ & $0.74\!\,\text{\scriptsize(0.022)}$ \\

& $0.1$ & Joint & $0.106\!\,\text{\scriptsize(0.005)}$ & $[0.047\!\,\text{\scriptsize(0.004)},\; 0.197\!\,\text{\scriptsize(0.009)}]$ & $\mathbf{0.94}\!\,\text{\scriptsize(0.034)}$ \\
& & Separate & $0.162\!\,\text{\scriptsize(0.004)}$ & $[0.082\!\,\text{\scriptsize(0.001)},\; 0.263\!\,\text{\scriptsize(0.002)}]$ & $0.76\!\,\text{\scriptsize(0.024)}$ \\

\rowcolor{lightviolet}
\multicolumn{6}{c}{Sample size: $n = 2500$} \\

\multirow{2}{*}{$\!\begin{bmatrix}
9 & 5 & 4 \\
5 & 3 & \tfrac{9}{4} \\
4 & \tfrac{9}{4} & 2
\end{bmatrix}\!$}
& $0.3$ & Joint & $0.228\!\,\text{\scriptsize(0.005)}$ & $[0.183\!\,\text{\scriptsize(0.004)},\; 0.309\!\,\text{\scriptsize(0.006)}]$ & $\mathbf{0.84}\!\,\text{\scriptsize(0.061)}$ \\
& & Separate & $0.336\!\,\text{\scriptsize(0.004)}$ & $[0.295\!\,\text{\scriptsize(0.002)},\; 0.376\!\,\text{\scriptsize(0.005)}]$ & $0.66\!\,\text{\scriptsize(0.068)}$ \\

& $0.1$ & Joint & $0.096\!\,\text{\scriptsize(0.001)}$ & $[0.080\!\,\text{\scriptsize(0.001)},\; 0.115\!\,\text{\scriptsize(0.002)}]$ & $\mathbf{0.90}\!\,\text{\scriptsize(0.043)}$ \\
& & Separate & $0.112\!\,\text{\scriptsize(0.001)}$ & $[0.099\!\,\text{\scriptsize(0.001)},\; 0.124\!\,\text{\scriptsize(0.001)}]$ & $0.68\!\,\text{\scriptsize(0.067)}$ \\

\hline

\multirow{2}{*}{$\!\begin{bmatrix}
9 & 0 & 0 \\
0 & 3 & 0 \\
0 & 0 & 2
\end{bmatrix}\!$}
& $0.3$ & Joint & $0.323\!\,\text{\scriptsize(0.006)}$ & $[0.255\!\,\text{\scriptsize(0.005)},\; 0.400\!\,\text{\scriptsize(0.008)}]$ & $\mathbf{0.88}\!\,\text{\scriptsize(0.046)}$ \\
& & Separate & $0.333\!\,\text{\scriptsize(0.003)}$ & $[0.294\!\,\text{\scriptsize(0.001)},\; 0.372\!\,\text{\scriptsize(0.004)}]$ & $0.72\!\,\text{\scriptsize(0.064)}$ \\

& $0.1$ & Joint & $0.124\!\,\text{\scriptsize(0.003)}$ & $[0.094\!\,\text{\scriptsize(0.001)},\; 0.157\!\,\text{\scriptsize(0.005)}]$ & $\mathbf{0.84}\!\,\text{\scriptsize(0.052)}$ \\
& & Separate & $0.112\!\,\text{\scriptsize(0.001)}$ & $[0.098\!\,\text{\scriptsize(0.000)},\; 0.124\!\,\text{\scriptsize(0.001)}]$ & $0.68\!\,\text{\scriptsize(0.067)}$ \\

\hline
\end{tabular}
\end{table}

Overall, when off-diagonal entries of $\bm{\Sigma}$ are nonzero, joint modeling substantially mitigates upward bias and improves coverage, particularly under weak spatial correlation and small sample sizes. When off-diagonal entries of $\bm{\Sigma}$ are zeros, gains are less pronounced in posterior means but persist in uncertainty quantification, with the joint model delivering more reliable interval calibration at the expense of slightly wider credible intervals.
The reported coverage values in Table~\ref{tab:phi_bgp} demonstrate that, in the trivariate Binomial–Gaussian–Poisson setting, the joint model continues to outperform separate models, highlighting the substantial gains in accurately recovering $\phi$ achieved by leveraging cross-dependence through joint modeling.

\end{document}